\documentclass[reqno]{amsart}

\usepackage{amsmath, amsfonts, amssymb, amsthm, amscd}
\usepackage[mathscr]{eucal} 

\usepackage{graphicx}
\usepackage{xcolor}

\usepackage{float}
\usepackage{ytableau}
\usepackage{enumitem}

\usepackage{mathrsfs}

\usepackage{tikz}
\usetikzlibrary{arrows.meta, calc, matrix, positioning}
\tikzset{>=Stealth}

\usepackage[colorlinks=true, pdfstartview=FitV,
  linkcolor=blue, citecolor=blue, urlcolor=blue]{hyperref}

\usepackage{geometry}
\renewcommand{\arraystretch}{1.5}

\newtheorem{theorem}{Theorem}
\newtheorem{lemma}[theorem]{Lemma}
\newtheorem{conjecture}[theorem]{Conjecture}
\newtheorem{proposition}[theorem]{Proposition}
\newtheorem{corollary}[theorem]{Corollary}

\theoremstyle{definition}

\newtheorem{example}[theorem]{Example}
\newtheorem{remark}[theorem]{Remark}

\newcommand{\aaa}{\mathbf{a}}

\newcommand{\kk}{\mathbf{k}}

\newcommand{\bs}{\boldsymbol{\sigma}}
\newcommand{\br}{\boldsymbol{\rho}}
\newcommand{\bt}{\boldsymbol{\tau}}

\DeclareMathOperator{\wt}{wt}

\newcommand{\Z}{{\mathbb Z}}

\newcommand{\C}{{\mathbb C}}

\newcommand{\ok}{{\rm{\bf k}}}
\newcommand{\am}{{\rm\aaa}^{\!-} }
\newcommand{\ap}{{\rm\aaa}^{\!+} }
\newcommand{\FF}{{\mathcal{F}}}
\newcommand{\s}{{\mathbb{S}}}

\newcommand{\LL}{\mathcal{L}}
\newcommand{\BB}{\mathscr{B}}
\newcommand{\TT}{\mathscr{T}}

\newcommand{\KK}{\mathbb{K}}

\newcommand{\qbinom}[2]{\bgroup\renewcommand*{\arraystretch}{1}\begin{bmatrix} #1 \\ #2\end{bmatrix} \egroup}

\begin{document}

\title[Multispecies $t$-PushTASEP]
{Multispecies inhomogeneous $t$-PushTASEP
\\  with general capacity}

\author[Arvind Ayyer]{Arvind Ayyer}
\address{Arvind Ayyer, Department of Mathematics, Indian Institute of Science,
Bangalore 560012, India}
\email{arvind@iisc.ac.in}

\author[Atsuo Kuniba]{Atsuo Kuniba}
\address{Atsuo Kuniba, Graduate School of Arts and Sciences, University of Tokyo, Komaba, Tokyo, 153-8902, Japan}
\email{atsuo.s.kuniba@gmail.com}

\date{\today}

\begin{abstract}
We study an $n$-species $t$-PushTASEP, an integrable long-range stochastic process,
on a one-dimensional periodic lattice with inhomogeneities
$x_1,\ldots,x_L$ and arbitrary capacity $l$ at each lattice site.
The Markov matrix is identified with an alternating sum of commuting transfer
matrices over all fundamental representations of
$U_t(\widehat{sl}_{n+1})$.
Stationary probabilities are expressed in a matrix product form involving
a fusion of quantized corner transfer matrices for the strange five-vertex
model introduced by Okado, Scrimshaw, and the second author.
The resulting partition function, which serves as the normalization factor
of the stationary probabilities, is obtained from the $l=1$ case by a finite
plethystic substitution of length $l$.
\end{abstract}


\maketitle

\section{Introduction}\label{sec:intro}

PushTASEP is a class of totally asymmetric simple exclusion processes (TASEPs) 
of interacting particles on one-dimensional lattice.
Its characteristic feature is a stochastic dynamics in which multiple particles
may move simultaneously over long distances by pushing one another according to
prescribed rules.
By now, various versions of PushTASEPs have been introduced and studied extensively;
see, for example,
\cite{ANP23,AK25,AM23,AMW24,BW22,CP13,P19} and  the references therein.

In this paper we study a PushTASEP formulated as a continuous-time Markov
process on a periodic lattice of length $L$, in which each local state is
given by a length-$l$ row-shaped semistandard tableau with entries in
$\{0,1,\ldots,n\}$, namely,
$\{(i_1,\ldots,i_l)\in \Z^l \mid 0 \le i_1 \le \cdots \le i_l \le n\}$.
Such a local state may be interpreted as an assembly of particles of species
$1,2,\ldots,n$, together with
$\delta_{i_1,0}+\cdots+\delta_{i_l,0}$ vacant slots, within the maximal
capacity $l$ at each site. (In the main text, however, we treat $0$ as a
particle species as well.)
The transition rates depend on the parameter $t$ and on inhomogeneity parameters
$x_1,\ldots,x_L$ attached to each site.
The model is Markovian, and hence physically meaningful, in the parameter region
$t\ge 0$ and $x_1,\ldots,x_L>0$.
We refer to this model, which has $\binom{n+l}{l}$ local states at each site,
as the $n$-species capacity-$l$ $t$-PushTASEP.

A key ingredient underlying our approach is the quantum $R$-matrix
$S^k_{\;\;l}(z)$ acting on $V^k\otimes V_l$ in an appropriate gauge,
where $V^k$ and $V_l$ denote the degree-$k$ antisymmetric tensor representation and
the degree-$l$ symmetric tensor representation of the quantum affine algebra
$U_t(\widehat{sl}_{n+1})$, respectively.
This coupling of antisymmetric and symmetric tensor representations
appears to be new in the context of integrable probability.
In principle, $S^k_{\;\;l}(z)$ can be obtained \cite{KRS81} either by $l$-fold symmetric fusion
of $S^k_{\;\;1}(z)$ or by $k$-fold antisymmetric fusion of $S^1_{\;\;l}(z)$.
In this paper, we adopt a more efficient construction based on three-dimensional integrability
explored in \cite[Chap.~11.3]{K22}.
It directly yields an explicit formula for $S^k_{\;\;l}(z)$ given in \eqref{R1}--\eqref{S2} 
in terms of the 3D $L$-operator \eqref{LLe}.

Based on the $R$-matrix, we construct the commuting family of transfer matrices 
$T^k(z|x_1,\ldots, x_L)$ of the solvable vertex model (cf.~\cite{Bax83}) 
with spectral parameter $z$ and inhomogeneities $x_1,\ldots, x_L$ for $0 \le k \le n+1$.
In the terminoogy of the quantum inverse scattering method \cite{STF80},
it has the auxiliary space $V^k$ and acts on the quantum space $V_l^{\otimes L}$.

Let $H_{n,l}(x_1,\ldots, x_L)$ be the Markov matrix of 
the $n$-species capacity-$l$ $t$-PushTASEP in \eqref{Hdef} 
whose transition rates are given by \eqref{wm2} and also described
combinatorially in Section \ref{ss:cd}.
Our first main result is the following formula (Theorem \ref{th:main1}):
\begin{align}\label{Hc0}
H_{n,l}(x_1,\ldots, x_L) = D^{-1}_{\bf m} 
\left. \sum_{k=0}^{n+1}(-1)^{k-1}
\frac{dT^k(z|x_1,\ldots, x_L)}{dz}\right|_{z=0} - 
\Bigl(\sum_{j=1}^L\frac{1}{x_j}\Bigr)\mathrm{Id}.
\end{align}
Here, $D_{\mathbf m}$ is a scalar factor defined in \eqref{Dm}, determined by 
the particle multiplicity ${\mathbf m}$.
The RHS is an alternating sum of the derivative of the transfer matrices whose auxiliary spaces range 
over all fundamental representations $V^0,\ldots,V^{n+1}$.
The identity \eqref{Hc0} extends a Baxter-type formula for quantum Hamiltonians
(cf.~\cite[eq.~(10.14.20)]{Bax83}) to an inhomogeneous stochastic setting,
generalizing our earlier result for the $l=1$ case \cite{AK25} to arbitrary capacity.

A notable aspect of \eqref{Hc0} is that neither the individual transfer matrix
$T^k(z|x_1,\ldots,x_L)$ nor its derivative at $z=0$ is stochastic in general:
their matrix elements need not be positive and do not satisfy probability conservation.
Nevertheless, the alternating sum in \eqref{Hc0} acts as an inclusion--exclusion mechanism,
retaining admissible particle motions with correct rates while cancelling forbidden channels.

The result \eqref{Hc0} reduces the problem of finding stationary states of the model
to that of constructing a joint eigenstate of the commuting transfer matrices.
Our second main result provides such a construction explicitly: 
the stationary probability of a configuration $(\bs_1,\ldots,\bs_L)$
is expressed in a matrix product form
\begin{align}\label{pA0}
\mathbb{P}(\bs_1,\ldots, \bs_L)
=\mathrm{Tr}\left(A_{\bs_1}(x_1)\cdots A_{\bs_L}(x_L)\right),
\end{align}
up to normalization.
Here the operator $A_{\mathbf i}(z)$ associated with a local state
${\mathbf i}=(i_1,\ldots,i_l)$ is defined as (cf.~\eqref{Abi}, \eqref{Asym})
\begin{align}\label{af}
A_{\mathbf i}(z)
=\sum A_{i'_1}(z)A_{i'_2}(tz)\cdots A_{i'_l}(t^{l-1}z),
\end{align}
where the sum runs over distinct permutations
$(i'_1,\ldots,i'_l)$ of $(i_1,\ldots,i_l)$.
For the basic case $l=1$, the operators $A_0(z),\ldots,A_n(z)$ in \eqref{ctm1} coincide,
up to a minor conventional change, with those introduced in \cite{KOS24}.
They are {\em  corner transfer matrices} (CTMs \cite[Chap.13]{Bax83}) 
of the strange five-vertex model, which are {\em quantized} in the sense that 
the ``Boltzmann weights" take values in the $t$-oscillator algebra \eqref{s5V}.
The construction of the CTMs for higher $l$ in \eqref{af}
may be viewed as a symmetric fusion at the level of matrix product operators.

Let 
$Z_{l,\mathbf m}(x_1,\ldots, x_L;t) = \sum_{(\bs_1,\ldots, \bs_L) 
\in \mathcal{S}({\bf m})} \mathbb{P}(\bs_1,\ldots, \bs_L)$
denote the normalization factor of the stationary probabilities
for capacity $l$, where $\mathcal{S}({\mathbf m})$ is defined in \eqref{Sm}.
This quantity defines a symmetric polynomial in $x_1,\ldots,x_L$
and is commonly referred to as the partition function.
From \eqref{pA0} and \eqref{af}, it admits a simple reduction to the $l=1$ case:
\begin{align}\label{zpr0}
Z_{l, \mathbf m}(x_1,\ldots,x_L;t)
=Z_{1,\mathbf m}\!\left(\frac{1-t^l}{1-t}x_1,\ldots,\frac{1-t^l}{1-t}x_L; t\right),
\end{align}
where the notation $(1-t^l)x_j/(1-t)$ represents the length-$l$ plethystic substitution
$x_j\mapsto x_j,tx_j,\ldots,t^{\,l-1}x_j$ for each $j$.
Consequently, the right-hand side involves $lL$ inhomogeneity parameters 
geometrically weighted as
$x_1,tx_1,\ldots,t^{\,l-1}x_1,\ldots,x_L,tx_L,\ldots,t^{\,l-1}x_L$,
which correspond to a system with $lL$ sites.

For $l=1$, the partition function
$Z_{1,{\mathbf m}}(x_1,\ldots,x_L;t)$
has essentially been identified with a Macdonald polynomial
at $q=1$ \cite{AMW24,CMW22}.
(See also \cite{CDW15} for an earlier result in the context of ASEP.)
More generally, symmetric fusion of $R$-matrices,
or of the associated vertex models,
has been observed to manifest itself as a plethystic substitution
at the level of partition functions \cite{GW20,M25}.
The relation \eqref{zpr0} gives an explicit realization of this phenomenon
for general capacity $l$.

Our model reduces to the one studied in \cite{AMW24} in the special case $l=1$.
For general capacity $l$, the inhomogeneous $n$-species $t$-PushTASEP 
has also been studied in \cite[Sec.~12.5]{BW22} and \cite[Sec.~6]{ANP23}.
Although the formulations adopted in those works appear rather different from the
one presented here, the underlying dynamics is expected to commute with the
Markov and transfer matrices constructed in this paper under periodic boundary
conditions, since both approaches are based on $R$-matrices for
$U_t(\widehat{sl}_{n+1})$ and act on the same state space.
The present paper focuses on the explicit realization of the Markov matrix as 
an alternating sum of transfer matrices and on the matrix product formula 
for the stationary probabilities, aspects that have not been addressed so far.
 
The outline of the paper is as follows.
In Section~\ref{sec:tpush}, we define the $n$-species capacity-$l$ $t$-PushTASEP
and provide a combinatorial description of its transition rates.
In Section~\ref{sec:s}, we present the $R$-matrix $S^k_{\;\;l}(z)$ based on the
three-dimensional approach \cite[Chap.11]{K22}, together with an appropriate gauge choice.
Its matrix elements are described in detail, which will be used in
Section~\ref{sec:HT}.
In Section~\ref{sec:T}, we give a standard construction of the commuting transfer
matrices $T^k(z|x_1,\ldots,x_L)$ from the $R$-matrix $S^k_{\;\;l}(z)$.
In Section~\ref{sec:HT}, we prove Theorem~\ref{th:main1}, which identifies the
Markov matrix introduced in Section~\ref{sec:tpush} with an alternating sum of
the derivatives of the transfer matrices evaluated at $z=0$.
The strategy of the proof parallels that of the $l=1$ case~\cite{AK25}.
In Section~\ref{sec:zf}, we introduce the quantized corner transfer matrices
$A_{\mathbf i}(z)$ and explain their properties, most notably the
Zamolodchikov--Faddeev algebra.
In Section~\ref{sec:mp}, we derive the matrix product formula for the stationary
probabilities, building on the results of the preceding sections.
We also present a few immediate consequences for the partition function.
Appendix~\ref{app:rs} recalls the quantum $R$-matrix $\s_{k,l}(z)$ on symmetric tensor
representations in a gauge adapted to the present paper.

\section{$n$-species $t$-PushTASEP with capacity $l$}\label{sec:tpush}

\subsection{Preliminary}

Given \( n \ge 1 \), define the following vector spaces and index sets:
\begin{align}
V^k &= \bigoplus_{{\bf i} \in \BB^k} \C v^{{\bf i}}, \quad
\BB^k = \{{\bf i} = (i_0,\ldots, i_n) \in \{0,1\}^{n+1} \mid |{\bf i}|=k\}
\quad (0 \le k \le n+1),
\label{vk} \\
V_l &= \bigoplus_{{\bf i} \in \BB_l} \C v_{{\bf i}}, \quad
\BB_l = \{{\bf i} = (i_0,\ldots, i_n) \in (\Z_{\ge 0})^{n+1} \mid |{\bf i}|=l\}
\quad (l \in \Z_{\ge 0}),
\label{vl} \\
|{\bf i}| &= i_0+\cdots + i_n, \quad
{\bf e}_j = (\delta_{j,0},\ldots, \delta_{j,n}) \in \BB^1 =\BB_1 \quad (0 \le j \le n).
\label{ej}
\end{align}
The label \({\bf i} = (i_0,\ldots, i_n)\) of a basis vector is referred to as the \emph{multiplicity representation}.

We further introduce the following sets of integer arrays:
\begin{align}
\TT^k &= \{(I_1,\ldots, I_k) \in \Z^k \mid 0 \le I_1 < \cdots < I_k \le n\} 
\qquad (0 \le k \le n+1),
\label{tk} \\
\TT_l &=  \{(I_1,\ldots, I_l) \in \Z^l \mid 0 \le I_1 \le \cdots \le I_l \le n\}
\qquad (l \in \Z_{\ge 0}).
\label{tl}
\end{align}

Elements of~$\TT^k$ (resp. $\TT_l$) 
are interpreted as semistandard tableaux of column shape 
with depth~$k$ (resp. row shape of width $l$), 
filled with entries from~$\{0,1,\ldots,n\}$.  
We identify $\TT^k$ with~$\BB^k$ (resp. $\TT_l$ with~$\BB_l$) through the bijection
\begin{equation}\label{iic}
i_\alpha = \delta_{\alpha, I_1} + \cdots + \delta_{\alpha, I_k}
\qquad (0 \le \alpha \le n).
\end{equation}

We refer to~\eqref{tk} and \eqref{tl} as the \emph{tableau representation} 
of the basis of~$V^k$ and $V_l$, respectively.
Both representations, the multiplicity representation and the tableau representation, will be used throughout this paper.
The tableau representation aligns naturally with a particle interpretation, 
whereas many formulas are more conveniently expressed in the multiplicity representation.

For two elements $\bs, \bs' \in \BB_l$, we define the relation
\begin{equation}\label{prel}
\bs \prec \bs' \overset{\rm def}{\Longleftrightarrow}
\bs-\bs' = {\bf e}_{r_1}-{\bf e}_{r_2} + \cdots + {\bf e}_{r_{\mu-1}}-{\bf e}_{r_\mu}
\; \text{for some} \; 0 \le r_1<\cdots < r_\mu  \le n\; \text{and} \; \mu\in  2\Z_{\ge 1}.
\end{equation}
The RHS represents the alternating sum $\sum_{i=1}^{\mu}(-1)^{i-1}{\bf e}_{r_i}$.
We write $\bs  \preceq \bs'$ if $\bs \prec \bs'$ or $\bs=\bs'$.
Note that $\bs\preceq \bs'$ and $\bs'\preceq \bs''$ do {\em not} imply $\bs\preceq \bs''$.

We use the notation $\theta(\text{true}) = 1$ and $\theta(\text{false}) = 0$ throughout.

\subsection{Definition of $n$-species capacity-$l$ inhomogeneous $t$-PushTASEP}\label{ss:model}

For $n, l \ge 1$, we introduce the $n$-species, capacity-$l$, inhomogeneous $t$-PushTASEP  
on a one-dimensional periodic lattice of length $L$.  
It is a continuous-time Markov process on the space $\mathbb{V} = V_l^{\otimes L}$,  
where $V_l$ is defined in \eqref{vl}.
We often write a basis vector  
$v_{\bs_1} \otimes \cdots \otimes v_{\bs_L}$ simply as  
$|\bs_1, \ldots, \bs_L\rangle$ or $|\vec{\bs}\rangle$,  
with the shorthand $\vec{\bs} = (\bs_1, \ldots, \bs_L)$
with $\bs_j=(\sigma_{j,0},\ldots, \sigma_{j,n}) \in \BB_l$.
A vector $v_{\bs} \in V_l$, or equivalently the label $\bs \in \BB_l$, is regarded as a \emph{local state}.

In the \emph{multiplicity representation}, a local state
$\bs = (\sigma_0, \ldots, \sigma_n) \in \BB_l$ means that the site contains
$\sigma_\alpha$ particles of species $\alpha$ for $0 \le \alpha \le n$.%
\footnote{We regard $0$ also as a particle, but follow the
conventional terminology and refer to the model as an $n$-species system.}
On the other hand, the \emph{tableau representation}  
$\bs = (I_1, \ldots, I_l) \in \TT_l$ lists the species of the $l$ particles present at the site.
The integer parameter $l$ is referred to as the \emph{capacity}, meaning the maximum number of particles  
that can occupy a single site.

Let $\mathbb{V}({\bf m}) \subset \mathbb{V}$ be the subspace 
specified by the {\em multiplicity} ${\bf m} = (m_0,\ldots, m_n)\in (\Z_{\ge 1})^{n+1}$ 
of the particles as follows:
\begin{align}
\mathbb{V}({\bf m})
&=\bigoplus_{ (\bs_1,\ldots, \bs_L) \in \mathcal{S}({\bf m})}
\mathbb{C}  |\bs_1,\ldots, \bs_L\rangle,
\label{Vm}
\\
\mathcal{S}({\bf m}) &= \{(\bs_1,\ldots, \bs_L) \mid 
\bs_i = (\sigma_{i,0},\ldots, \sigma_{i,n}) \in \BB_l,\;
\sigma_{1,\alpha}+\cdots + \sigma_{L,\alpha}=m_\alpha\, (0 \le  \alpha \le n)\}.
\label{Sm}
\end{align}
Note that $m_0 + \cdots + m_n = Ll$.
We set
\begin{align}
K_\alpha &=\begin{cases}
 l &  \; \text{for}\; \alpha=0,
 \\
 m_0+ \cdots + m_{\alpha-1}  &\; \text{for}\; 1 \le \alpha \le n,
 \end{cases}
 \qquad \KK_\alpha = m_0+ \cdots + m_{\alpha-1}  \; \; \,\text{for }\; 0 \le \alpha \le n,
\label{Ki}\\
D_{\bf m} &= \prod_{\alpha=0}^n(1-t^{K_\alpha}).\label{Dm}
\end{align}
We shall exclusively consider the case $m_0,\ldots, m_n\ge 1$  throughout the article, 
hence $K_0,\ldots, K_n \ge 1$, $\KK_0=0$, $\KK_1,\ldots, \KK_n \ge 1$ and $D_{\bf m} \neq 0$.

The $n$-species capacity-$l$ inhomogeneous $t$-PushTASEP is a stochastic process on each $\mathbb{V}({\bf m})$ 
governed  by the master equation 
\begin{align}
\frac{d}{ds}|\mathbb{P}(s)\rangle = H_{n,l}(x_1,\ldots, x_L)|\mathbb{P}(s)\rangle,
\end{align}
where the state vector is given by 
$|\mathbb{P}(s)\rangle  = \sum_{(\bs_1,\ldots, \bs_L) \in \mathcal{S}({\bf m})}\mathbb{P}(\bs_1,\ldots, \bs_L;s)
|\bs_1,\ldots, \bs_L\rangle$, and 
$\mathbb{P}(\bs_1,\ldots, \bs_L;s)$ denotes the probability that 
the configuration $(\bs_1,\ldots, \bs_L)$ occurs  at time $s$.

The Markov matrix 
$H_{n,l}=H_{n,l}(x_1,\ldots, x_L): \mathbb{V}({\bf m}) \rightarrow \mathbb{V}({\bf m})$ is given by 
\begin{subequations}
\begin{align}
H_{n,l}|\vec{\bs}\rangle &= 
\sum_{\substack{\vec{\bs}' \in \mathcal{S}({\bf m}) \\ \vec{\bs}' \neq \vec{\bs}}}
\sum_{o=1}^L \frac{1}{x_o}\prod_{0 \le h \le n}
w^{(o)}_{\vec{\bs}, \vec{\bs}'}(h)|\vec{\bs}'\rangle
+\Bigl(\sum_{o=1}^L\frac{C_{\vec{\bs}, o}(t)-1}{x_o}\Bigr) |\vec{\bs}\rangle,
\label{Hdef}\\
C_{\vec{\bs},o}(t) &= 
\prod_{0 \le h \le n}\frac{1-t^{\KK_h+\sigma_{o,h}}}{1-t^{K_h}},
\label{Co}
\end{align}
\end{subequations}
where $\vec{\bs} = (\bs_1,\ldots, \bs_L)$ and $\vec{\bs}' = (\bs'_1,\ldots, \bs'_L)$
with 
$\bs_j = (\sigma_{j,0},\ldots, \sigma_{j,n}), 
\bs'_j = (\sigma'_{j,0},\ldots, \sigma'_{j,n}) \in \BB_l$.
Each parameter $x_o>0$ associated with a lattice site $o \in \{1,\ldots, L\}$ 
represents the site-wise inhomogeneity of the system.
The factor $w^{(o)}_{\vec{\bs}, \vec{\bs}'}(h)$ is a rational function of $t$, 
and constitutes the main part of $H_{n,l}$. Its explicit definition will be given below.

Let $\vec{\bs} = (\bs_1,\ldots, \bs_L)$ and $\vec{\bs}' = (\bs'_1,\ldots, \bs'_L) \in \mathcal{S}({\bf m})$
with $\bs_j = (\sigma_{j,0}, \ldots, \sigma_{j,n}), \bs'_j = (\sigma'_{j,0}, \ldots, \sigma'_{j,n}) \in \BB_l$.
A necessary condition\footnote{Sufficiency requires an additional condition $\sigma'_{o,0} \ge 1$, as will be shown in Proposition~\ref{pr:zr}.}
for $w^{(o)}_{\vec{\bs}, \vec{\bs}'}(h) \neq 0$  for $\vec{\bs} \neq \vec{\bs}'$
is that $\vec{\bs}'$ is obtained from $\vec{\bs}$ by a sequence of push-out moves involving particles
$h_g, h_{g-1}, \ldots, h_1, h_0$ in this order, for some 
$0 \le h_0 < \cdots < h_g \le n$ and $1 \le g \le n$.%
\footnote{For simplicity, we refer to a “particle of species~$h$” simply as “particle~$h$”.
Species~$0$ is also regarded as a particle for the purpose of this description.}

\begin{table}[htbp]
\centering
\begin{tabular}{c|rcccl}
particles & $h_g\;$ & $h_{g-1}$ & $\cdots$ & $h_1$ & $h_0$ \\
\hline
departure site & $o = p(h_g)$ & $p(h_{g-1})$ & $\cdots$ & $p(h_1)$ & $p(h_0)$ \\
arrival site   & $p'(h_g)$ & $p'(h_{g-1})$ & $\cdots$ & $p'(h_1)$ & $p'(h_0) = o$
\end{tabular}
\caption{List of moving particles. 
The relations $p'(h_\alpha) = p(h_{\alpha-1})$ $(1 \le \alpha \le g)$ and $p'(0) = p(h_g) = o$ are assumed.}
\label{tab1}
\end{table}

Table~\ref{tab1} illustrates a process in which a particle of species $h_g$
departs from site~$o$ and arrives at site $p'(h_g)=p(h_{g-1})$, thereby pushing out
a smaller particle of species $h_{g-1}$.
The particle $h_{g-1}$ in turn moves to $p'(h_{g-1})=p(h_{g-2})$, pushing out
$h_{g-2}$, and this cascading motion continues until a particle $h_0$ finally
reaches and {\em refills} the original site~$o$.
This refill is necessary because we regard $0$ as a particle species as well,
and each site is required to accommodate exactly $l$ particles.
The departure and arrival sites must satisfy the relations
\begin{align}
p'(h_\alpha) &= p(h_{\alpha-1}) \quad (1 \le \alpha \le g), \quad
p'(h_0) = p(h_g) = o, \\
\sigma_{p(h_\alpha), h_\alpha} &\ge 1, \quad
p(h_\alpha) \ne p'(h_\alpha) \quad (0 \le \alpha \le g). \label{spp}
\end{align}
These conditions do not exclude the possibility $p(h_\alpha) = p'(h_\beta)$ for $\alpha > \beta \ge 0$.
In particular, a single site may experience multiple (up to $l$) push-out events during a process.
See Figure \ref{fig:1}.

 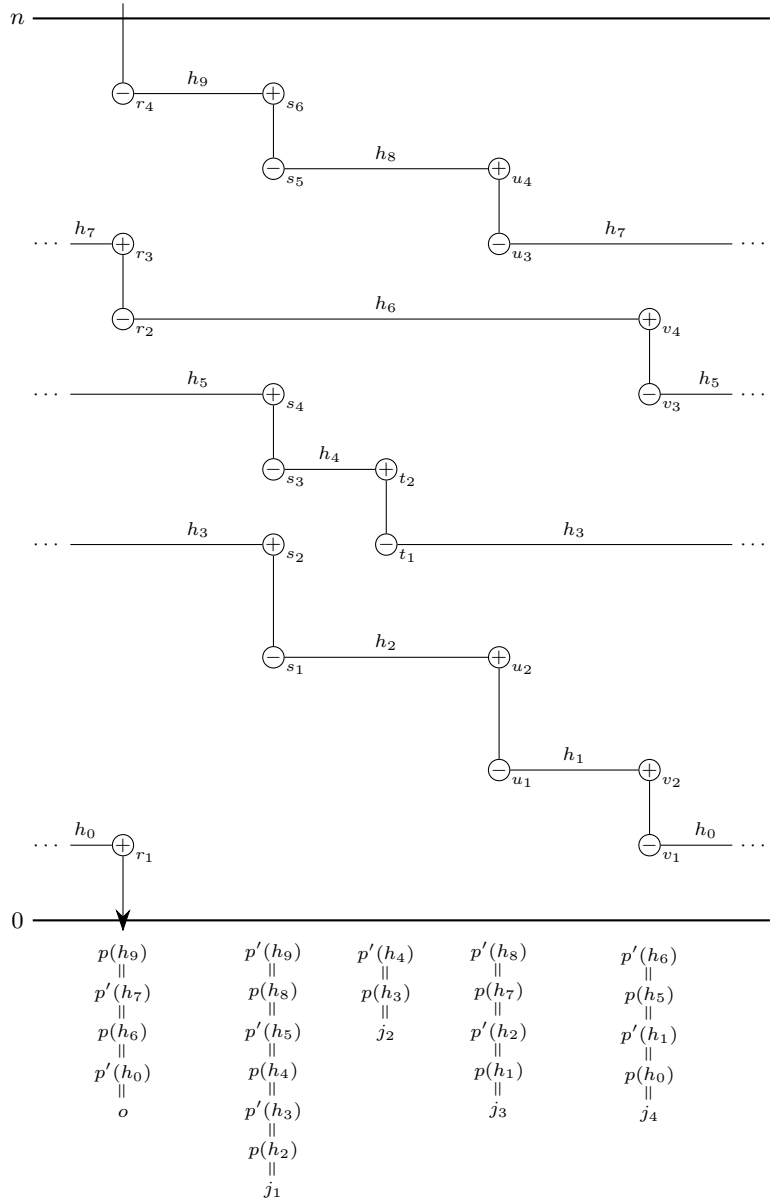
\begin{figure}[H]
  \centering
  \begin{tikzpicture}[>=stealth, font=\small]

    \node (v1)  [draw, circle, inner sep=0pt] at (0,10) {\scriptsize$-$};
    \node [anchor=base west, xshift=-1.8pt, yshift=-2.3pt] at (v1.south east) {\tiny$r_4$};

    \node (v2)  [draw, circle, inner sep=0pt] at (2,10) {\scriptsize$+$};
    \node [anchor=base west, xshift=-1.8pt, yshift=-2.3pt] at (v2.south east) {\tiny$s_6$};

    \node (v3)  [draw, circle, inner sep=0pt] at (2,9) {\scriptsize$-$};
    \node [anchor=base west, xshift=-1.8pt, yshift=-2.3pt] at (v3.south east) {\tiny$s_5$};

    \node (v4)  [draw, circle, inner sep=0pt] at (5,9) {\scriptsize$+$};
    \node [anchor=base west, xshift=-1.8pt, yshift=-2.3pt] at (v4.south east) {\tiny$u_4$};

    \node (v5)  [draw, circle, inner sep=0pt] at (5,8) {\scriptsize$-$};
    \node [anchor=base west, xshift=-1.8pt, yshift=-2.3pt] at (v5.south east) {\tiny$u_3$};

    \node (v6)  [draw, circle, inner sep=0pt] at (0,8) {\scriptsize$+$};
    \node [anchor=base west, xshift=-1.8pt, yshift=-2.3pt] at (v6.south east) {\tiny$r_3$};

    \node (v7)  [draw, circle, inner sep=0pt] at (0,7) {\scriptsize$-$};
    \node [anchor=base west, xshift=-1.8pt, yshift=-2.3pt] at (v7.south east) {\tiny$r_2$};

    \node (v8)  [draw, circle, inner sep=0pt] at (7,7) {\scriptsize$+$};
    \node [anchor=base west, xshift=-1.8pt, yshift=-2.3pt] at (v8.south east) {\tiny$v_4$};

    \node (v9)  [draw, circle, inner sep=0pt] at (7,6) {\scriptsize$-$};
    \node [anchor=base west, xshift=-1.8pt, yshift=-2.3pt] at (v9.south east) {\tiny$v_3$};

    \node (v10) [draw, circle, inner sep=0pt] at (2,6) {\scriptsize$+$};
    \node [anchor=base west, xshift=-1.8pt, yshift=-2.3pt] at (v10.south east) {\tiny$s_4$};

    \node (v11) [draw, circle, inner sep=0pt] at (2,5) {\scriptsize$-$};
    \node [anchor=base west, xshift=-1.8pt, yshift=-2.3pt] at (v11.south east) {\tiny$s_3$};

    \node (v12) [draw, circle, inner sep=0pt] at (3.5,5) {\scriptsize$+$};
    \node [anchor=base west, xshift=-1.8pt, yshift=-2.3pt] at (v12.south east) {\tiny$t_2$};

    \node (v13) [draw, circle, inner sep=0pt] at (3.5,4) {\scriptsize$-$};
    \node [anchor=base west, xshift=-1.8pt, yshift=-2.3pt] at (v13.south east) {\tiny$t_1$};

    \node (v14) [draw, circle, inner sep=0pt] at (2,4) {\scriptsize$+$};
    \node [anchor=base west, xshift=-1.8pt, yshift=-2.3pt] at (v14.south east) {\tiny$s_2$};

    \node (v15) [draw, circle, inner sep=0pt] at (2,2.5) {\scriptsize$-$};
    \node [anchor=base west, xshift=-1.8pt, yshift=-2.3pt] at (v15.south east) {\tiny$s_1$};

    \node (v16) [draw, circle, inner sep=0pt] at (5,2.5) {\scriptsize$+$};
    \node [anchor=base west, xshift=-1.8pt, yshift=-2.3pt] at (v16.south east) {\tiny$u_2$};

    \node (v17) [draw, circle, inner sep=0pt] at (5,1) {\scriptsize$-$};
    \node [anchor=base west, xshift=-1.8pt, yshift=-2.3pt] at (v17.south east) {\tiny$u_1$};

    \node (v18) [draw, circle, inner sep=0pt] at (7,1) {\scriptsize$+$};
    \node [anchor=base west, xshift=-1.8pt, yshift=-2.3pt] at (v18.south east) {\tiny$v_2$};

    \node (v19) [draw, circle, inner sep=0pt] at (7,0) {\scriptsize$-$};
    \node [anchor=base west, xshift=-1.8pt, yshift=-2.3pt] at (v19.south east) {\tiny$v_1$};

    \node (v20) [draw, circle, inner sep=0pt] at (0,0) {\scriptsize$+$};
    \node [anchor=base west, xshift=-1.8pt, yshift=-2.3pt] at (v20.south east) {\tiny$r_1$};

    \draw (0,11.2)--(v1);
    \draw (v1) -- (v2) -- (v3) -- (v4) -- (v5);
    \draw (v6) -- (v7) -- (v8) -- (v9);
    \draw (v10) -- (v11) -- (v12) -- (v13);
    \draw (v14) -- (v15) -- (v16) -- (v17) -- (v18) -- (v19);
\draw[-{Stealth[length=3mm, width=2mm]}] (v20)--(0,-1.15);
    
\node at ($(v1)!0.5!(v2)+(0,0.2)$) {\scriptsize$h_9$};
\node at ($(v3)!0.5!(v4)+(0,0.2)$) {\scriptsize$h_8$};
\node at (-0.5,8.2) {\scriptsize$h_7$};
\node at ($(v5)!0.5!(8.1,8)+(0,0.2)$) {\scriptsize$h_7$};
\node at ($(v7)!0.5!(v8)+(0,0.2)$) {\scriptsize$h_6$};
\node at (1,6.2) {\scriptsize$h_5$};
\node at ($(v9)!0.5!(8.6,6)+(0,0.2)$) {\scriptsize$h_5$};
\node at ($(v11)!0.5!(v12)+(0,0.2)$) {\scriptsize$h_4$};
\node at (1,4.2) {\scriptsize$h_3$};
\node at ($(v13)!0.5!(8.5,4)+(0,0.2)$) {\scriptsize$h_3$};
\node at ($(v15)!0.5!(v16)+(0,0.2)$) {\scriptsize$h_2$};
\node at ($(v17)!0.5!(v18)+(0,0.2)$) {\scriptsize$h_1$};
\node at (-0.5,0.2) {\scriptsize$h_0$};
\node at ($(v19)!0.5!(8.5,0)+(0,0.2)$) {\scriptsize$h_0$};

    \draw (v5) -- (8.1,8);
    \draw (v9) -- (8.1,6);
    \draw (v13) -- (8.1,4);
    \draw (v19) -- (8.1,0);

    \node at (8.4,8)  {\scriptsize$\cdots$};
    \node at (8.4,6)  {\scriptsize$\cdots$};
    \node at (8.4,4)  {\scriptsize$\cdots$};
    \node at (8.4,0)  {\scriptsize$\cdots$};

    \draw (v6)  -- (-0.7,8);
    \draw (v10) -- (-0.7,6);
    \draw (v14) -- (-0.7,4);
    \draw (v20) -- (-0.7,0);

    \node at (-1.0,8) {\scriptsize$\cdots$};
    \node at (-1.0,6) {\scriptsize$\cdots$};
    \node at (-1.0,4) {\scriptsize$\cdots$};
    \node at (-1.0,0) {\scriptsize$\cdots$};

\draw[thick] (-1.2,11) -- (8.6,11);
\node at (-1.4,11) {$n$};

\draw[thick] (-1.2,-1) -- (8.6,-1);
\node at (-1.4,-1) {$0$};

\node at (0,-1.-1.5) 
  {\scriptsize$
    \renewcommand{\arraystretch}{0.9}
    \begin{array}{c}
      p(h_9) \\[-0.2pt]
      \rotatebox{90}{=} \\[-0.7pt]
      p'(h_7) \\[-0.2pt]
      \rotatebox{90}{=} \\[-0.7pt]
      p(h_6) \\[-0.2pt]
      \rotatebox{90}{=} \\[-0.7pt]
      p'(h_0) \\[-0.2pt]
      \rotatebox{90}{=} \\[-0.7pt]
      o
    \end{array}
  $};
  
\node at (2,-1.-2) 
  {\scriptsize$
    \renewcommand{\arraystretch}{0.9}
    \begin{array}{c}
      p'(h_9) \\[-0.2pt]
      \rotatebox{90}{=} \\[-0.7pt]
      p(h_8) \\[-0.2pt]
      \rotatebox{90}{=} \\[-0.7pt]
      p'(h_5) \\[-0.2pt]
      \rotatebox{90}{=} \\[-0.7pt]
      p(h_4) \\[-0.2pt]
      \rotatebox{90}{=} \\[-0.7pt]
      p'(h_3) \\[-0.2pt]
      \rotatebox{90}{=} \\[-0.7pt]
      p(h_2)  \\[-0.2pt]
      \rotatebox{90}{=} \\[-0.7pt]
      j_1
    \end{array}
  $};
  
\node at (3.5,-1.97) 
  {\scriptsize$
    \renewcommand{\arraystretch}{0.9}
    \begin{array}{c}
      p'(h_4) \\[-0.2pt]
      \rotatebox{90}{=} \\[-0.7pt]
p(h_3)\\[-0.2pt]
      \rotatebox{90}{=} \\[-0.7pt]
      j_2
    \end{array}
  $};
  
  \node at (5,-2.47) 
  {\scriptsize$
    \renewcommand{\arraystretch}{0.9}
    \begin{array}{c}
      p'(h_8) \\[-0.2pt]
      \rotatebox{90}{=} \\[-0.7pt]
      p(h_7) \\[-0.2pt]
      \rotatebox{90}{=} \\[-0.7pt]
      p'(h_2) \\[-0.2pt]
      \rotatebox{90}{=} \\[-0.7pt]
p(h_1)\\[-0.2pt]
      \rotatebox{90}{=} \\[-0.7pt]
      j_3
      \end{array}
  $};
  
    \node at (7,-2.52) 
  {\scriptsize$
    \renewcommand{\arraystretch}{0.9}
    \begin{array}{c}
      p'(h_6) \\[-0.2pt]
      \rotatebox{90}{=} \\[-0.7pt]
      p(h_5) \\[-0.2pt]
      \rotatebox{90}{=} \\[-0.7pt]
      p'(h_1) \\[-0.2pt]
\rotatebox{90}{=}\\[-0.7pt]
      p(h_0)\\[-0.2pt]
      \rotatebox{90}{=} \\[-0.7pt]
      j_4
      \end{array}
  $};
  
  \end{tikzpicture}
  \caption{Schematic plot of a transition 
$\vec{\bs} \rightarrow \vec{\bs}'$ in Table~\ref{tab1} 
for the case 
$0 \le h_0 < \cdots < h_{g=9} \le n$.
The vertical axis ranges from $0$ to $n$, corresponding to the components of arrays in $\BB_l$ defined in~\eqref{vl}, 
while the lattice sites $1,\ldots, L$ are aligned along the horizontal axis with periodic boundary conditions.
At each point $(j,h) \in \{1,\ldots, L\} \times [0,n]$, 
the symbol $\oplus_h$, $\ominus_h$, or blank is placed according to 
$\sigma'_{j,h} - \sigma_{j,h} = 1$, $-1$, or $0$, respectively, 
where 
$\vec{\bs} = (\bs_1,\ldots, \bs_L)$ with $\bs_j = (\sigma_{j,0},\ldots, \sigma_{j,n}) \in \BB_l$, and 
$\vec{\bs}' = (\bs'_1,\ldots, \bs'_L)$ with $\bs'_j = (\sigma'_{j,0},\ldots, \sigma'_{j,n}) \in \BB_l$ 
in the multiplicity representation.
Only the sites where the local state changes are depicted and labeled by $o, j_1, \ldots, j_4$; 
all other sites are omitted in accordance with the notion of a {\em reduced diagram} 
introduced in Section~\ref{ss:rd}.
For instance, 
$\bs'_{o} - \bs_{o} = {\bf e}_{r_1} - {\bf e}_{r_2} + {\bf e}_{r_3} - {\bf e}_{r_4}$ 
for some $0 \le r_1 < \cdots < r_4 \le n$, 
consistent with~\eqref{rr1}, and 
$\bs'_{j_2} - \bs_{j_2} = -{\bf e}_{t_1} + {\bf e}_{t_2}$ 
for some $0 \le t_1 < t_2 \le n$, as in~\eqref{rr2}.
A segment of the form 
$\ominus_p \overset{h_q}{\rule[0.5ex]{2em}{0.15pt}} \oplus_{p'}$ 
is understood to assume $h_q = p = p'$,
and represents the movement of a particle of species $h_q$ 
from the site with $\ominus_p$ to that with $\oplus_{p'}$,
under the periodic boundary condition.
In column~$o$, a vertical line entering the diagram from the top and an arrow exiting downward are added, 
so that it may be viewed as a path descending through the system, possibly wrapping around it.  
This viewpoint will be useful for the combinatorial description of the transition rates in Section~\ref{ss:cd}.
}
  \label{fig:1}
\end{figure}

Let $\{\bar{h}_1, \ldots, \bar{h}_{n-g}\} = \{0, \ldots, n\} \setminus \{h_0, \ldots, h_g\}$; that is,
\begin{align} \label{hhbar}
\{0, \ldots, n\} = \{h_0, \ldots, h_g\} \sqcup \{\bar{h}_1, \ldots, \bar{h}_{n-g}\}.
\end{align}
All particles of species $\bar{h}_1, \ldots, \bar{h}_{n-g}$ remain within their original sites.

\medskip

The above condition $w^{(o)}_{\vec{\bs}, \vec{\bs}'}(h) \ne 0$ can be equivalently stated as follows:
\begin{enumerate}[label=(\roman*)]

\item $o$ is the unique site such that $\bs'_o \prec \bs_o$ in the sense of \eqref{prel}.
For all other sites $j \ne o$, one has $\bs_j \preceq \bs'_j$.

\item Let $\bs_j = (\sigma_{j,0}, \ldots, \sigma_{j,n})$ and $\bs'_j = (\sigma'_{j,0}, \ldots, \sigma'_{j,n})$.
\begin{enumerate}[label=(\roman{enumi}-\roman*)]

\item For each $h \in \{h_0, \ldots, h_g\}$, there exists exactly one site $p(h)$ such that
$\sigma'_{p(h), h} = \sigma_{p(h), h} - 1$, and exactly one site $p'(h)$ such that
$\sigma'_{p'(h), h} = \sigma_{p'(h), h} + 1$.

\item For each $h \in \{\bar{h}_1, \ldots, \bar{h}_{n-g}\}$, one has
$\sigma_{j,h} = \sigma'_{j,h}$ for all $1 \le j \le L$.

\end{enumerate}
\end{enumerate}

We define the factor $w^{(o)}_{\vec{\bs}, \vec{\bs}'}(h)$ in \eqref{Hdef} as follows:
\begin{align}
w^{(o)}_{\vec{\bs}, \vec{\bs}'}(h) &= 
\begin{cases}
\displaystyle 
\frac{(1 - t^{\sigma_{p(h),h}})\, t^{\ell_h}}{1 - t^{K_h}} & \text{if } h \in \{h_0, \ldots, h_g\}, \\[0.7em]
\displaystyle 
\frac{1 - t^{\KK_h + \Phi_h}}{1 - t^{K_h}} & \text{if } h \in \{\bar{h}_1, \ldots, \bar{h}_{n-g}\}.
\end{cases}
\label{wm1}
\end{align}

Here, $\ell_h$ is defined by
\begin{equation}\label{lh}
\begin{split}
\ell_h &= \text{number of particles in $\vec{\bs}$ with species in $[0,h)$ in } \\
&\quad  \; \text{the cyclic (clockwise) interval $[p(h), p'(h))$ including $p(h)$ but excluding $p'(h)$}.
\end{split}
\end{equation}
For $l = 1$, this reduces to the number of particles with species in $[0,h)$ 
within the cyclic (clockwise) interval $(p(h), p'(h))$, 
excluding both endpoints, in agreement with $\ell_h$ in \cite[eq.(2.7)]{AK25}.

The quantity $\Phi_h = \Phi_h(\vec{\bs}, \vec{\bs}')$ is given by
\begin{align}\label{wdef}
\Phi_h = \varphi'_h(\bs_o, \bs'_o) 
+ \sum_{\substack{1 \le j \le L \\ j \neq o}} \varphi_h(\bs_j, \bs'_j).
\end{align}
To define $\varphi'_h(\bs_o, \bs'_o)$, 
we use the assumption $\bs'_o \prec \bs_o$ in (i) and recall the definition \eqref{prel}. Then:
\begin{align}
\varphi'_h(\bs_o, \bs'_o) = 
\begin{cases}
\sigma_{o,h} & \text{if } h \in [0, r_1) \sqcup (r_2, r_3) \sqcup \cdots \sqcup (r_{\mu-2}, r_{\mu-1}) \sqcup (r_\mu, n], \\
0 & \text{otherwise},
\end{cases}
\label{phip}
\end{align}
where the numbers $0 \le r_1 < r_2 < \cdots < r_\mu \le n$ (with $\mu\in  2\Z_{\ge 1}$) are defined by
\begin{align}\label{rr1}
\bs'_o - \bs_o = {\bf e}_{r_1} - {\bf e}_{r_2} + \cdots + {\bf e}_{r_{\mu - 1}} - {\bf e}_{r_\mu}.
\end{align}

The function $\varphi_h(\bs_j, \bs'_j)$ should be 
defined for $j \ne o$ and $\bs_j \preceq \bs'_j$ in the light of the condition in (i).
We set
\begin{align}
\varphi_h(\bs_j, \bs'_j) = 
\begin{cases}
\sigma_{j,h} & \text{if } \bs_j \prec \bs'_j \text{ and } 
h \in (r_1, r_2) \sqcup (r_3, r_4) \sqcup \cdots \sqcup (r_{\mu - 1}, r_\mu), \\
0 & \text{otherwise},
\end{cases}
\label{phi}
\end{align}
where, for $\bs_j \prec \bs'_j$, the sequence $0 \le r_1 < r_2 < \cdots < r_\mu \le n$ (with $\mu\in  2\Z_{\ge 1}$) 
is defined via
\begin{align}\label{rr2}
\bs_j - \bs'_j = {\bf e}_{r_1} - {\bf e}_{r_2} + \cdots + {\bf e}_{r_{\mu - 1}} - {\bf e}_{r_\mu}.
\end{align}

Observe that 
$\varphi_0(\bs_j, \bs'_j) = 0$ for all $j \ne o$, since
$0 \not\in (r_1, r_2) \sqcup (r_3, r_4) \sqcup \cdots \sqcup (r_{\mu - 1}, r_\mu)$
in \eqref{phi} for any sequence $0 \le r_1 < \cdots < r_\mu \le n$.
Consequently,  when $h=0$, the definition\eqref{wdef} simplifies to 
\begin{align}\label{P0}
\Phi_0 = \varphi'_0(\bs_o, \bs'_o).
\end{align}

From \eqref{wm1},  the main factor in the transition rate in \eqref{Hdef} becomes
\begin{align}
\prod_{h = 0}^n w^{(o)}_{\vec{\bs}, \vec{\bs}'}(h) =
\prod_{h \in \{h_0, \ldots, h_g\}}
\frac{(1 - t^{\sigma_{p(h),h}})\, t^{\ell_h}}{1 - t^{K_h}} 
\prod_{h \in \{\bar{h}_1, \ldots, \bar{h}_{n-g}\}}
\frac{1 - t^{\KK_h + \Phi_h}}{1 - t^{K_h}}.
\label{wm2}
\end{align}

\begin{proposition}\label{pr:zr}
In the process $\vec{\bs} \rightarrow \vec{\bs}'$ described in Table \ref{tab1},
the product \eqref{wm2} with generic $t$ vanishes if and only if the following condition is satisfied:
\begin{align}\label{so}
\sigma'_{o,0} = 0.
\end{align}
\end{proposition}
\begin{proof}
Note that $\sigma_{p(h),h}\ge 1$ by \eqref{spp}, $K_0,\ldots, K_n \ge 1$ by  \eqref{Ki},
and $\Phi_0,\ldots, \Phi_n \ge 0$ by \eqref{wdef}--\eqref{phi}.
Therefore, the only factor in \eqref{wm2}  that can vanish is 
$(1-t^{\Phi_0})$, which occurs when $0 \in \{\bar{h}_1, \ldots, \bar{h}_{n-g}\}$.

Suppose $0 \in \{\bar{h}_1, \ldots, \bar{h}_{n-g}\}$, i.e., particles of species $0$ do not move.
This implies $r_1 \ge 1$ in \eqref{rr1}, hence $0 \in [0, r_1)$ in \eqref{phip}.
Therefore, we obtain
$\Phi_0 \overset{\eqref{P0}}{=} \varphi'_0(\bs_o, \bs'_o) = \sigma_{o,0} = \sigma'_{o,0}$.
Thus, if \eqref{wm2} vanishes, it must be that $\sigma'_{o,0} = 0$, proving the ``only if" part.

Conversely, suppose \eqref{so} holds. Then $\sigma_{o,0} = 0$ must also hold, since otherwise 
$\sigma_{o,0} \ge 1$ would contradict \eqref{rr1}.
The fact that $\sigma'_{o,0} = \sigma_{o,0} = 0$ implies that species $0$ particles either do not move at all, 
or , if they do, their motion must be confined to sites other than $o$.
Can there exist such a particle of species $0$ that moves from a site $\alpha$ to 
a different site $\beta$ within $\{1, 2, \ldots, L\} \setminus \{o\}$?
The answer is no, since such a move would imply $\sigma_{\beta,0} < \sigma'_{\beta,0}$, which contradicts the inequality 
$\bs_\beta \prec \bs'_\beta$ in condition (i) of the assumption (see \eqref{prel}).
It follows that species $0$ particles do not move, so $0 \in \{\bar{h}_1, \ldots, \bar{h}_{n-g}\}$.
Thus, \eqref{wm2} contains the vanishing factor
$(1 - t^{\Phi_0}) = (1 - t^{\varphi'_0(\bs_o, \bs'_o)}) = (1 - t^{\sigma_{o,0}}) = 0$, 
which completes the proof.
\end{proof}

For a pair of distinct configurations $\vec{\bs} \neq \vec{\bs}' \in \mathcal{S}({\bf m})$,
we say that $\vec{\bs}'$ is {\em admissible} to $\vec{\bs}$
if the conditions in Table 1, or equivalently (i) and (ii), are satisfied, and in addition, 
$\sigma'_{o,0} \ge 1$ holds.
In view of Proposition \ref{pr:zr},  the expansion of 
$H_{n,l}| \vec{\bs}\rangle$ contains a term $|\vec{\bs}'\rangle$ with nonzero coefficient
if and only if $\vec{\bs}'$ is admissible to $\vec{\bs}$.

\begin{remark}\label{re:l=1}
Let us consider the special case $l = 1$ of \eqref{wm2}.
As for the site $o$, the relation \eqref{rr1} holds for 
$\mu = 2$, and $r_1$ and  $r_2$ are the final and the initial local states at site $o$, respectively.
Since $\sigma_{o,h} = 0$ for $h \in [0,r_1) \sqcup (r_2, n]$,  we have $\varphi'_h(\bs_o, \bs'_o) = 0$.
The relation \eqref{rr2} holds for 
$\mu = 2$, and $r_1$ and $r_2$ are the initial and the final local states at site $j (\neq o)$, respectively.
Since $\sigma_{j,h} = 0$ for $h \in (r_1, r_2)$,  we have $\varphi_h(\bs_j, \bs'_j) = 0$ for $j \neq o$.
It follows that 
\begin{align}\label{Ph0}
\Phi_0 = \cdots = \Phi_n = 0 \; \,\text{for $l = 1$}.
\end{align}

Suppose that $0 \in \{h_0, \ldots, h_g\}$.
Then, from the assumption $0 \le h_0 < \cdots < h_g \le n$, it follows that $h_0 = 0$.
Hence, the second product in \eqref{wm2} becomes 1 due to \eqref{Ph0}.
We also have $\sigma_{p(h), h} = 1$ by \eqref{spp}, $K_0 = 1$ from \eqref{Ki}, and $\ell_0 = 0$ from \eqref{lh}.
Thus, \eqref{wm2} reduces to
\begin{align}
\prod_{h = 0}^n w^{(o)}_{\vec{\bs}, \vec{\bs}'}(h) =
\prod_{h \in \{h_1, \ldots, h_g\}} \frac{(1 - t)\, t^{\ell_h}}{1 - t^{K_h}},
\label{wm3}
\end{align}
where $1 \le h_1, \ldots, h_g \le n$ denote the species of the moved particles.
The formula \eqref{wm3} coincides with the known expression given in
\cite[eq.~(2.2)]{AMW24} and \cite[eq.~(2.7)]{AK25}.

On the other hand, suppose that $0 \notin \{h_0, \ldots, h_g\}$, so that
$0 \in \{\bar{h}_1, \ldots, \bar{h}_{n - g}\}$.
In this case, the second product in \eqref{wm2} contains the factor
$\frac{1 - t^{\Phi_0}}{1 - t^{K_0}}$,
which vanishes due to \eqref{Ph0}.\footnote{This fact can, of course, also be deduced from Proposition \ref{pr:zr}.}
This result is consistent with the constraint implied by the dynamics defined in \cite{AMW24, AK25}.
\end{remark}

\begin{example}\label{ex:wpro}
For $n=4, l=3, L=3$, consider the following states in multiplicity and tableau representations \eqref{tl} 
in $\mathbb{V}({\bf m})$ with ${\bf m}=(2,3,1,2,1)$:
\begin{align*}
|\vec{\bs}\rangle &=|(1,1,1,0,0),(0,2,0,0,1), (1,0,0,2,0)\rangle = |012,114,033\rangle,
\\ 
|\vec{\bs}'\rangle &=|(0,2,0,1,0), (1,1,1,0,0), (1,0,0,1,1)\rangle = |113,012,034\rangle,
\\ 
|\vec{\bs}''\rangle &=|(0,2,1,0,0), (1,1,0,0,1), (1,0,0,2,0)\rangle = |112,014,033\rangle,
\\ 
|\vec{\bs}'''\rangle &=|(1,2,0,0,0), (0,1,1,0,1), (1,0,0,2,0)\rangle = |011,124,033\rangle.
\end{align*}
They correspond to $(K_0,K_1,K_2,K_3,K_4)=(3,2,5,6,8)$.
The Markov matrix acts as  
\begin{equation}
\begin{split}
&H_{4,3}(x_1,x_2,x_3) |\vec{\bs}\rangle 
\\
&= \frac{t^6(1-t)^3(1-t^2)}{x_2(1-t^3)(1-t^5)(1-t^6)(1-t^8)}|\vec{\bs}'\rangle 
+ \frac{t(1-t)(1-t^9)}{x_2(1-t^3)(1-t^8)}|\vec{\bs}''\rangle + 
\frac{t^3(1-t)^2}{x_1(1-t^3)(1-t^5)}|\vec{\bs}'''\rangle \cdots,
\end{split}
\end{equation}
which corresponds to $o=2,2$ and $1$ in \eqref{Hdef}, respectively.
Since $\sigma'_{2,0}=\sigma''_{2,0}=\sigma'''_{1,0}=1$, 
the configurations $\vec{\bs}',  \vec{\bs}'', \vec{\bs}'''$ are all admissible to $\vec{\bs}$
in the sense defined after Proposition \ref{pr:zr}.

Let us derive the above expansion from \eqref{Hdef}, i.e.,
\begin{equation}
\begin{split}
H_{4,3}(x_1,x_2,x_3) |\vec{\bs}\rangle &
= \frac{1}{x_2} \prod_{h=0}^4 w^{(2)}_{\vec{\bs}, \vec{\bs}'} (h)|\vec{\bs}'\rangle 
+ \frac{1}{x_2} \prod_{h=0}^4 w^{(2)}_{\vec{\bs}, \vec{\bs}''} (h)  |\vec{\bs}'' \rangle 
+ \frac{1}{x_1} \prod_{h=0}^4 w^{(1)}_{\vec{\bs}, \vec{\bs}'''} (h)  |\vec{\bs}'''\rangle 
+ \cdots.
\end{split}
\end{equation}

For the first term $|\vec{\bs}'\rangle$, we have $o=2$, $g=4$ and 
\begin{align*}
(h_0,h_1,h_2,h_3,h_4)&=(0,1,2,3,4),\quad (p(h_0), p(h_1), p(h_2), p(h_3), p(h_4)) =(1,2,1,3,2),
\\
(\ell_0,\ell_1,\ell_2, \ell_3, \ell_4) &= (0,1,2,1,2).
\end{align*}
Thus \eqref{wm2} is calculated as
\begin{equation}
\begin{split}
\prod_{h=0}^4
w^{(2)}_{\vec{\bs}, \vec{\bs}'}(h)
&= 
\frac{(1-t^{\sigma_{1,0}})t^{\ell_0}}{1-t^{K_0}}
\frac{(1-t^{\sigma_{2,1}})t^{\ell_1}}{1-t^{K_1}}
\frac{(1-t^{\sigma_{1,2}})t^{\ell_2}}{1-t^{K_2}}
\frac{(1-t^{\sigma_{3,3}})t^{\ell_3}}{1-t^{K_3}}
\frac{(1-t^{\sigma_{2,4}})t^{\ell_4}}{1-t^{K_4}}
\\
&=\frac{1-t}{1-t^3}
\frac{(1-t^2)t}{1-t^2}
\frac{(1-t)t^2}{1-t^5}
\frac{(1-t^2)t}{1-t^6}
\frac{(1-t)t^2}{1-t^8}
=\frac{t^6(1-t)^3(1-t^2)}{(1-t^3)(1-t^5)(1-t^6)(1-t^8)}.
\end{split}
\end{equation}

For the second term $|\vec{\bs}''\rangle$, we have $o=2$, $g=1$ and 
\begin{align*}
(h_0, h_1) &= (0,1), \quad (p(h_0), p(h_1)) = (1,2), \quad (\ell_0, \ell_1) = (0,1), 
\quad 
(\Phi_2, \Phi_3, \Phi_4)=(0,0,1).
\end{align*}
In particular, 
$\Phi_4 = \varphi'_4(\bs_2, \bs''_2) + \varphi_4(\bs_1, \bs''_1) + \varphi_4(\bs_3, \bs''_3)=1$ is evaluated as 
\begin{align*}
\varphi_4(\bs_1, \bs''_1) =0 : \;&\bs_1-\bs''_1 = {\bf e}_{r_1} - {\bf e}_{r_2},\quad(r_1,r_2)=(0,1),\quad 
4 \not\in (0,1) \; \, \text{in \eqref{phi}},
\\
\varphi'_4(\bs_2, \bs''_2) =1: \; &\bs''_2-\bs_2 = {\bf e}_{r_1} - {\bf e}_{r_2},\quad (r_1,r_2)=(0,1), \quad
4 \in [0,0) \sqcup (1,4], \quad \sigma_{2,4}=1\; \, \text{in \eqref{phip}},
\\
\varphi_4(\bs_3, \bs''_3)=0: \; &\bs_3=\bs''_3.
\end{align*}
Thus \eqref{wm2} is calculated as
\begin{equation}
\begin{split}
\prod_{h=0}^4
w^{(2)}_{\vec{\bs}, \vec{\bs}''}(h)
&= 
\frac{(1-t^{\sigma_{1,0}})t^{\ell_0}}{1-t^{K_0}}
\frac{(1-t^{\sigma_{2,1}})t^{\ell_1}}{1-t^{K_1}}
\frac{1-t^{\KK_2+\Phi_2}}{1-t^{K_2}}
\frac{1-t^{\KK_3+\Phi_3}}{1-t^{K_3}}
\frac{1-t^{\KK_4+\Phi_4}}{1-t^{K_4}}
\\
&= 
\frac{1-t}{1-t^3}\frac{(1-t^2)t}{1-t^2} \frac{1-t^5}{1-t^5}\frac{1-t^6}{1-t^6}\frac{1-t^9}{1-t^8}
= \frac{t(1-t)(1-t^9)}{(1-t^3)(1-t^8)}.
\end{split}
\end{equation}

For the third term $|\vec{\bs}'''\rangle$, we have $o=1$, $g=1$ and 
\begin{align*}
(h_0,h_1)=(1,2), \quad
(p(h_0), p(h_1)) = (2,1), \quad
(\ell_1,\ell_2) =(1,2), \quad
(\Phi_0, \Phi_3, \Phi_4)=(1,0,0).
\end{align*}
In particular, $\Phi_0=\varphi'_0(\bs_1, \bs'''_1)=1$ originates in 
\begin{align*}
\bs'''_1-\bs_1 = {\bf e}_1-{\bf e}_2,\quad (r_1, r_2)=(1,2), \quad 
0 \in [0,1) \sqcup (2,4], \quad
\sigma_{1,0}=1.
\end{align*}
Thus \eqref{wm2} is evaluated as 
\begin{equation}
\begin{split}
\prod_{h=0}^4
w^{(1)}_{\vec{\bs}, \vec{\bs}'''}(h)
&= 
\frac{(1-t^{\sigma_{2,1}})t^{\ell_1}}{1-t^{K_1}}
\frac{(1-t^{\sigma_{1,2}})t^{\ell_2}}{1-t^{K_2}}
\frac{1-t^{\Phi_0}}{1-t^{K_0}}
\frac{1-t^{\KK_3+\Phi_3}}{1-t^{K_3}}
\frac{1-t^{\KK_4+\Phi_4}}{1-t^{K_4}}
\\
&= 
\frac{(1-t^2)t}{1-t^2}\frac{(1-t)t^2}{1-t^5} \frac{1-t}{1-t^3}
= \frac{t^3(1-t)^2}{(1-t^3)(1-t^5)}.
\end{split}
\end{equation}

\end{example}

\subsection{Combinatorial description of the transition rates}\label{ss:cd}

Let $\vec{\bs}$ be a configuration of the $n$-species capacity-$l$ $t$-PushTASEP with particle 
multiplicity given by 
${\bf m} = (m_0, \dots, m_n)$. 
As shown in Figure \ref{fig:1}, we can think of 
$\vec{\bs}=(\bs_1,\ldots, \bs_L)$ with $\bs_j=(\sigma_{j,0},\ldots, \sigma_{j,n})$ 
as a matrix with $n+1$ rows indexed by
$i=0, 1,\dots, n$ from bottom to top and the $L$ columns indexed by $j=1,2,\ldots, L$ from the left to the right,
where $\sigma_{j,i}$ at the position $(i,j)$ denotes the number of particles of species $i$ at site $j$. 

As an example, let $n=4, l=3, L=3, {\bf m} = (2, 3, 1, 2, 1)$. Then a configuration 
$\vec{\bs}$ (or a corresponding state denoted by $|\vec{\bs}\rangle$)  is expressed in a matrix form as
\begin{equation}
\label{smat}
|\vec{\bs}\rangle = |(1,1,1,0,0),(0,2,0,0,1), (1,0,0,2,0)\rangle = 
|012,114,033\rangle = 
\renewcommand{\arraystretch}{0.85}
\begin{pmatrix}
\,0 & 1 & 0 \,\\
\,0 & 0 & 2 \,\\
\,1 & 0 & 0 \,\\
\,1 & 2 & 0 \,\\
\, 1 & 0 & 1\,
\end{pmatrix}
\renewcommand{\arraystretch}{1}.
\end{equation}

Recall from Figure \ref{fig:1} that transitions can be interpreted as paths 
on $\vec{\bs}$ starting from the top row and leaving from the bottom row, 
where we impose periodic boundary conditions horizontally. 
It will be convenient for us here to interpret $\vec{\bs}$ instead as an infinite periodic array of height $n+1$, where
$\sigma_{j+L,i} = \sigma_{j,i}$ for all $0 \leq i \leq n, j \in \mathbb{N}$.
For the above example, we will write 
\[
|\vec{\bs}\rangle = 
\renewcommand{\arraystretch}{0.85}
\begin{pmatrix}
\,0 & 1 & 0 & 0 & 1 & 0 & 0 & 1 & 0 & \dots \\
\,0 & 0 & 2 & 0 & 0 & 2 & 0 & 0 & 2 & \dots \\
\,1 & 0 & 0 & 1 & 0 & 0 & 1 & 0 & 0 & \dots \\
\,1 & 2 & 0 & 1 & 2 & 0 & 1 & 2 & 0 & \dots \\
\,1 & 0 & 1 & 1 & 0 & 1 & 1 & 0 & 1 & \dots 
\end{pmatrix}
\renewcommand{\arraystretch}{1}.
\]
We now consider down-right paths $P$ (drawn in blue) that satisfy the following conditions:

\begin{equation}\label{brule}
\renewcommand{\arraystretch}{1.06}
\begin{array}{l}
\text{$\bullet$ they start and end at the same column modulo $L$ overlaid on $\vec{\bs}$,} \\
\text{$\bullet$ they can turn rightward\footnotemark\ at $(i,j)$ only if $\sigma_{j,i} \ge 1$,} \\
\text{$\bullet$ they can turn downward at $(i,j)$ only if $\sigma_{j,i} \le l-1$,} \\
\text{$\bullet$ they can move horizontally in a single row for at most $L-1$ steps.}
\end{array}
\end{equation}
\footnotetext{``Rightward'' means to the right \emph{in the figure},
not from the viewpoint of the traveler along the path.}

These conditions imply that the set of such paths is finite.
They are in one-to-one correspondence with the possible transitions 
$\vec{\bs} \rightarrow \vec{\bs}'$ by the rule explained in Figure \ref{fig:1}.

We will enclose the entries strictly below this path by a 
red Young diagram $Y$ in French notation.
For example, 
\begin{equation}
\label{bpa}
\begin{tikzpicture}[>=stealth, thick, every node/.style={minimum size=1cm}]
  \matrix[matrix of math nodes, left delimiter={(}, right delimiter={)}, row sep=0.2cm, column sep=0.2cm, nodes={anchor=center}] (m) {
	0 & 1 & 0 & 0 & 1 & 0 & 0 & 1 & 0 & \dots \\
	0 & 0 & 2 & 0 & 0 & 2 & 0 & 0 & 2 & \dots \\
	1 & 0 & 0 & 1 & 0 & 0 & 1 & 0 & 0 & \dots \\
	1 & 2 & 0 & 1 & 2 & 0 & 1 & 2 & 0 & \dots \\
	1 & 0 & 1 & 1 & 0 & 1 & 1 & 0 & 1 & \dots \\
  };

  \draw[blue, ultra thick, ->]
    ([yshift=5pt]m-1-1.north) -- 
    (m-3-1.center) --             
    (m-3-2.center) --             
    (m-4-2.center) --             
    (m-4-4.center) --             
	([yshift=-5pt]m-5-4.south);    
	
  \draw[red, thick, fill=red!20, opacity=0.5]
    ([xshift=-5pt,yshift=5pt]m-4-1.north west) --  
    ([xshift=5pt,yshift=5pt]m-4-1.north east) --   
    ([xshift=5pt,yshift=5pt]m-5-1.north east) --   
    ([xshift=5pt,yshift=5pt]m-5-3.north east) --   
    ([xshift=5pt,yshift=-5pt]m-5-3.south east) --  
    ([xshift=-5pt,yshift=-5pt]m-5-1.south west) -- 
    cycle;

\end{tikzpicture}
\end{equation}
In case $P$ ends up vertically below where it starts, the Young diagram is empty, i.e., $Y = \emptyset$. 
We assign weights to paths as follows. Locally, the path at a given entry $\sigma_{j,i} = a$ 
can look like one of the four configurations shown in Figure \ref{fig:ab}, 
and to each configuration, the weight of that entry is shown. 

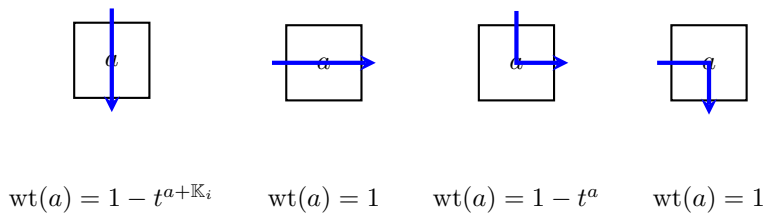
\begin{figure}[h!]
\begin{tikzpicture}[>=stealth, thick]

  \node[draw, minimum size=1cm, inner sep=0pt] (A) at (0,0) {$a$};

  \node[below=of A] (B) {$\wt(a) = 1 - t^{a + \KK_i}$};

  \draw[blue, ultra thick, ->]
    ([yshift=5pt]A.north) -- 
    ([yshift=-5pt]A.south);

\end{tikzpicture}
\quad
\begin{tikzpicture}[>=stealth, thick]

  \node[draw, minimum size=1cm, inner sep=0pt] (A) at (0,0) {$a$};

  \node[below=of A] (B) {$\wt(a) = 1$};

  \draw[blue, ultra thick, ->]
    ([xshift=-5pt]A.west) -- 
    (A.center) -- 
    ([xshift=5pt]A.east);

\end{tikzpicture}
\quad
\begin{tikzpicture}[>=stealth, thick]

  \node[draw, minimum size=1cm, inner sep=0pt] (A) at (0,0) {$a$};

  \node[below=of A] (B) {$\wt(a) = 1 - t^{a}$};

  \draw[blue, ultra thick, ->]
    ([yshift=5pt]A.north) -- 
    (A.center) -- 
    ([xshift=5pt]A.east);

\end{tikzpicture}
\quad
\begin{tikzpicture}[>=stealth, thick]

  \node[draw, minimum size=1cm, inner sep=0pt] (A) at (0,0) {$a$};

  \node[below=of A] (B) {$\wt(a) = 1$};

  \draw[blue, ultra thick, ->]
    ([xshift=-5pt]A.west) -- 
    (A.center) -- 
    ([yshift=-5pt]A.south);

\end{tikzpicture}
\caption{Weights of local configurations of blue paths at the cell containing $a \in [0,l]$.
In the leftmost case, the index 
$i \in [0,n]$ is the row to which the cell belongs, and $\KK_i$ is defined in \eqref{Ki}.
In the third and the fourth cases, $a=0$ and $a=l$ are forbidden, respectively by \eqref{brule}.}
\label{fig:ab}
\end{figure}

Notice that each row of the matrix will only contribute one nontrivial weight.
The weight of the path $P$ is then given by
\begin{equation}\label{wY}
\wt(P) = t^{\sum Y} \prod_{a \in P} \wt(a),
\end{equation}
where $\sum Y$ means the sum of all entries in the Young diagram $Y$ associated to $P$. 
For example, the configuration in \eqref{bpa} has 
$(\KK_0,\KK_1,\KK_2,\KK_3,\KK_4)=(0,2,5,6,8)$ and its weight is 
$t^3 (1-t^8) (1-t^6) (1-t) (1-t^2) (1-t)$.

Suppose that the transition $\vec{\bs} \rightarrow \vec{\bs}'$ corresponds to a path $P$ 
in the sense illustrated in Figure~\ref{fig:1}.  
Then the power $\sum_{h \in \{h_0,\ldots, h_g\}} \ell_h$ of $t$ in 
\eqref{wm2}, \eqref{lh}  is equal to $\sum Y$.
The above definition of $\mathrm{wt}(P)$ provides a restatement of the transition rate \eqref{wm2} as
\begin{align}\label{rwm}
\prod_{h=0}^n w^{(o)}_{\vec{\bs}, \vec{\bs}'}(h)
= \frac{\mathrm{wt}(P)}{D_{\mathbf{m}}},
\end{align}
where $o$ specifies the column from which the path $P$ enters the diagram from the top,
and $D_{\mathbf{m}}$ is defined in~\eqref{Dm}.

In the special case $\vec{\bs}' = \vec{\bs}$, 
the path $P$ runs straight down the diagram, as mentioned earlier, 
yielding the weight 
$\prod_{h=0}^n (1 - t^{\KK_h + \sigma_{o,h}})$, 
which coincides with the numerator of~\eqref{Co}.  
Therefore, the formula~\eqref{Hdef} can be expressed more compactly as
\begin{align}
H_{n,l}|\vec{\bs}\rangle 
&= \sum_{\vec{\bs}' \in \mathcal{S}({\bf m})}
\sum_{o=1}^L \frac{1}{x_o}
\prod_{0 \le h \le n}
w^{(o)}_{\vec{\bs}, \vec{\bs}'}(h)
|\vec{\bs}'\rangle
- \Bigl(\sum_{o=1}^L \frac{1}{x_o}\Bigr)
|\vec{\bs}\rangle,
\label{Hdef2}
\\[2mm]
&= \sum_{o=1}^L \frac{1}{x_o}
\left(
\frac{1}{D_{\bf m}} \sum_P{}^{(o)} \,\mathrm{wt}(P)\,|\vec{\bs}'\rangle
- |\vec{\bs}\rangle
\right),
\label{Hdef3}
\end{align}
where $\sum^{(o)}_P$ denotes the sum running over all paths entering the diagram at column~$o$.

\begin{remark}\label{re:t0}

Our model is well defined at $t=0$.
Indeed, from \eqref{wY}, the blue paths are restricted to those satisfying $\sum Y=0$.
Taking the rules \eqref{brule} into account, the $t=0$ dynamics is expressed as
\begin{align}
\left.H_{n,l}\right|_{t=0}|\vec{\bs}\rangle
&= \sum_{\substack{1 \le j \le L\\ \sigma_{j,0}=0}}
\frac{1}{x_j} \Bigl(|\vec{\bs}(j)\rangle -|\vec{\bs}\rangle\Bigr).
\end{align}
Here $\vec{\bs}(j)$ is the unique state determined from $\vec{\bs}=(\bs_1,\ldots, \bs_L)$
with $\bs_m=(\sigma_{m,0},\ldots, \sigma_{m,n}) \in \BB_l$ \eqref{vl}, 
and the site index $j$ satisfying $\sigma_{j,0}=0$.
It corresponds to the unique blue path that enters column $j$ vertically 
from the top and exits column $j+L$ vertically from the bottom.
In each column, the path proceeds downward until it reaches the lowest cell inscribed with a positive integer, 
where it turns to the right.
If all cells below the entry point in that column are inscribed with $0$, the path simply continues horizontally.
At the final step, the path enters column $j+L$ from the left at the bottom cell inscribed with $\sigma_{j,0}=0$, 
and exits from that position after turning downward.

Interpreting $\sigma_{m,i}$ as the number of particles of species $i\in [0,n]$ at site $m$, the transition
$\vec{\bs} \rightarrow \vec{\bs}(j)$ associated with the above blue path can be described as follows.
A particle of the smallest available species (necessarily $\ge 1$ under the assumption $\sigma_{j,0}=0$) 
is activated at site $j$.
It moves cyclically to the nearest site on its right that contains particles of smaller species,
where it activates and bumps out a smallest such particle, which then continues the same procedure.
This sequential process continues until the motion returns to site $j$ after one full wrap around the lattice,
where the last activated particle, necessarily of species $0$, refills that site.

The $t=0$ dynamics is not irreducible in general.
For example, for $n=l=2$ and $L=3$, there is an invariant subspace spanned by
$|12,02,11\rangle$, $|02,12,11\rangle$, and $|12,12,01\rangle$
(in tableau representation)
within $\mathbb{V}({\mathbf m})$ with ${\mathbf m}=(1,3,2)$
since the $2$'s are untouched ($\dim \mathbb{V}({\mathbf m})=15$).
A general result on irreducibility will be given in Proposition~\ref{pr:ir}.
\end{remark}
 
\subsection{Probability conservation}\label{ss:pc}

Here we show, by a direct manipulation, that the Markov matrix $H_{n,l}$
indeed conserves the total probability.
An alternative proof based on the integrability of the underlying vertex model
will be given around~\eqref{pc}.

Set
\begin{equation}
\wt_j(\vec{\bs}) = \sum_P{}^{(j)} \wt(P).
\end{equation}
In view of \eqref{Hdef3}, the following result ensures the total probability conservation for the Markov matrix \eqref{Hdef3}.
\begin{theorem}
\label{th:pc}
Let $\vec{\bs}$ be a configuration of the $n$-species capacity-$l$ 
$t$-PushTASEP with particle multipicity given by ${\bf m} = (m_0, \dots, m_n)$. Then
$\wt_j(\vec{\bs}) = D_{\bf m}$ for all $1 \le j \le L$.
\end{theorem}

The proof of this result will follow from a slightly more general lemma, which we now state and prove. 

Let $n, L$ be as above. Let $\vec{\bt} = (\tau_{j,i})$ be a periodic array of nonnegative integers, 
where $0 \leq i \leq n$ and $j \in \mathbb{N}$ and the periodicity condition is as before, 
namely $\tau_{j + L,i} = \tau_{j,i}$. 
In addition, \textit{define} $m_i = \sum_{j = 1}^L \tau_{j,i}$ for $0 \leq i \leq n-1$,
${\bf m} = (m_0, \dots, m_n)$ and $l = \sum_{i = 0}^n \tau_{1,i}$.
(The undefined $m_n$ can be arbitrary since $D_{\mathbf m}$ appearing 
in Lemma \ref{le:pc} below does not depend on it.  See \eqref{Ki}--\eqref{Dm}.)
Notice that these conditions mean that there is no constraint on the entries
$\tau_{2,n}, \dots, \tau_{L,n}$. 
We then define the set of paths overlaid on $\vec{\bt}$ and define $\wt_j(\vec{\bt})$ exactly as above.

\begin{lemma}
\label{le:pc}
Let $\vec{\bt}$ be an integer matrix described above. Then
$\wt_1(\vec{\bt}) = D_{\bf m}$. 
\end{lemma}

\begin{proof}
We will perform induction on $n$. For $n = 0$, $\vec{\bt}$ can be written as
$\vec{\bt} = (l, \tau_{2,0}, \dots \tau_{L,0}, l, \tau_{2,0}, \dots \tau_{L,0}, \dots)$. 
Then there is a single path $P$ starting above column $1$ 
and that is forced to leave below column $1$ as well. 
By  Figure \ref{fig:ab}, 
$\wt(P) = 1 - t^l = D_{(m_0)}$, proving the result in this case.

Now suppose the result holds for all configurations with rows $0,\ldots, n-1$, 
and we want to prove the result for a configuration $\vec{\bt}$ with rows indexed  by $0, \dots, n$.
Consider a path $P$ entering above the entry $\tau_{1,n}$. 
There are two possibilities: either it goes straight down, or it turns right at the point.

Suppose first that it goes straight down.
Then, by the induction hypothesis, the weight of any such path is $1 - t^{\tau_{1,n} + K_n}$ times $D_{\bf m'}$, where
${\bf m'} = (m_0, \dots, m_{n-1})$. To be precise, letting $l' = l - \tau_{1,n}$,
the total weight of all such paths is  
\begin{equation}
\label{case1}
(1 - t^{\tau_{1,n} + m_0 + \cdots + m_{n-1}}) (1 - t^{l'}) \prod_{i = 0}^{n-2} (1 - t^{m_0 + \cdots + m_i}).
\end{equation}
Note that we have used the fact that the sum over all such entries in the first column is $l'$.

We observe that if $\tau_{1,n} = 0$, then the path has to necessarily go down at position $(1,n)$. 
In that case, \eqref{case1} is the total weight of all such paths and it is equal to $D_{\bf m}$ since $l' = l$.

Now, suppose the path turns right at position $(1,n)$, where we are assuming
$\tau_{1,n} \geq 1$. 
The weight of that cell is $1 - t^{\tau_{1,n}}$.
By the fourth property of the path, it must turn downward at position $(k,n)$ for some $k \leq L$. 
Therefore, the Young diagram $Y$ necessarily contains the rectangle bounded by the points 
$(1,0), (1,n-1), (k-1,0)$ and $(k-1, n-1)$.
We now apply the induction hypothesis to all paths starting at position $(k,n-1)$. 
Let $l^{(k)} = \tau_{k,0} + \cdots + \tau_{k,n-1}$. 
Notice that the sum of all the entries in the rectangle is $\sum_{j=1}^{k-1} l^{(j)}$.
Then the sum over all such paths gives
\begin{equation}
\label{case2k}
t^{\sum_{j=1}^{k-1} l^{(j)}} (1 - t^{\tau_{1,n}}) (1 - t^{l^{(k)}}) 
\prod_{i = 0}^{n-2} (1 - t^{m_0 + \cdots + m_i}) = 
(1 - t^{\tau_{1,n}}) \prod_{i = 0}^{n-2} (1 - t^{m_0 + \cdots + m_i})
\left(t^{\sum_{j=1}^{k-1} l^{(j)}} - t^{\sum_{j=1}^{k} l^{(j)}} \right).
\end{equation}
We now sum \eqref{case2k} over $k$ ranging from $2$ to $L$, noticing that the last factor gives a telescoping sum to get
\begin{equation}
\label{case2}
(1 - t^{\tau_{1,n}}) \prod_{i = 0}^{n-2} (1 - t^{m_0 + \cdots + m_i})
\left(t^{l^{(1)}} - t^{\sum_{j=1}^{L} l^{(j)}} \right)
= 
(1 - t^{\tau_{1,n}}) \prod_{i = 0}^{n-2} (1 - t^{m_0 + \cdots + m_i})
\left(t^{l'} - t^{m_0 + \cdots + m_{n-1}} \right),
\end{equation}
where the last power of $t$ is the sum over all entries of $\vec{\bt}$ in rows $0$ through $n-1$, 
and we have used $l^{(1)} = l'$.

To complete the proof, we add \eqref{case1} and \eqref{case2} to get
\[
\prod_{i = 0}^{n-2} (1 - t^{m_0 + \cdots + m_i})
\left((1 - t^{\tau_{1,n} + m_0 + \cdots + m_{n-1}}) (1 - t^{l'})
+ (1 - t^{\tau_{1,n}}) (t^{l'} - t^{m_0 + \cdots + m_{n-1}}) \right),
\]
and the sum inside the right parenthesis gives exactly
$(1 - t^{\tau_{1,n} + l'})(1 - t^{m_0 + \cdots + m_{n-1}})$,
which is what we need.
\end{proof}

\begin{proof}[Proof of  Theorem \ref{th:pc}]
The conditions in Lemma \ref{le:pc} work for all $\vec{\bs} \in \mathcal{S}({\bf m})$ and thus, the result goes through.
\end{proof}

\section{$R$-matrix $S(z)$ on (antisymmetric $\otimes$ symmetric)  tensor representation}\label{sec:s}

\subsection{3D construction of the quantum $R$ matrix $R(z)$}

For $a,b,i,j \in \{0,1\}$ and $m \in \Z_{\ge 0}$, define $\LL^{a,b,m'}_{i,j,m}$ by  
\begin{equation}\label{LLe}
\begin{split}
&\LL^{0,0,m'}_{0,0,m} = \LL^{1,1,m'}_{1,1,m} = \delta^{m'}_m,
\quad
\LL^{1,0,m'}_{1,0,m} = (-q)^m\delta^{m'}_m, 
\quad \LL^{0,1,m'}_{0,1,m} = q(-q)^m\delta^{m'}_m,
\\
&\LL^{1,0,m'}_{0,1,m} = \delta^{m'}_{m+1},
\quad
\LL^{0,1,m'}_{1,0,m} = (1-q^{2m})\delta^{m'}_{m-1},
\end{split}
\end{equation}
and the ``weight conservation" property
\begin{align}\label{wc}
\LL^{a,b,m'}_{i,j ,m} =0 \; \, \text{unless} \, \; (a+b, b+m')=(i+j,j+m).
\end{align} 

For $0 \le k \le n+1$ and $l \ge 1$, we introduce the linear map 
$R(z) =R^{k}_{\;\,l}(z,q) \in \mathrm{End}(V^k \otimes V_l)$ via its matrix elements 
in the multiplicity representation as follows:
\begin{subequations}
\begin{align}
R(z)(v^{\bf i} \otimes v_{\bf j}) &=\sum_{{\bf a}\in \BB^k, {\bf b} \in \BB_l}
R(z)^{{\bf a}, {\bf b}}_{{\bf i}, {\bf j}}\,v^{\bf a} \otimes v_{\bf b}
\qquad ({\bf i} \in \BB^k, {\bf j} \in \BB_l),
\label{R1}\\
R(z)^{{\bf a}, {\bf b}}_{{\bf i}, {\bf j}} & = 
\sum_{\alpha_0,\ldots, \alpha_n =0,1}
z^{\alpha_0}
\LL^{\alpha_0,a_n, b_n}_{\alpha_n, i_n,  j_n}
\LL^{\alpha_n,a_{n-1},b_{n-1}}_{\alpha_{n-1}, i_{n-1},j_{n-1}}
 \cdots \LL^{\alpha_1, a_0, b_0}_{\alpha_0, i_0,  j_0}.
\label{R2}
\end{align}
\end{subequations}
The vectors $v^{\bf i}$ and $v_{\bf j}$ have been introduced in \eqref{vk} and \eqref{vl}.
This is so-called trace reduction of the 3D $L$-operator 
over the first component (cf. \cite[eq.(11.44)]{K22}), 
which is known to yield the quantum $R$ matrix on $V^k \otimes V_l$ regarded as 
the tensor product of the degree-$k$ antisymmetric tensor representation and the 
degree-$l$ symmetric tensor representation of $U_{q^2}(\widehat{sl}_{n+1})$ \cite[Th.11.5]{K22}.
As we demonstrate shortly,  
this formulation allows one to compute matrix elements far more efficiently 
than via the conventional fusion procedure.

\begin{example}\label{ex:r11}
Consider the simplest case $k=l=1$.
Let $i, j \in \{0,\ldots, n\}$ and assume $i \neq j$. 
Then the nonzero elements \eqref{R2} are given by
\begin{align*}
R(z)^{{\bf e}_i, {\bf e}_i}_{{\bf e}_i, {\bf e}_i}= z - q^2,
\quad
R(z)^{{\bf e}_i, {\bf e}_j}_{{\bf e}_i, {\bf e}_j} = q(1-z),
\quad
R(z)^{{\bf e}_i, {\bf e}_j}_{{\bf e}_j, {\bf e}_i} = z^{\theta(i<j)}(1-q^2).
\end{align*}
\end{example}

Using $R(z) = R^{k}_{\;\,l}(z,q)$ defined in \eqref{R1}--\eqref{R2}, we introduce 
$S(z) = S^{k}_{\;\,l}(z,t) \in \mathrm{End}(V^k \otimes V_l)$ via its matrix elements 
in the multiplicity representation as follows:
\begin{subequations}
\begin{align}
S(z)(v^{\bf i} \otimes v_{\bf j}) &=\sum_{{\bf a}\in \BB^k, {\bf b} \in \BB_l}
S(z)^{{\bf a}, {\bf b}}_{{\bf i}, {\bf j}}\,v^{\bf a} \otimes v_{\bf b}
\qquad ({\bf i} \in \BB^k, {\bf j} \in \BB_l),
\label{S1}\\
S(z)^{{\bf a}, {\bf b}}_{{\bf i}, {\bf j}} & = \left.
(-1)^{k-1}(-q)^{k(l-1)}(-q)^{\eta^{{\bf a}, {\bf b}}_{{\bf i}, {\bf j}}}z 
R\bigl((-1)^{l-1}q^{k-l}z^{-1},q)^{{\bf a}, {\bf b}}_{{\bf i}, {\bf j}}\right|_{q  \rightarrow t^{1/2}},
\label{S2}
\end{align}
\end{subequations}
where the power $\eta^{{\bf a}, {\bf b}}_{{\bf i}, {\bf j}}$ is given by 
\begin{align}
\eta^{{\bf a}, {\bf b}}_{{\bf i}, {\bf j}}&= \sum_{0 \le r < s \le n}(b_r a_s-i_r j_s).
\label{eta}
\end{align}

From \eqref{wc} and \eqref{R2}, it is evident that $S(z)$ satisfies the weight conservation property:
\begin{align}\label{wcs}
S(z)^{{\bf a}, {\bf b}}_{{\bf i}, {\bf j}} = 0 \quad \text{unless} \quad {\bf a} + {\bf b} = {\bf i} + {\bf j},
\end{align}
where the condition is understood as an equality in $\Z^{n+1}$ in the multiplicity representation.
Moreover, $S(z)$ obeys the same Yang--Baxter equations as those for $R$-matrices acting 
on tensor products of symmetric or antisymmetric tensor representations, 
by the same reasoning as in the proof of \cite[Prop.~4]{KMMO16}.

\begin{example}\label{ex:s11}
For comparison with Example \ref{ex:r11}, we also consider the case  $k=l=1$.
Then, for $0 \le i \neq j \le n$,  the nonzero elements of $S(z)=S^1_{\;\; 1}(z)$ in \eqref{S2} are given by
\begin{align*}
S(z)^{{\bf e}_i, {\bf e}_i}_{{\bf e}_i, {\bf e}_i} = 1-tz,
\quad
S(z)^{{\bf e}_i, {\bf e}_j}_{{\bf e}_i, {\bf e}_j} = t^{\theta(i>j)}(1-z),
\quad
S(z)^{{\bf e}_i, {\bf e}_j}_{{\bf e}_j, {\bf e}_i} = z^{\theta(i>j)}(1-t).
\end{align*}
The element $S(z)^{{\bf e}_a, {\bf e}_b}_{{\bf e}_i, {\bf e}_j}$ here 
coincides with $S(z)^{ab}_{ij}$ in~\cite[Eq.~(3.6)]{AK25} 
and with $(1 - tz)\,R(z)^{n-a,n-b}_{n-i,n-j}$ in~\cite[Eq.~(16)]{KOS24}.
\end{example}

From \eqref{R2} and \eqref{S2}, $S(z)=S^k_{\;\, l}(z)$ is a first-order polynomial in $z$ for general $k$ and $l$.
Thus, its structure is fully determined by specifying the values at $z=0$ and the derivative at $z=0$:
\begin{align}\label{Sd}
S(z)^{{\bf a}, {\bf b}}_{{\bf i}, {\bf j}}= S(0)^{{\bf a}, {\bf b}}_{{\bf i}, {\bf j}} 
+ z \dot{S}(0)^{{\bf a}, {\bf b}}_{{\bf i}, {\bf j}},
\qquad 
\dot{S}(0)^{{\bf a}, {\bf b}}_{{\bf i}, {\bf j}}
= \left.\frac{d}{dz}S(z)^{{\bf a}, {\bf b}}_{{\bf i}, {\bf j}}\right|_{z=0}.
\end{align}
They are explicitly given in  Propositions \ref{pr:s0} below.
In what follows, we use the notation 
\begin{align*}
 i_\Omega = \sum_{u \in \Omega}i_u, 
\quad 
(ja)_\Omega = \sum_{u \in \Omega}j_ua_u,
\end{align*} 
for intervals of the form  $\Omega=(r,s), [r,s), (s,r], [r,s]$, and similarly for 
$a_\Omega,  j_\Omega, (ij)_\Omega$, etc.

\begin{proposition}\label{pr:s0}
Let $0 \le k \le n$ and $l \ge 1$, and suppose that ${\bf a}, {\bf i} \in \BB^k$ and ${\bf b}, {\bf j} \in \BB_l$
are multiplicity representations satisfying ${\bf a} + {\bf b} = {\bf i} + {\bf j}$.
Then,  $S(0)^{{\bf a}, {\bf b}}_{{\bf i}, {\bf j}}$ is given by
\begin{align}
S(0)^{{\bf a}, {\bf b}}_{{\bf i}, {\bf j}} &=\theta({\bf j} \preceq {\bf b}) 
(-1)^\alpha t^\beta(1-t^{j_{r_1}})(1-t^{j_{r_3}})\cdots (1-t^{j_{r_{\mu-1}}}),
\label{Sz0}
\\
\alpha &= i_{(r_1,r_2)} + i_{(r_3,r_4)}+ \cdots + i_{(r_{\mu-1},r_\mu)},
\label{alp0}\\
\beta &= \sum_{0 \le r < s \le n}j_r a_s
+ (ja)_{(r_1,r_2)} + (ja)_{(r_3,r_4)}+ \cdots + (ja)_{(r_{u-1},r_u)},
\label{bet0}
\end{align}
where the product in \eqref{Sz0} range over $(1-t^{j_{r_{\mathrm odd}}})$, and 
the numbers $0 \le r_1 < \cdots < r_\mu \le n$ with even  $\mu$ are determined by 
\begin{align}\label{bj0}
{\bf j}-{\bf b} = {\bf e}_{r_1}-{\bf e}_{r_2}+\cdots + {\bf e}_{r_{\mu-1}}-{\bf e}_{r_\mu}
\end{align}
according to (\ref{prel}). 
In particular, the case ${\bf j}={\bf b}$ corresponds to $\mu=0$,  hence 
$\alpha=0$ and $\beta = \sum_{0 \le r < s \le n}j_r a_s$.

Similarly, the derivative at $z=0$ is given by 
\begin{align}
\dot{S}(0)^{{\bf a}, {\bf b}}_{{\bf i}, {\bf j}}
&=\theta({\bf b} \preceq {\bf j}) 
(-1)^\alpha t^\beta(1-t^{j_{r_2}})(1-t^{j_{r_4}})\cdots (1-t^{j_{r_{\mu}}}),
\label{Sz1}
\\
\alpha &= 1+ i_{[r_1,r_2)} + i_{[r_3,r_4)}+ \cdots + i_{[r_{\mu-1},r_\mu)},
\label{alp1}\\
\beta &= \sum_{0 \le r < s \le n}j_r a_s
+ (ja)_{[0,r_1)} + (ja)_{(r_2,r_3)} + \cdots + (ja)_{(r_{\mu-2},r_{\mu-1})} + (ja)_{(r_\mu, n]},
\label{bet1}
\end{align}
where the product in \eqref{Sz1} range over $(1-t^{j_{r_{\mathrm even}}})$, 
and the numbers $0 \le r_1 < \cdots < r_\mu \le n$ with even $\mu$ are determined by 
\begin{align}\label{bj1}
{\bf b}-{\bf j} = {\bf e}_{r_1}-{\bf e}_{r_2}+\cdots + {\bf e}_{r_{\mu-1}}-{\bf e}_{r_\mu}
\end{align}
according to (\ref{prel}).
In particular, the case ${\bf j}={\bf b}$ corresponds to $\mu=0$,  hence 
$\alpha=l $ and $\beta = \sum_{0 \le r \le s \le n}j_r a_s$.
\end{proposition}

\begin{proof}
Let $R(z)^{{\bf a}, {\bf b}}_{{\bf i}, {\bf j}} = R(0)^{{\bf a}, {\bf b}}_{{\bf i}, {\bf j}} 
+z \dot{R}(0)^{{\bf a}, {\bf b}}_{{\bf i}, {\bf j}}$ be the expansion of 
\eqref{R2} with respect to the spectral parameter $z$ around $z=0$.
Then, from the definition \eqref{S2}, we have
\begin{align}
S(z)^{{\bf a}, {\bf b}}_{{\bf i}, {\bf j}}&=\left.(S_0+ z S_1)\right|_{q\rightarrow t^{1/2}},
\\
S_0 &:=(-1)^{k-l}(-q)^{k(l-1)}(-q)^{\eta^{{\bf a}, {\bf b}}_{{\bf i}, {\bf j}}}
q^{k-l}\dot{R}(0)^{{\bf a}, {\bf b}}_{{\bf i}, {\bf j}},
\\
S_1 &:=(-1)^{k-1}(-q)^{k(l-1)}(-q)^{\eta^{{\bf a}, {\bf b}}_{{\bf i}, {\bf j}}}R(0)^{{\bf a}, {\bf b}}_{{\bf i}, {\bf j}}.
\label{SS1}
\end{align}
Thus, it suffices to show that
$S_0|_{q\rightarrow t^{1/2}}= \eqref{Sz0}$ and 
$S_1|_{q\rightarrow t^{1/2}}= \eqref{Sz1}$. 
We present a derivation of the latter, as the former can be shown quite similarly.

\medskip
First, we calculate $R(0)^{{\bf a}, {\bf b}}_{{\bf i}, {\bf j}}$.
From \eqref{R2},  nonzero contributions come from $\alpha=(\alpha_0, \ldots, \alpha_n) \in \{0,1\}^{n+1}$ with 
$\alpha_0=0$.
Such $\alpha$ are parametrized by $0 \le r_1 < \cdots < r_\mu \le n$ with even $\mu$ as
\begin{align}\label{alu}
\alpha_u=\begin{cases}
0  & \; \text{for }\; u \in [0,r_1] \sqcup (r_2,r_3] \sqcup \cdots \sqcup (r_{\mu-2},  r_{\mu-1}] \sqcup (r_\mu, n],
\\
1 & \; \text{for }\; u \in (r_1, r_2] \sqcup (r_3, r_4] \sqcup \cdots \sqcup (r_{\mu-1}, r_\mu].
\end{cases}
\end{align}
The contribution to \eqref{R2} from the $\alpha$ in \eqref{alu} is
\begin{align}
\prod_{\nu: \mathrm{odd}}\LL^{1,a_{r_\nu},b_{r_\nu}}_{0,i_{r_\nu}, j_{r_\nu}}
\prod_{\nu: \mathrm{even}}\LL^{0,a_{r_\nu},b_{r_\nu}}_{1,i_{r_\nu}, j_{r_\nu}}
\prod_{u \in U_0}\LL^{0,a_u,b_u}_{0,i_u, j_u}
\prod_{u \in U_1}\LL^{1,a_u,b_u}_{1,i_u, j_u}.
\label{L41}
\end{align}
Here and in what follows, we use the index sets $U_0, U_1, U^+_0$ and $U^+_1$ defined by 
\begin{align*}
U_0 &= [0,r_1)\sqcup (r_2, r_3)\sqcup \cdots \sqcup (r_{\mu-2},r_{\mu-1}) \sqcup (r_\mu,n],
\\
U^+_0 &= [0,r_1)\sqcup [r_2, r_3)\sqcup \cdots \sqcup [r_{\mu-2},r_{\mu-1}) \sqcup [r_\mu,n],
\\
U_1 &= (r_1, r_2) \sqcup (r_3,r_4) \sqcup \cdots \sqcup (r_{\mu-1},r_\mu),
\\
U^+_1 &= [r_1, r_2) \sqcup [r_3,r_4) \sqcup \cdots \sqcup [r_{\mu-1},r_\mu).
\end{align*}
The relation
\begin{align}\label{uAll}
U^+_0 \sqcup U^+_1 = [0,n]
\end{align}
will be used, for example, in the identity $i_{U^+_0}+i_{U^+_1}=k$, which is due to 
${\bf i} = (i_0,\ldots, i_n) \in \BB^k$ \eqref{vk}.

From the conservation property \eqref{wc}, which also implies $\LL^{\alpha', a,b}_{\alpha, i,j}=0$ unless 
$b-j = \alpha'-\alpha$, we see that \eqref{L41} is nonvanishing only if 
${\bf b}-{\bf j} = {\bf e}_{r_1}-{\bf e}_{r_2}+\cdots + {\bf e}_{r_{\mu-1}}-{\bf e}_{r_\mu}$.
This verifies the condition \eqref{bj1}.
The conservation property \eqref{wc} further implies the following:
\begin{align}
a_u &= i_u,\; b_u= j_u \;\, \mathrm{for }\;\,  u \not\in \{r_1,\ldots, r_\mu\},\quad
a_{r_\nu}=\begin{cases} 1 & \nu: \;\text{even},\\ 0 & \nu: \; \text{odd}, \end{cases}
\quad
i_{r_\nu}=\begin{cases} 0 & \nu: \;\text{even},\\ 1 & \nu: \; \text{odd}, \end{cases}
\nonumber
\\
a_{(u,n]} &= i_{(u,n]} + \theta(u\in U^+_1).
\label{L70}
\end{align}
We utilize these facts and notations in the following calculation extensively.
Now the expression \eqref{L41} becomes
\begin{align}
&\theta({\bf b} \preceq {\bf j})
\prod_{\nu: \mathrm{odd}}\LL^{1,0,j_{r_\nu}+1}_{0,1, j_{r_\nu}}
\prod_{\nu: \mathrm{even}}\LL^{0,1,j_{r_\nu}-1}_{1,0, j_{r_\nu}}
\prod_{u \in U_0}\LL^{0,i_u,j_u}_{0,i_u, j_u}
\prod_{u \in U_1}\LL^{1,i_u,j_u}_{1,i_u, j_u}
\nonumber\\
&\overset{\eqref{LLe}}{=}
 \theta({\bf b} \preceq {\bf j})
 \prod_{\nu: \mathrm{even}}(1-q^{2j_{r_\nu}})
 \prod_{u \in U_0}q^{i_u}(-q)^{j_ui_u}
 \prod_{u \in U_1}(-q)^{j_u(1-i_u)}
 \nonumber\\
&=
 \theta({\bf b} \preceq {\bf j})
 \prod_{\nu: \mathrm{even}}(1-q^{2j_{r_\nu}})
 \prod_{u \in U^+_0}q^{i_u}(-q)^{j_ui_u}
 \prod_{u \in U^+_1}(-q)^{j_u(1-i_u)}
\label{L415}\\
 &=
  (-1)^{(ij)_{[0,n]}+j_{U^+_1}} q^{(ij)_{U^+_0}-(ij)_{U^+_1}+i_{U^+_0}+j_{U^+_1}}
 \theta({\bf b} \preceq {\bf j})
 \prod_{\nu: \mathrm{even}}(1-q^{2j_{r_\nu}}).
 \label{L42}
\end{align}
Note that the factor 
$\theta({\bf b} \preceq {\bf j})
 \prod_{\nu: \mathrm{even}}(1-q^{2j_{r_\nu}})$ reproduces a part of the expression \eqref{Sz1} upon the 
 substitution $q \rightarrow t^{1/2}$.
 Hence, in what follows, we concentrate on the remaining part, 
 namely the sign factor and the power of $t$.

As for $\eta^{{\bf a}, {\bf b}}_{{\bf i}, {\bf j}}$ in \eqref{SS1}, 
the definition \eqref{eta} leads to 
\begin{align}
\eta^{{\bf a}, {\bf b}}_{{\bf i}, {\bf j}}&=
\sum_{\nu: \mathrm{odd}}(j_{r_\nu}\!+1)a_{(r_{\nu},n]}+
\sum_{\nu: \mathrm{even}}(j_{r_\nu}\!-1)a_{(r_{\nu},n]}
-\sum_{1 \le \nu \le \mu}  i_{[0, r_\nu)}j_{r_\nu}
+ \sum_{u \in [0,n]\setminus\{r_1,\ldots, r_\mu\}}j_u(a_{(u,n]}-i_{[0,u)})
\nonumber\\
&=a_{(r_1,r_2]}+a_{(r_3,r_4]}+\cdots + a_{(r_{\mu-1},r_\mu]}+\sum_{u \in [0,n]}j_u(a_{(u,n]}-i_{[0,u)})
\nonumber\\
&=k-(a_{[0, r_1)}+a_{(r_2,r_3)}+\cdots + a_{(r_\mu,n]})+\sum_{u \in [0,n]}j_u(a_{(u,n]}-k+i_u+i_{(u,n]})
\nonumber\\
&\overset{\eqref{L70}}= k - i_{U_0}-kl+(ij)_{[0,n]}+\sum_{u \in [0,n]}j_u(2a_{(u,n]}-\theta(u \in U^+_1))
\nonumber\\
&=k(1-l) - i_{U^+_0}+(ij)_{[0,n]}-j_{U^+_1}+2\sum_{0 \le r < s \le n}j_r a_s.
\label{L43}
\end{align}
From \eqref{L42} and \eqref{L43},  
the power of $q$ in $S_1$ in \eqref{SS1} is given by 
\begin{align*}
&k(l-1)+ (ij)_{U^+_0}-(ij)_{U^+_1}+i_{U^+_0}+j_{U^+_1}
+k(1-l) - i_{U^+_0}+(ij)_{[0,n]}-j_{U^+_1}+2\sum_{0 \le r < s \le n}j_r a_s
\nonumber\\
&=2(\sum_{0 \le r < s \le n}j_r a_s+(ij)_{U^+_0})=2(\sum_{0 \le r < s \le n}j_r a_s+(ja)_{U_0}).
\end{align*}
This is equal to $2\beta$ for $\beta$ in \eqref{bet1}.
To derive the formula \eqref{alp1} for the sign factor, 
note that the sign is contained via the combination $(-q)$ almost everywhere,
with the exceptions $(-1)^{k-1}$ in \eqref{SS1} and $q^{i_u}$ in \eqref{L415}.
Thus the total contribution to the power of $(-1)$ is equal to the sum of the power of $q$ just evaluated, 
which is even, and the additional contributions $k-1$  and $i_{U^+_0}$ modulo 2.
Thus, it is equal to $k-1+i_{U^+_0}\equiv 1+i_{U^+_1}$ due to a comment made after \eqref{uAll}, which agrees with \eqref{alp1}.
\end{proof}

According to Proposition \ref{pr:s0},  the diagonal elements are given  by 
\begin{align}\label{sdiag}
S^{k}_{\;\,l}(z)^{{\bf a}, {\bf b}}_{{\bf a}, {\bf b}}= 
t^{\sum_{0 \le r < s \le n}b_r a_s}\bigl(1-t^{\sum_{0 \le r \le n} a_rb_r}z\bigr)
\qquad ({\bf a} \in \BB^k, {\bf b} \in \BB_l).
\end{align}

\section{Transfer matrix $T^k(z)$}\label{sec:T}

\subsection{Definition}\label{ss:Tdef}
Recall that $\mathbb{V}=V_l^{\otimes L}$, where $V_l$ is defined in \eqref{vk}. 
Define the transfer matrix 
$T^k(z) = T^k(z|x_1,\ldots, x_L) : \mathbb{V} \longrightarrow  \mathbb{V}$ on the length $L$ periodic lattice by 
\begin{align}
T^k(z)= \mathrm{Tr}_{V^k}\left(S_{0,L} \Bigl( \frac{z}{x_L} \Bigr)
\cdots S_{0,1} \Bigl( \frac{z}{x_1} \Bigr)\right)
\qquad  (0 \le k \le n+1),
\end{align}
where  the index 0  denotes the auxiliary space $V^k$ over which the trace is taken.
The factor $S_{0,j}(z/x_j)$ is the matrix $S^{k}_{\;\,l}(z/x_j)$ defined by 
\eqref{S1}, \eqref{Sd} and Proposition \ref{pr:s0}. 
It acts on $V^k \otimes (\text{$j$ th component of $\mathbb{V}$ from the left})$.
Explicitly, one has
\begin{subequations}
\begin{align}
T^k(z) |\bs_1,\ldots, \bs_L\rangle  &= \sum_{\bs'_1,\ldots, \bs'_L \in \BB_l}
T^k(z)^{\bs'_1,\ldots, \bs'_L}_{\bs_1,\ldots, \bs_L}
|\bs'_1, \ldots, \bs'_L\rangle,
\label{tkdef}\\
T^k(z)^{\bs'_1,\ldots, \bs'_L}_{\bs_1,\ldots, \bs_L}
&= \sum_{{\bf a}_1,\ldots, {\bf a}_L \in \BB^k}
S^{k}_{\;\,l} \Bigl( \frac{z}{x_1} \Bigr)^{{\bf a}_2, \bs'_1}_{{\bf a}_1, \bs_1}
S^{k}_{\;\,l} \Bigl( \frac{z}{x_2} \Bigr)^{{\bf a}_3, \bs'_2}_{{\bf a}_2, \bs_2} 
\cdots 
S^{k}_{\;\,l} \Bigl( \frac{z}{x_L} \Bigr)^{{\bf a}_1, \bs'_L}_{{\bf a}_L, \bs_L}.
\label{tke}
\end{align}
\end{subequations}
We will also use the basis $\langle \vec{\bs}| = \langle \bs_1,\ldots, \bs_L|$
of the dual space of $\mathbb{V}({\mathbf m})$, defined by
$\langle \vec{\bs}|\vec{\bs}'\rangle = \theta(\vec{\bs}=\vec{\bs}')$.
The transfer matrices $T^k(z)$ act on these vectors from the right as
\begin{equation}\label{tkr}
\langle \bs_1, \ldots, \bs_L|T^k(z) = 
 \sum_{\bs'_1,\ldots, \bs'_L \in \BB_l}
 \langle\bs'_1, \ldots, \bs'_L|
T^k(z)^{\bs_1,\ldots, \bs_L}_{\bs'_1,\ldots, \bs'_L},
\end{equation}
so that $(\langle \vec{\bs}'|T^k(z))|\vec{\bs}\rangle = 
\langle \vec{\bs}'|(T^k(z)|\vec{\bs}\rangle)$ holds.
The element \eqref{tke} as $\langle \vec{\bs}'|T^k(z)| \vec{\bs}\rangle$ is depicted as Figure \ref{fig:tk}.

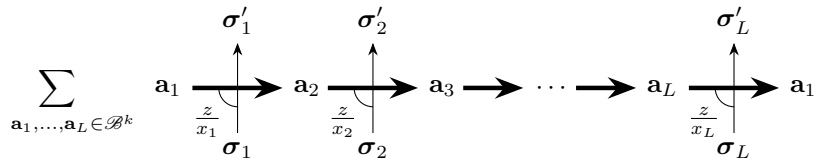
\begin{figure}[H]
\centering
\begin{tikzpicture}

\node at (-1.8,-0.18){$\displaystyle{\sum_{\phantom{AA}{\bf a}_1,\ldots, {\bf a}_L \in \BB^k}}$};
  \draw[->] (0.6,-0.6) node[below]{$\bs_1$}-- (0.6,0.6) node[above]{$\bs'_1$};
  \draw[->,line width=0.6mm] (0,0) node[left]{${\bf a}_1$}-- (1.2,0) node[right]{${\bf a}_2$};
  \draw(0.6,-0.25) arc[start angle=-90, end angle=-180, radius=0.25cm];
  \node at (0.2,-0.45) {$\frac{z}{x_1}$};
\begin{scope}[shift={(1.8,0)}]
  \draw[->] (0.6,-0.6) node[below]{$\bs_2$}-- (0.6,0.6) node[above]{$\bs'_2$};
  \draw[->,line width=0.6mm] (0,0) -- (1.2,0) node[right]{${\bf a}_3$};
  \draw(0.6,-0.25) arc[start angle=-90, end angle=-180, radius=0.25cm];
  \node at (0.2,-0.45) {$\frac{z}{x_2}$};
   \end{scope}

  \draw[->,line width=0.6mm] (3.6,0) -- (4.4,0);
  \node at (4.8,0) {$\cdots$};
  \draw[->,line width=0.6mm] (5.1,0) -- (5.9,0);
   
   \begin{scope}[shift={(6.6,0)}]
     \draw[->] (0.6,-0.6) node[below]{$\bs_L$}-- (0.6,0.6) node[above]{$\bs'_L$};
  \draw[->,line width=0.6mm] (0,0) node[left]{${\bf a}_L$}-- (1.2,0) node[right]{${\bf a}_1$};
  \draw(0.6,-0.25) arc[start angle=-90, end angle=-180, radius=0.25cm];
  \node at (0.2,-0.45) {$\frac{z}{x_L}$};
     \end{scope}
\end{tikzpicture}
\caption{Diagram representation of the matrix element 
$\langle \vec{\bs}'|T^k(z)| \vec{\bs}\rangle$.
Each vertex here gives a diagrammatic representation of 
$S^{k}_{\;\,l} \left( \frac{z}{x_j} \right)^{{\mathbf a}_{j+1}, \bs'_j}_{{\mathbf a}_j, \bs_j}$ 
in \eqref{tke}, where thick and thin arrows represent $V^k$ and $V_l$, respectively.}
\label{fig:tk}
\end{figure}

The parameter $z$ is referred to as the \emph{spectral parameter},
while $x_1,\ldots, x_L$ represent 
the inhomogeneity associated with the vertices.
Following the terminology in the box-ball systems \cite{IKT12},
we refer to ${\bf a}_1, \ldots, {\bf a}_L \in \BB^k$ as {\em carriers} with capacity $k$.

\subsection{Basic properties}

From the Yang-Baxter relation of the matrices $S^{k}_{\;\,l}(z)$, 
the commutativity of the transfer matrices holds:
\begin{align}\label{tcom}
[T^k(z|x_1,\ldots, x_L), T^{k'}(z'|x_1,\ldots, x_L)]=0
\qquad (0 \le k, k' \le n+1).
\end{align}
It is essential that the inhomogeneities $x_1,\ldots, x_L$ in the two transfer matrices are chosen identically.
From the weight conservation property of $S^{k}_{\;\,l}(z)$ and the periodic boundary condition,
$T^k(z)$ preserves each sector $\mathbb{V}({\bf m})$ in \eqref{Vm}.

Let us examine the diagonal elements of $T^k(z)$ for general $k \in \{0,\ldots, n+1\}$.
When $\vec{\bs}' = \vec{\bs}$, 
all the carries ${\bf a}_j$ in Figure \ref{fig:tk} become identical due to the weight conservation.
Set $\vec{\bs}=(\bs_1,\ldots, \bs_L)$ with $\bs_j=(\sigma_{j,0},\ldots, \sigma_{j,n}) \in \BB_l$.
Then, from \eqref{sdiag},  we have
\begin{align}\label{sts}
\langle \vec{\bs}|T^k(z)| \vec{\bs}\rangle
= \sum_{{\bf a} \in \BB^k}\prod_{j=1}^L
S^{k}_{\;\,l}\Bigl( \frac{z}{x_j} \Bigr)^{{\bf a}, \bs_j}_{{\bf a}, \bs_j}
= \sum_{{\bf a} \in \BB^k}\prod_{j=1}^L
 t^{\sum_{0 \le r < s \le n} \sigma_{j,r}a_s}
\left( 1-t^{\sum_{0 \le r \le n}a_r\sigma_{j,r}} \frac{z}{x_j} \right).
\end{align}
In the special cases $k = 0$, $k = n+1$, and $z = 0$, the transfer matrix
$T^k(z)$ becomes diagonal.  
\begin{align}
T^0(z)
 &= \prod_{j=1}^L\Bigl(1 - \frac{z}{x_j}\Bigr)\,\mathrm{Id},
\label{t0f}
\\
T^{n+1}(z)
 &= t^{\KK_1+\cdots+\KK_n}
    \prod_{j=1}^L\Bigl(1 - \frac{t^lz}{x_j}\Bigr)\,\mathrm{Id},
\label{tnf}
\\
T^k(0)
&= e_k(t^{\KK_0}, t^{\KK_1}, \ldots, t^{\KK_n})\mathrm{Id},
\label{tk0}
\end{align}
where $\KK_i$ is determined by \eqref{Ki} from the subspace
$\mathbb{V}({\mathbf m})$ \eqref{Vm} under consideration.
The formulas \eqref{t0f} and \eqref{tnf} are obtained by restricting the
sum in \eqref{sts} to
${\mathbf a} = (0,\ldots,0)$ and ${\mathbf a} = (1,\ldots,1)$,
respectively.  
Note that $T^k(0)$ is diagonal since it is not only weight preserving but
also satisfies the condition
$\bs_j \preceq \bs'_j$ for all sites $j$ owing to \eqref{Sz0}.
The result \eqref{tk0} then follows from \eqref{sdiag}.

The functions $e_0,\ldots,e_{n+1}$ denote the elementary symmetric
polynomials in $(n+1)$ variables, given by
\begin{align}\label{ekb}
e_k(w_0,\ldots,w_n)
 = \sum_{\mathbf a \in \BB^k} w_0^{a_0}\cdots w_n^{a_n},
\end{align}
and characterized by the generating relation
\begin{align}\label{ekg}
(1 + \zeta w_0)\cdots(1 + \zeta w_n)
 = \sum_{k=0}^{n+1}\zeta^k\,e_k(w_0,\ldots,w_n).
\end{align}
We set $e_k(w_0,\ldots,w_n)=0$ for $k<0$ or $k>n+1$.
The formulas \eqref{t0f} and \eqref{tnf} reduce to
\cite[eqs.\,(4.5), (4.6)]{AK25} when $l=1$.

We write  the derivative of the transfer matrices 
simply as $\dot{T}^k(z) = \frac{dT^k(z)}{dz}$.
It is not diagonal in general, but the calculation of the diagonal elements is easy by using \eqref{sts}.
The results read
\begin{align}
\langle \vec{\bs}|\dot{T}^k(0)| \vec{\bs}\rangle
&= -\sum_{j=1}^L \frac{1}{x_j}e_k(t^{\KK_0+ \sigma_{j,0}}, t^{\KK_1+\sigma_{j,1}},\ldots, t^{\KK_n+\sigma_{j,n}}).
\label{tdk0}
\end{align}

\begin{example}\label{ex:Tk}
We set $n=2, l=2, L=3$, and consider the transfer matrices acting on the sector 
$\mathbb{V}({\bf m})$ with ${\bf m}= (3,2,1)$.
Then, from \eqref{Ki}, we have $(\KK_0,\KK_1,\KK_2)=(0,3,5)$.
We adopt the tableau representation in \eqref{tl} for the local states from $\BB_2$ as 
$00=(2,0,0)$, $01=(1,1,0)$, $02=(1,0,1)$, $11=(0,2,0)$ and $12=(0,1,1)$. 
The action of the transfer matrices is given by:
\begin{align*}
T^0(z)|12,00,01\rangle 
&=\frac{(-z+x_1)(-z+x_2)(-z+x_3)}{x_1x_2x_3}|12,00,01\rangle,
\\
T^1(z)|12,00,01\rangle 
&= \frac{(1 - t)^2 t^3 z (-z + x_2)}{x_1 x_2} \, |01, 00, 12\rangle
  + \frac{(1 - t)^2 t^2 z (-z + x_2)}{x_1 x_2} \, |02, 00, 11\rangle \\
&  + \frac{(1 - t)^2 t^4 z (-z + x_2)}{x_1 x_2} \, |11, 00, 02\rangle 
 + \frac{(1 - t^2)(1 - t) z (-t z + x_3)}{x_1 x_3} \, |02, 01, 01\rangle \\
&  + \frac{(1 - t^2)(1 - t) t z (-t z + x_3)}{x_1 x_3} \, |01, 02, 01\rangle 
 + \frac{(1 - t^2)(1 - t)^2 t^2 z^2}{x_1 x_3} \, |11, 02, 00\rangle \\
&+ \frac{(1 - t^2)(1 - t) t z (-t z + x_1)}{x_1 x_3} \, |12, 01, 00\rangle 
 + \mathcal{D}_1(z)|12,00,01\rangle,
\\
T^2(z)|12,00,01\rangle 
&= -\frac{(1 - t)^2 (1 - t^2) t^4 z}{x_1} \, |01, 01, 02\rangle
  + \frac{(1 - t)^2 t^7 z (-z + x_2)}{x_1 x_2} \, |01, 00, 12\rangle \\
&  + \frac{(1 - t)^2 t^8 z (-z + x_2)}{x_1 x_2} \, |02, 00, 11\rangle
+ \frac{(1 - t)^2 t^4 z (-t^2 z + x_2)}{x_1 x_2} \, |11, 00, 02\rangle \\
 & + \frac{(1 - t^2)(1 - t) t^6 z (-t^2 z + x_1)}{x_1 x_3} \, |12, 01, 00\rangle
+  \frac{(1 - t^2)(1 - t) t^4 z (-t^2 z + x_3)}{x_1 x_3} \, |01, 02, 01\rangle \\
& + \frac{(1 - t^2)(1 - t)t^6 z (-t z + x_3)}{x_1 x_3} \, |02, 01, 01\rangle 
 + \mathcal{D}_2(z)|12,00,01\rangle,
 \\
T^3(z)|12,00,01\rangle 
&=\frac{t^8(-t^2z+x_1)(-t^2z+x_2)(-t^2z+x_3)}{x_1x_2x_3}|12,00,01\rangle.
\end{align*}
The coefficient \eqref{sts} of the diagonal term generated by $T^k(z)$ is denoted by 
$\mathcal{D}_k(z)$.
In particular, the expressions for $\mathcal{D}_0(z)$ and $\mathcal{D}_3(z)$ admit compact forms owing to 
\eqref{t0f} and \eqref{tnf}, respectively.
\begin{align*}
\dot{T}^0(0)|12,00,01\rangle 
&= -\bigl(\frac{1}{x_1}+\frac{1}{x_2}+\frac{1}{x_3}\bigr) |12,00,01\rangle,
\\
\dot{T}^1(0)|12,00,01\rangle 
&= \frac{(1 - t)^2 t^3}{x_1} \, |01, 00, 12\rangle
  + \frac{(1 - t)(1 - t^2) t}{x_1} \, |01, 02, 01\rangle
  + \frac{(1 - t)^2 t^2}{x_1} \, |02, 00, 11\rangle \\
 & + \frac{(1 - t)(1 - t^2)}{x_1} \, |02, 01, 01\rangle 
 + \frac{(1 - t)^2 t^4}{x_1} \, |11, 00, 02\rangle
  + \frac{(1 - t)(1 - t^2) t}{x_3} \, |12, 01, 00\rangle \\
&  +\dot{\mathcal D}_1(0)|12,00,01\rangle,
\\
 \dot{T}^2(0)|12,00,01\rangle 
&= \frac{(1 - t)^2 t^7}{x_1} \, |01, 00, 12\rangle
  - \frac{(1 - t)^2 (1 - t^2) t^4}{x_1} \, |01, 01, 02\rangle
  + \frac{(1 - t)(1 - t^2) t^4}{x_1} \, |01, 02, 01\rangle \\
&  + \frac{(1 - t)^2 t^8}{x_1} \, |02, 00, 11\rangle  
+ \frac{(1 - t)(1 - t^2) t^6}{x_1} \, |02, 01, 01\rangle
  + \frac{(1 - t)^2 t^4}{x_1} \, |11, 00, 02\rangle \\
&  + \frac{(1 - t)(1 - t^2) t^6}{x_3} \, |12, 01, 00\rangle
  +\dot{\mathcal D}_2(0)|12,00,01\rangle,
 \\
 \dot{T}^3(0)|12,00,01\rangle 
&= -\bigl(\frac{1}{x_1}+\frac{1}{x_2}+\frac{1}{x_3}\bigr)t^8|12,00,01\rangle.
\end{align*}
Note that the coefficients of the off-diagonal terms in $\dot{T}^2(0)$ 
are neither all positive nor all negative in the range $0 <t < 1$ and 
$\forall x_j>0$.

According to \eqref{tdk0}, 
the functions $\dot{\mathcal D}_k(0) = \left.\frac{d\mathcal{D}_k(z)}{dz}\right|_{z=0}$ $(0 \le k \le 3)$ are given as
\begin{equation}\label{eeex}
\begin{split}
&\langle 12,00,01|\dot{T}^k(0)| 12,00,01\rangle
\\
&=  -\frac{1}{x_1}e_k(t^{\sigma_{1,0}}, t^{\KK_1+\sigma_{1,1}}, t^{\KK_2+\sigma_{1,2}})
- \frac{1}{x_2}e_k(t^{\sigma_{2,0}}, t^{\KK_1+\sigma_{2,1}}, t^{\KK_2+\sigma_{2,2}})
- \frac{1}{x_3}e_k(t^{\sigma_{3,0}}, t^{\KK_1+\sigma_{3,1}}, t^{\KK_2+\sigma_{3,2}})
\\
&=- \frac{1}{x_1}e_k(1, t^{\KK_1+1}, t^{\KK_2+1})
-\frac{1}{x_2}e_k(t^2, t^{\KK_1}, t^{\KK_2})
-\frac{1}{x_3}e_k(t, t^{\KK_1+1}, t^{\KK_2})
\\
&= -\frac{1}{x_1}e_k(1,t^4,t^6)-\frac{1}{x_2}e_k(t^2,t^3,t^5)-\frac{1}{x_3}e_k(t,t^4,t^5),
\end{split}
\end{equation}
where the multiplicity representations 
$\bs_1= (0,1,1), \bs_2=(2,0,0)$ and $\bs_3=(1,1,0)$ are used.
\end{example}

\section{Markov matrix from transfer matrices}\label{sec:HT}

Based on the transfer matrices in Section \ref{sec:T},
we introduce a linear operator
$\mathcal{H}_{n,l}=\mathcal{H}_{n,l}(x_1,\ldots, x_L)$ on $\mathbb{V}({\bf m})$
depending on the parameters 
$x_1, \ldots, x_L$ as an alternating sum
\begin{align}\label{Hc}
\mathcal{H}_{n,l} = D^{-1}_{\bf m} \sum_{k=0}^{n+1}(-1)^{k-1}\dot{T}^k(0) - 
\Bigl(\sum_{j=1}^L\frac{1}{x_j}\Bigr)\mathrm{Id},
\end{align}
where $D_{\bf m}$ is defined by \eqref{Dm}.
The first main result in this paper is the following:

\begin{theorem}\label{th:main1}
The Markov matrix $H_{n,l}$ of the $n$-species capacity-$l$ $t$-PushTASEP 
in \eqref{Hdef}-\eqref{Co} is identified with 
$\mathcal{H}_{n,l}$ in \eqref{Hc}.
Namely, the following equality holds in each sector $\mathbb{V}({\bf m})$:
\begin{align}\label{main1}
H_{n,l}(x_1,\ldots, x_L) = \mathcal{H}_{n,l}(x_1,\ldots, x_L) .
\end{align}
\end{theorem}

\begin{example}\label{ex:TH}
Consider the same case $n=l=2$, $L=3$ as Example \ref{ex:Tk}.
Then the Markov matrix \eqref{Hdef}-\eqref{Co} gives 
\begin{equation}
\begin{split}\label{h22}
H_{2,2}|12,00,01\rangle 
&= \frac{(1 - t)^2 t^3 (1 - t^4)}{(1 - t^2)(1 - t^3)(1 - t^5)\, x_1} \, |01, 00, 12\rangle
+ \frac{(1 - t)^2 t^4}{(1 - t^3)(1 - t^5)\, x_1} \, |01, 01, 02\rangle
\\
&+ \frac{(1 - t) t}{(1 - t^5)\, x_1} \, |01, 02, 01\rangle  
+ \frac{(1 - t)^2 t^2 (1 - t^6)}{(1 - t^2)(1 - t^3)(1 - t^5)\, x_1} \, |02, 00, 11\rangle
\\
&+ \frac{(1 - t)(1 - t^6)}{(1 - t^3)(1 - t^5)\, x_1} \, |02, 01, 01\rangle
+ \frac{(1 - t) t}{(1 - t^3)\, x_3} \, |12,01, 00\rangle
\\
&-\left(\frac{1}{x_1}+\frac{t(1-t)}{x_3(1-t^3)}\right)|12,00,01\rangle.
\end{split}
\end{equation}
On the other hand, from Example \ref{ex:Tk}, we get
\begin{equation}\label{td4}
\begin{split}
&(-\dot{T}^0(0)+\dot{T}^1(0)-\dot{T}^2(0)+\dot{T}^3(0)) |12,00,01\rangle
\\
&=
 \frac{(1 - t)^2 t^3 (1 - t^4)}{x_1} \, |01, 00, 12\rangle
+ \frac{(1 - t)^2 t^4 (1 - t^2)}{x_1} \, |01, 01, 02\rangle 
\\
&+ \frac{(1 - t) t (1 - t^2)(1 - t^3)}{x_1} \, |01, 02, 01\rangle 
 + \frac{(1 - t)^2 t^2 (1 - t^6)}{x_1} \, |02, 00, 11\rangle
 \\
&+ \frac{(1 - t)(1 - t^2)(1 - t^6)}{x_1} \, |02, 01, 01\rangle
+ \frac{(1 - t) t (1 - t^2)(1 - t^5)}{x_3} \, |12, 01, 00\rangle 
\\
& + \left( \frac{(1 - t^2)(1 - t^3)(1 - t^5)}{x_2}
+ \frac{(1 - t)(1 - t^4)(1 - t^5)}{x_3} \right) \, |12, 00, 01\rangle.
\end{split}
\end{equation}
In particular, the coefficient of the diagonal term in the last line of \eqref{td4} 
has been obtained by taking the alternating sum of \eqref{eeex}:
\begin{align*}
\sum_{k=0}^3(-1)^{k-1}\left(
-\frac{1}{x_1}e_k(1,t^4,t^6)-\frac{1}{x_2}e_k(t^2,t^3,t^5)-\frac{1}{x_3}e_k(t,t^4,t^5)
\right)
\end{align*} 
by means of \eqref{ekg}.
Using \eqref{td4} and  \( D_{\mathbf{m}} = (1 - t^2)(1 - t^3)(1 - t^5) \), one can verify that 
\(\mathcal{H}_{2,3} |12,00,01\rangle\) in \eqref{Hc} reproduces 
\(H_{2,3} |12,00,01\rangle\) as given in \eqref{h22}.
\end{example}

The rest of this section is devoted to the proof of Theorem \ref{th:main1}.

\subsection{Diagonal elements}\label{ss:de}
We first prove \eqref{main1} for the diagonal elements, i.e.,
\begin{align}
\langle \vec{\bs}|H_{n,l}|\vec{\bs}\rangle
=
\langle \vec{\bs}|\mathcal{H}_{n,l}|\vec{\bs}\rangle.
\end{align}
From \eqref{tdk0}, \eqref{ekg} and \eqref{Dm}, the RHS is evaluated as
\begin{align}
\langle \vec{\bs}|\mathcal{H}_{n,l}|\vec{\bs}\rangle
&=D_{\bf m}^{-1}\sum_{k=0}^{n+1}(-1)^{k-1}
\langle \vec{\bs}|\dot{T}^k(0)|\vec{\bs}\rangle-\sum_{j=1}^L\frac{1}{x_j}
\nonumber\\
&=D_{\bf m}^{-1}\sum_{k=0}^{n+1}(-1)^k
\sum_{j=1}^L\frac{1}{x_j}
e_k(t^{\sigma_{j,0}}, t^{\KK_1+\sigma_{j,1}},\ldots, t^{\KK_n+\sigma_{j,n}})
-\sum_{j=1}^L\frac{1}{x_j}
\nonumber\\
&=
\sum_{j=1}^L\frac{1}{x_j}
\prod_{h=0}^{n}\frac{1-t^{\KK_h + \sigma_{j,h}}}{1-t^{K_h}}
-\sum_{j=1}^L\frac{1}{x_j}
= \sum_{j=1}^L\frac{C_{\vec{\bs},j}(t)-1}{x_j},
\nonumber
\end{align}
where $C_{\vec{\bs},j}(t)$ is defined in \eqref{Co}.
This coincides with $\langle \vec{\bs}|H_{n,l}|\vec{\bs}\rangle$,
which is the coefficient of $|\vec{\bs}\rangle$ in the second term of \eqref{Hdef}.

\subsection{Reduced diagram and its depth}\label{ss:rd}

From now on, we assume $\vec{\bs}'\neq \vec{\bs}$
and concentrate on the off-diagonal elements 
$\langle\vec{\bs}'|H_{n,l}| \vec{\bs} \rangle$ and 
$\langle\vec{\bs}'|\mathcal{H}_{n,l}|\vec{\bs}\rangle$.
The former is given, from \eqref{Hdef}, as
\begin{subequations}
\begin{align}
\langle \vec{\bs}'|H_{n,l}| \vec{\bs}\rangle
&= \sum_{o=1}^L 
\langle \vec{\bs}'|H_{n,l}| \vec{\bs}\rangle_o,
\label{Hx}
\\
\langle \vec{\bs}'|H_{n,l}| \vec{\bs}\rangle_o
&= \frac{1}{x_o} \prod_{0 \le h \le n }
w^{(o)}_{\vec{\bs}, \vec{\bs}'}(h),
\label{Hj}
\end{align}
\end{subequations}
where the factor $w^{(o)}_{\vec{\bs}, \vec{\bs}'}(h)$ 
has been defined in \eqref{wm1}.
On the other hand $\langle\vec{\bs}'|\mathcal{H}_{n,l}|\vec{\bs}\rangle$ is given,
from \eqref{tke} and \eqref{Hc}, as
\begin{subequations}
\begin{align}
\langle\vec{\bs}'|\mathcal{H}_{n,l}|\vec{\bs}\rangle
&=  D_{\bf m}^{-1}\sum_{k=0}^{n+1}(-1)^{k-1}
\sum_{o=1}^L\langle\vec{\bs}'|\dot{T}^k(0)|\vec{\bs}\rangle_o,
\label{Hct}
\\
\langle\vec{\bs}'|\dot{T}^k(0)|\vec{\bs}\rangle_o
&= \frac{1}{x_o} \sum_{{\bf a}_1,\ldots, {\bf a}_L \in \BB^k}
S^{k}_{\;\,l}(0)^{{\bf a}_2, \bs'_1}_{{\bf a}_1, \bs_1}
\cdots 
\dot{S}^{k}_{\;\,l}(0)^{{\bf a}_{o+1}, \bs'_o}_{{\bf a}_o, \bs_o} 
\cdots 
S^{k}_{\;\,l}(0)^{{\bf a}_1, \bs'_L}_{{\bf a}_L,  \bs_L}.
\label{Hcj}
\end{align}
\end{subequations}
where $\dot{S}(z)= \frac{dS(z)}{dz}$.
Thus the equality 
$\langle\vec{\bs}'|H_{n,l}|\vec{\bs}\rangle
=\langle\vec{\bs}'|\mathcal{H}_{n,l}|\vec{\bs}\rangle$
for any $\vec{\bs} \neq \vec{\bs}' \in S({\bf m})$
follows once we show
\begin{equation}\label{HHj}
D_{\bf m}^{-1}\sum_{k=0}^{n+1}(-1)^{k-1}
x_o \langle \vec{\bs}'|\dot{T}^k(0)| \vec{\bs}\rangle_o
= \prod_{0 \le h \le n}w^{(o)}_{\vec{\bs}, \vec{\bs}'}(h).
\end{equation}
From \eqref{Dm},  \eqref{wm2} and \eqref{Hcj}, this is equivalent to 
\begin{equation}\label{HHj2}
\begin{split}
&\sum_{k=0}^{n+1}(-1)^{k-1}
\sum_{{\bf a}_1,\ldots, {\bf a}_L \in \BB^k}
S^{k}_{\;\,l}(0)^{{\bf a}_2, \bs'_1}_{{\bf a}_1, \bs_1}
\cdots 
\dot{S}^{k}_{\;\,l}(0)^{{\bf a}_{o+1}, \bs'_o}_{{\bf a}_o, \bs_o} 
\cdots 
S^{k}_{\;\,l}(0)^{{\bf a}_1, \bs'_L}_{{\bf a}_L,  \bs_L}
\\
&= 
\prod_{h \in \{h_0, \ldots, h_g\}}
(1 - t^{\sigma_{p(h),h}})\, t^{\ell_h}
\prod_{h \in \{\bar{h}_1, \ldots, \bar{h}_{n-g}\}}
(1 - t^{\KK_h + \Phi_h})
\qquad (\vec{\bs} \neq \vec{\bs}').
\end{split}
\end{equation}
This relation achieves  two simplifications from the original problem.
There is no summation over the sites $o=1,\ldots, L$, and the dependence on 
$x_1,\ldots, x_L$ is eliminated, leaving it dependent only on the parameter $t$.
Elements of $S^{k}_{\;\,l}(0)$ and $\dot{S}^{k}_{\;\,l}(0)$ consisting of \eqref{HHj2} have been 
obtained in Proposition \ref{pr:s0}.

\medskip
Let us depict $x_o \langle \vec{\bs}'|\dot{T}^k(0)| \vec{\bs}\rangle_o$ as in Figure \ref{fig:tk},
omitting the spectral parameters $z/x_i$ since they are set to zero.
All the vertical arrows from $\bs_i$ to $\bs'_i$ with $\bs_i=\bs'_i$,  
corresponding to local diagonal transitions, are suppressed. 
Moreover, for simplicity, we apply a cyclic shift such that the site $o$ appears in the leftmost position,
and attach the symbol $\circ$ to it to indicate that $\dot{S}(0)$ is used there, 
in contrast to $S(0)$ elsewhere.
We refer to such a diagram as {\em reduced diagram}.
See \eqref{red}, where ${\bf a}_i \in \BB^k$, $\bs_i\neq \br_i \in \BB_l$ for 
$0 \le i \le g$ with some $1 \le g <L$.\footnote{$\bs_0,\ldots, \bs_g$ should be understood as a
relabeling of the local states $\bs_{j_0},\ldots \bs_{j_g}$ that undergo nondiagonal transitions 
with respect to the original site indices.
In this context, the leftmost site 0 in \eqref{red} corresponds to the site $o$.}

\begin{equation}
\begin{tikzpicture}

\node at (-1.8,-0.18){$\displaystyle{\sum_{\phantom{AA}{\bf a}_0,\ldots, {\bf a}_g \in \BB^k}}$};
  \draw[->] (0.6,-0.6) node[below]{$\bs_0$}-- (0.6,0.6) node[above]{$\br_0$};
  \draw (0.59, 0) circle[radius=0.1cm];
  \draw[->,line width=0.6mm] (0,0) node[left]{${\bf a}_0$}-- (1.2,0) node[right]{${\bf a}_1$};
\begin{scope}[shift={(1.8,0)}]
  \draw[->] (0.6,-0.6) node[below]{$\bs_1$}-- (0.6,0.6) node[above]{$\br_1$};
  \draw[->,line width=0.6mm] (0,0) -- (1.2,0) node[right]{${\bf a}_2$};
   \end{scope}

  \draw[->,line width=0.6mm] (3.6,0) -- (4.4,0);
  \node at (4.8,0) {$\cdots$};
  \draw[->,line width=0.6mm] (5.1,0) -- (5.9,0);
   
   \begin{scope}[shift={(6.6,0)}]
     \draw[->] (0.6,-0.6) node[below]{$\bs_g$}-- (0.6,0.6) node[above]{$\br_g$};
  \draw[->,line width=0.6mm] (0,0) node[left]{${\bf a}_g$}-- (1.2,0) node[right]{${\bf a}_0$};
     \end{scope}
\end{tikzpicture}
\label{red}
\end{equation}
The diagram should be understood as representing the sum in \eqref{Hcj}, where the $L-g-1$ vertical arrows 
corresponding to the diagonal transitions are suppressed, but their associated vertex weights 
should still be accounted for.
Since the carriers ${\bf a}_i$'s remain unchanged when crossing the suppressed vertical arrows,
the sum reduces to those over ${\bf a}_0, \ldots, {\bf a}_g \in \BB^k$, where 
${\bf a}_{i+1} = {\bf a}_i + \bs_i-\br_i$ $(i \mod g+1)$ in multiplicity representation. 

\begin{lemma}\label{le:rs}
$\langle \vec{\bs}'|\dot{T}^k(0)| \vec{\bs}\rangle_o = 0$
unless the reduced diagram \eqref{red} for it satisfies the conditions:
\begin{subequations}
\begin{align}
&\br_0+ \cdots +\br_g = \bs_0 + \cdots + \bs_g,
\label{rscon1}\\
&\br_0 \prec \bs_0, \; \bs_1 \prec \br_1,\ldots, \bs_g \prec \br_g,
\label{rscon2}
\end{align}
where \eqref{rscon1} is an equality of the arrays in the multiplicity representation,
and $\prec$ is defined in \eqref{prel}.
\end{subequations}
\end{lemma}
\begin{proof}
From weight conservation, \eqref{red} vanishes unless the condition \eqref{rscon1} holds.
The condition \eqref{rscon2} follows rom Proposition \ref{pr:s0}.
\end{proof}

\begin{lemma}\label{le:rst}
Let $\vec{\bs} \rightarrow \vec{\bs}'$ be a transition of states in $(\BB_l)^L$.
Then the following two conditions are equivalent:

\begin{itemize}
  \item[(a)] The transition fits the scheme described 
  in Table~\ref{tab1};  equivalently, it satisfies conditions~(i) and~(ii) stated between \eqref{hhbar} and \eqref{wm1}.
  
  \item[(b)] The associated reduced diagram 
  \eqref{red}\footnote{The indices $j$ of $\bs_j$ in \eqref{red} do not necessarily match those 
in $\vec{\bs}=(\bs_1,\ldots, \bs_L)$, 
since the sites undergoing diagonal transitions are suppressed in the reduced diagram.} satisfies conditions~\eqref{rscon1} and~\eqref{rscon2}.
\end{itemize}
\end{lemma}
\begin{proof}
The implication (a) $\Rightarrow$ (b) is straightforward by reference to Figure~\ref{fig:1}.
Now assume (b). Let $\br_o - \bs_o = {\bf e}_{r_1} - \cdots$ in accordance with \eqref{prel} or \eqref{rr1}, 
where the ellipsis $\cdots$ denotes an alternating sum of some ${\bf e}_s$ with $s > r_1$. 
The same convention applies in the sequel.
Then by condition~\eqref{rscon1}, 
there must exist a pair $\bs_j \prec \br_j$ among those in~\eqref{rscon2} 
such that $\br_j = {\bf e}_{r_1} - \cdots$, 
and all other pairs satisfy $\bs_{j'} - \br_{j'} = {\bf e}_{r'} - \cdots$ with $r' > r_1$.
Repeating this argument shows that the transition described in~\eqref{red} fits the scheme in Table~\ref{tab1}.
\end{proof}

Suppose the reduced diagram \eqref{red} satisfies \eqref{rscon1} and \eqref{rscon2}.
To ensure the weight conservation at every vertex, the capacity $k$ 
of the carriers must be at least a certain value.
We define the minimum possible capacity as the {\em depth} $d$ of the reduced diagram or the transition 
$\vec{\bs} \rightarrow \vec{\bs}'$.
Clearly, the depth is unaffected by the diagonal part of the transition 
which is suppressed in the reduced diagram.
We refer to the carriers whose capacity equals the depth as {\em minimal carriers}. 
They are uniquely determined from $\vec{\bs}$ and $\vec{\bs}'$.
See Example \ref{ex:d} below.

\begin{example}\label{ex:d}
Examples of reduced diagrams and minimal carries for $n=4, l=2$ and depth $d=1,2,3$.
We adopt the tableau representation for elements from $\BB_2$ as
$00=(2,0,0,0,0)$, $01=(1,1,0,0,0)$, $02=(1,0,1,0,0)$, $03=(1,0,0,1,0)$, $04=(1,0,0,0,1)$, 
$24=(0,0,1,0,1)$, etc.
All the diagrams satisfy the conditions \eqref{rscon1} and \eqref{rscon2}
with the unique choice of carriers exhibited also in the (column strict) tableau representation.
The examples (1u), (2u), (3u) are {\em unwanted}, whereas (1w), (2w), (3w) are {\em wanted}
in the sense that the condition \eqref{so} is satisfied or not, respectively.
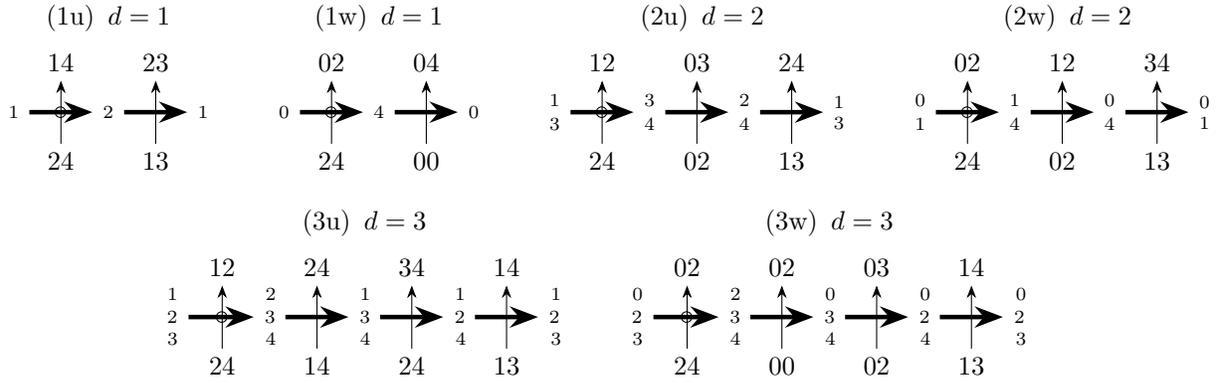
\begin{align*}
&\begin{tikzpicture}[scale=0.7]
\node at (1.5,1.8){(1u)\, $d=1$};
  \draw[->] (0.6,-0.6) node[below]{$24$}-- (0.6,0.6) node[above]{$14$};
  \draw (0.59, 0) circle[radius=0.1cm];
  \node at   (-0.3,0) {$\scriptstyle{1}$}; 
    \draw[->,line width=0.6mm] (0,0) node[left]{}-- (1.2,0) node[right]{};
\begin{scope}[shift={(1.8,0)}]
   \draw[->] (0.6,-0.6) node[below]{$13$}-- (0.6,0.6) node[above]{$23$};
   \node at   (-0.3,0) {$\scriptstyle{2}$}; 
  \draw[->,line width=0.6mm] (0,0) -- (1.2,0) node[right]{};
\node at   (1.5,0) {$\scriptstyle{1}$}; 
\end{scope}
\end{tikzpicture}
\qquad
\begin{tikzpicture}[scale=0.7]
\node at (1.5,1.8){(1w)\, $d=1$};
  \draw[->] (0.6,-0.6) node[below]{$24$}-- (0.6,0.6) node[above]{$02$};
  \draw (0.59, 0) circle[radius=0.1cm];
  \node at   (-0.3,0) {$\scriptstyle{0}$}; 
    \draw[->,line width=0.6mm] (0,0) node[left]{}-- (1.2,0) node[right]{};
\begin{scope}[shift={(1.8,0)}]
  \draw[->] (0.6,-0.6) node[below]{$00$}-- (0.6,0.6) node[above]{$04$};
 \node at   (-0.3,0) {$\scriptstyle{4}$}; 
  \draw[->,line width=0.6mm] (0,0) -- (1.2,0) node[right]{};
  \node at   (1.5,0) {$\scriptstyle{0}$}; 
 \end{scope}
\end{tikzpicture}
\qquad
\begin{tikzpicture}[scale=0.7]
\node at (2.5,1.8){(2u)\, $d=2$};
  \draw[->] (0.6,-0.6) node[below]{$24$}-- (0.6,0.6) node[above]{$12$};
  \draw (0.59, 0) circle[radius=0.1cm];
  \node at   (-0.3,0.22) {$\scriptstyle{1}$}; \node at   (-0.3,-0.22) {$\scriptstyle{3}$};
  \draw[->,line width=0.6mm] (0,0) node[left]{}-- (1.2,0) node[right]{};
\begin{scope}[shift={(1.8,0)}]
  \draw[->] (0.6,-0.6) node[below]{$02$}-- (0.6,0.6) node[above]{$03$};
   \node at   (-0.3,0.22) {$\scriptstyle{3}$}; \node at   (-0.3,-0.22) {$\scriptstyle{4}$};
  \draw[->,line width=0.6mm] (0,0) -- (1.2,0) node[right]{};
 \end{scope}
\begin{scope}[shift={(3.6,0)}]
  \draw[->] (0.6,-0.6) node[below]{$13$}-- (0.6,0.6) node[above]{$24$};
   \node at   (-0.3,0.22) {$\scriptstyle{2}$}; \node at   (-0.3,-0.22) {$\scriptstyle{4}$};
  \draw[->,line width=0.6mm] (0,0) -- (1.2,0) node[right]{};
     \node at   (1.5,0.2) {$\scriptstyle{1}$}; \node at   (1.5,-0.2) {$\scriptstyle{3}$};
 \end{scope}
\end{tikzpicture}
\qquad
\begin{tikzpicture}[scale=0.7]
\node at (2.5,1.8){(2w)\, $d=2$};
  \draw[->] (0.6,-0.6) node[below]{$24$}-- (0.6,0.6) node[above]{$02$};
  \draw (0.59, 0) circle[radius=0.1cm];
  \node at   (-0.3,0.22) {$\scriptstyle{0}$}; \node at   (-0.3,-0.22) {$\scriptstyle{1}$};
  \draw[->,line width=0.6mm] (0,0) node[left]{}-- (1.2,0) node[right]{};
\begin{scope}[shift={(1.8,0)}]
  \draw[->] (0.6,-0.6) node[below]{$02$}-- (0.6,0.6) node[above]{$12$};
   \node at   (-0.3,0.22) {$\scriptstyle{1}$}; \node at   (-0.3,-0.22) {$\scriptstyle{4}$};
  \draw[->,line width=0.6mm] (0,0) -- (1.2,0) node[right]{};
 \end{scope}
\begin{scope}[shift={(3.6,0)}]
  \draw[->] (0.6,-0.6) node[below]{$13$}-- (0.6,0.6) node[above]{$34$};
   \node at   (-0.3,0.22) {$\scriptstyle{0}$}; \node at   (-0.3,-0.22) {$\scriptstyle{4}$};
  \draw[->,line width=0.6mm] (0,0) -- (1.2,0) node[right]{};
     \node at   (1.5,0.2) {$\scriptstyle{0}$}; \node at   (1.5,-0.2) {$\scriptstyle{1}$};
 \end{scope}
\end{tikzpicture}
\\
&\qquad\qquad\qquad 
\begin{tikzpicture}[scale=0.7]
\node at (3.3,1.8){(3u)\, $d=3$};
  \draw[->] (0.6,-0.6) node[below]{$24$}-- (0.6,0.6) node[above]{$12$};
  \draw (0.59, 0) circle[radius=0.1cm];
\node at   (-0.34,0.43) {$\scriptstyle{1}$}; \node at   (-0.34,-0.42) {$\scriptstyle{3}$};
  \draw[->,line width=0.6mm] (-0.04,0) node[left]{$\scriptstyle 2$}-- (1.2,0) node[right]{$\scriptstyle 3$};
\begin{scope}[shift={(1.8,0)}]
  \draw[->] (0.6,-0.6) node[below]{$14$}-- (0.6,0.6) node[above]{$24$};
  \node at   (-0.28,0.43) {$\scriptstyle{2}$}; \node at   (-0.28,-0.42) {$\scriptstyle{4}$};
  \draw[->,line width=0.6mm] (0,0) -- (1.2,0) node[right]{$\scriptstyle 3$};
 \end{scope}
\begin{scope}[shift={(3.6,0)}]
  \draw[->] (0.6,-0.6) node[below]{$24$}-- (0.6,0.6) node[above]{$34$};
  \node at   (-0.28,0.43) {$\scriptstyle{1}$}; \node at   (-0.28,-0.42) {$\scriptstyle{4}$};
  \draw[->,line width=0.6mm] (0,0) -- (1.2,0) node[right]{$\scriptstyle 2$};
 \end{scope}
\begin{scope}[shift={(5.4,0)}]
  \draw[->] (0.6,-0.6) node[below]{$13$}-- (0.6,0.6) node[above]{$14$};
  \node at   (-0.28,0.43) {$\scriptstyle{1}$}; \node at   (-0.27,-0.42) {$\scriptstyle{4}$};
  \draw[->,line width=0.6mm] (0,0) -- (1.2,0) node[right]{$\scriptstyle 2$};
  \node at   (1.54,0.43) {$\scriptstyle{1}$}; \node at   (1.54,-0.42) {$\scriptstyle{3}$};
 \end{scope}
\end{tikzpicture}
\qquad
\begin{tikzpicture}[scale=0.7]
\node at (3.3,1.8){(3w)\, $d=3$};
  \draw[->] (0.6,-0.6) node[below]{$24$}-- (0.6,0.6) node[above]{$02$};
  \draw (0.59, 0) circle[radius=0.1cm];
\node at   (-0.34,0.43) {$\scriptstyle{0}$}; \node at   (-0.34,-0.42) {$\scriptstyle{3}$};
  \draw[->,line width=0.6mm] (-0.04,0) node[left]{$\scriptstyle 2$}-- (1.2,0) node[right]{$\scriptstyle 3$};
\begin{scope}[shift={(1.8,0)}]
  \draw[->] (0.6,-0.6) node[below]{$00$}-- (0.6,0.6) node[above]{$02$};
  \node at   (-0.28,0.43) {$\scriptstyle{2}$}; \node at   (-0.28,-0.42) {$\scriptstyle{4}$};
  \draw[->,line width=0.6mm] (0,0) -- (1.2,0) node[right]{$\scriptstyle 3$};
 \end{scope}
\begin{scope}[shift={(3.6,0)}]
  \draw[->] (0.6,-0.6) node[below]{$02$}-- (0.6,0.6) node[above]{$03$};
  \node at   (-0.28,0.43) {$\scriptstyle{0}$}; \node at   (-0.28,-0.42) {$\scriptstyle{4}$};
  \draw[->,line width=0.6mm] (0,0) -- (1.2,0) node[right]{$\scriptstyle 2$};
 \end{scope}
\begin{scope}[shift={(5.4,0)}]
  \draw[->] (0.6,-0.6) node[below]{$13$}-- (0.6,0.6) node[above]{$14$};
  \node at   (-0.28,0.43) {$\scriptstyle{0}$}; \node at   (-0.27,-0.42) {$\scriptstyle{4}$};
  \draw[->,line width=0.6mm] (0,0) -- (1.2,0) node[right]{$\scriptstyle 2$};
  \node at   (1.54,0.43) {$\scriptstyle{0}$}; \node at   (1.54,-0.42) {$\scriptstyle{3}$};
 \end{scope}
\end{tikzpicture}
\end{align*}
\end{example}

\begin{remark}\label{re:md}
Diagrams such as Figure~\ref{fig:1} depict transitions of states corresponding to 
the reduced diagram~\eqref{red} associated with minimal carriers.
Figure~\ref{fig:1} represents the case where ${\bf a}_0, \ldots, {\bf a}_{g=4} \in \BB^{d}$, 
with depth $d=4$ being the number of horizontal lines intersecting every vertical slice 
except at the positions $j = o, j_1, \ldots, j_4$.
The minimal carriers are specified by the heights $h_q$ 
of the horizontal lines in the multiplicity representation~\eqref{vk} as
\begin{equation*}
\begin{split}
&{\bf a}_0 = (\theta(i \in \{h_0, h_3, h_5, h_7\}))^n_{i=0}, \quad
{\bf a}_1 = (\theta(i \in\{ h_3, h_5, h_6, h_9\}))^n_{i=0}, \\
&{\bf a}_2 = (\theta(i \in\{ h_2, h_4, h_6, h_8\}))^n_{i=0}, \quad
{\bf a}_3 = (\theta(i \in\{ h_2, h_3, h_6, h_8\}))^n_{i=0},\\
&{\bf a}_4 = (\theta(i \in\{ h_1, h_3, h_6, h_7\}))^n_{i=0}.
\end{split}
\end{equation*}
As Example \ref{ex:d} demonstrates, we have $d \le g$ in general.
\end{remark}

\subsection{Contribution from minimum carrier}\label{ss:cmm}

For simplicity, 
we write $S^{k}_{\;\,l}(z)$ in \eqref{Hcj} as $S(z)$ in the rest of this section.

\begin{lemma}\label{le:mc}
Let $\vec{\bs} \rightarrow \vec{\bs}'$ be a transition whose 
reduced diagram \eqref{red} satisfies conditions \eqref{rscon1} and \eqref{rscon2}. 
Let $d$ denote its depth, 
and let ${\bf a}_1, \ldots, {\bf a}_L \in \BB^d$ be the minimum carrier.\footnote{The indices of 
$\bs_i$ and ${\bf a}_i$ do not necessarily match those in \eqref{red}, 
since the latter is a reduced diagram omitting the 
sites undergoing diagonal transitions.}
Let $0 \le h_0 < \cdots < h_g \le n$ be the species of the moved particles, 
in accordance with the scheme of Table~\ref{tab1}, as guaranteed by Lemma~\ref{le:rst}.
Then the following equality holds:
\begin{align}\label{mc1}
S(0)^{{\bf a}_2, \bs'_1}_{{\bf a}_1, \bs_1}
\cdots 
\dot{S}(0)^{{\bf a}_{o+1}, \bs'_o}_{{\bf a}_o, \bs_o} 
\cdots 
S(0)^{{\bf a}_1, \bs'_L}_{{\bf a}_L,  \bs_L}
=(-1)^{d+1}\prod_{h \in \{h_0, \ldots, h_g\}}\!\!(1 - t^{\sigma_{p(h),h}})\, t^{\ell_h}.
\end{align}
\end{lemma}
\begin{proof}

From Remark~\ref{re:md}, the transition can be depicted by a diagram as in Figure~\ref{fig:1}.
The LHS is evaluated via Proposition~\ref{pr:s0}, and it evidently consists of three parts: 
$(-1)^A t^B \prod (1 - t^\#)$.
We illustrate the correspondence of each part using the example in Figure~\ref{fig:1}.
The general case proceeds analogously.
The lattice sites relevant to the reduced diagram in Figure~\ref{fig:1} are 
$0, j_1, \ldots, j_4$.

\smallskip
\noindent
$\bullet$ \textbf{The} $\prod (1 - t^\#)$ \textbf{part.}
From \eqref{Sz1} with \eqref{rr1}, or from \eqref{Sz0} with \eqref{rr2}, the following contributions arise:
\begin{equation}
\begin{split}
\dot{S}(0)^{{\bf a}_{o+1}, \bs'_o}_{{\bf a}_o, \bs_o}: 
&\quad (1 - t^{\sigma_{o,r_2}})(1 - t^{\sigma_{o,r_4}})
= (1 - t^{\sigma_{o,h_6}})(1 - t^{\sigma_{o,h_9}})
= (1 - t^{\sigma_{p(h_6), h_6}})(1 - t^{\sigma_{p(h_9), h_9}}),
\\
S(0)^{{\bf a}_{j_1+1}, \bs'_{j_1}}_{{\bf a}_{j_1}, \bs_{j_1}}: 
&\quad (1 - t^{\sigma_{j_1,s_1}})(1 - t^{\sigma_{j_1,s_3}})(1 - t^{\sigma_{j_1,s_5}})
= (1 - t^{\sigma_{p(h_2), h_2}})(1 - t^{\sigma_{p(h_4), h_4}})(1 - t^{\sigma_{p(h_8), h_8}}),
\\
S(0)^{{\bf a}_{j_2+1}, \bs'_{j_2}}_{{\bf a}_{j_2}, \bs_{j_2}}: 
&\quad (1 - t^{\sigma_{j_2,t_1}})
= (1 - t^{\sigma_{p(h_3), h_3}}),
\\
S(0)^{{\bf a}_{j_3+1}, \bs'_{j_3}}_{{\bf a}_{j_3}, \bs_{j_3}}: 
&\quad (1 - t^{\sigma_{j_3,u_1}})(1 - t^{\sigma_{j_3,u_3}})
= (1 - t^{\sigma_{p(h_1), h_1}})(1 - t^{\sigma_{p(h_7), h_7}}),
\\
S(0)^{{\bf a}_{j_4+1}, \bs'_{j_4}}_{{\bf a}_{j_4}, \bs_{j_4}}: 
&\quad (1 - t^{\sigma_{j_4,v_1}})(1 - t^{\sigma_{j_4,v_3}})
= (1 - t^{\sigma_{p(h_0), h_0}})(1 - t^{\sigma_{p(h_5), h_5}}).
\end{split}
\end{equation}
The product of all these factors matches the RHS 
$\prod_{h \in \{h_0, \ldots, h_9\}} (1 - t^{\sigma_{p(h), h}})$.
Note that the sites suppressed in the reduced diagram do not contribute to this, 
since the associated factor is the $z = 0$ specialization of \eqref{sdiag}.

\smallskip
\noindent
$\bullet$ \textbf{The sign factor} $(-1)^A$.
From \eqref{alp0} and \eqref{alp1}, the following contributions to $A$ arise:
\begin{align}
\dot{S}(0)^{{\bf a}_{o+1}, \bs'_o}_{{\bf a}_o, \bs_o}: 
&\quad  1 + a_{o, [r_1, r_2)} + a_{o, [r_3, r_4)},
\label{dpl} \\
S(0)^{{\bf a}_{j_1+1}, \bs'_{j_1}}_{{\bf a}_{j_1}, \bs_{j_1}}: 
&\quad a_{j_1, (s_1, s_2)} + a_{j_1, (s_3, s_4)} + a_{j_1, (s_5, s_6)}, \nonumber \\
S(0)^{{\bf a}_{j_2+1}, \bs'_{j_2}}_{{\bf a}_{j_2}, \bs_{j_2}}: 
&\quad a_{j_2, (t_1, t_2)}, \nonumber \\
S(0)^{{\bf a}_{j_3+1}, \bs'_{j_3}}_{{\bf a}_{j_3}, \bs_{j_3}}: 
&\quad a_{j_3, (u_1, u_2)} + a_{j_3, (u_3, u_4)}, \nonumber \\
S(0)^{{\bf a}_{j_4+1}, \bs'_{j_4}}_{{\bf a}_{j_4}, \bs_{j_4}}: 
&\quad a_{j_4, (v_1, v_2)} + a_{j_4, (v_3, v_4)}. \nonumber
\end{align}
Here, for example, 
$a_{o, [r_1, r_2)} = \sum_{r_1 \le \kappa < r_2} a_{o,\kappa}$ for ${\bf a}_o = (a_{o,0}, \ldots, a_{o,n}) \in \BB^d$.
The sum $a_{o, [r_1, r_2)} + a_{o, [r_3, r_4)}$ in \eqref{dpl} 
counts the total number of horizontal lines 
entering the site $o$ from the left\footnote{This originates the alignment of four indices 
of $\dot{S}(0)$ in  its diagram representation mentioned in Figure \ref{fig:tk}. } in Figure~\ref{fig:1}.
It equals $4$ in the figure and more generally equals $d$, in accordance with Remark~\ref{re:md}.
Thus, \eqref{dpl} already yields the sign factor $(-1)^{d + 1}$.

All other quantities in the above are zero, 
since there are no horizontal lines arriving at sites $j_1, \ldots, j_4$ at heights not equal to some $h_q$.
As before, the sites suppressed in the reduced diagram do not contribute to the sign factor, 
since this corresponds to the $z = 0$ case of \eqref{sdiag}.

\smallskip
\noindent
$\bullet$ \textbf{The power} $t^B$.
Consider $\dot{S}(0)^{{\bf a}_{o+1}, \bs'_o}_{{\bf a}_o, \bs_o}$ for example. 
From \eqref{bet1}, nonzero contributions to the power of $t$ arise only from the first term 
$\sum_{0 \le r<s \le n} \sigma_{o,r} a_{o,s}$ since 
$a_{o,s}=\theta(s \in \{h_3,h_5,h_6,h_9\})$ leads to 
$(\sigma_o a_o)_{[0,r_1)} + (\sigma_o a_o)_{(r_2,r_3)} + (\sigma_o a_o)_{(r_4,n]} =0$. 
(${\bf a}_o$ here is ${\bf a}_1$ in Remark \ref{re:md}.)
Similarly, nonzero contributions to the power of \( t \) from 
\( S(0)^{{\bf a}_{j+1}, \bs'_j}_{{\bf a}_j, \bs_j} \) also come only 
from the first term in \eqref{bet0}. Thus, in our working example 
of Figure~\ref{fig:1}, the following contributions to the exponent 
\( B \) arise from the nondiagonal local transitions:
\begin{equation}\label{pow1}
\begin{split}
\dot{S}(0)^{{\bf a}_{o+1}, \bs'_o}_{{\bf a}_o, \bs_o}: 
&\; \bigl(\sum_{0 \le r < h_3}+\sum_{0 \le r < h_5}+\sum_{0 \le r < h_6}
+\sum_{0 \le r < h_9}\bigr) \sigma_{o,r}
= \sum_{h \in \{h_6,h_9\}}\sigma_{p(h),[0,h)}
+\sigma_{o,[0,h_3)} + \sigma_{o,[0,h_5)},
\\
S(0)^{{\bf a}_{j_1+1}, \bs'_{j_1}}_{{\bf a}_{j_1}, \bs_{j_1}}: 
&\; \bigl(\sum_{0 \le r < h_2}+\sum_{0 \le r < h_4}+\sum_{0 \le r < h_6}
+\sum_{0 \le r < h_8}\bigr) \sigma_{j_1,r}
= \sum_{h \in \{h_2,h_4, h_8\}}\sigma_{p(h),[0,h)}
+\sigma_{j_1,[0,h_6)},
\\
S(0)^{{\bf a}_{j_2+1}, \bs'_{j_2}}_{{\bf a}_{j_2}, \bs_{j_2}}: 
&\; \bigl(\sum_{0 \le r < h_2}+\sum_{0 \le r < h_3}+\sum_{0 \le r < h_6}
+\sum_{0 \le r < h_8}\bigr) \sigma_{j_2,r}
= \sigma_{p(h_3),[0,h_3)}
+\sigma_{j_2, [0,h_2)}+\sigma_{j_2, [0,h_6)}+\sigma_{j_2, [0,h_8)},
\\
S(0)^{{\bf a}_{j_3+1}, \bs'_{j_3}}_{{\bf a}_{j_3}, \bs_{j_3}}: 
&\; \bigl(\sum_{0 \le r < h_1}+\sum_{0 \le r < h_3}+\sum_{0 \le r < h_6}
+\sum_{0 \le r < h_7}\bigr) \sigma_{j_3,r}
= \sum_{h \in \{h_1,h_7\}}\sigma_{p(h),[0,h)}
+\sigma_{j_3,[0,h_3)}+\sigma_{j_3,[0,h_6)},
\\
S(0)^{{\bf a}_{j_4+1}, \bs'_{j_4}}_{{\bf a}_{j_4}, \bs_{j_4}}: 
&\; \bigl(\sum_{0 \le r < h_0}+\sum_{0 \le r < h_3}+\sum_{0 \le r < h_5}
+\sum_{0 \le r < h_7}\bigr) \sigma_{j_4,r}
= \sum_{h \in \{h_0,h_5\}}\sigma_{p(h),[0,h)}
+\sigma_{j_4,[0,h_3)}+\sigma_{j_4, [0,h_7)}.
\end{split}
\end{equation}
In addition to these, the terms 
$S(0)^{{\bf a}_{j+1}, \bs'_j}_{{\bf a}_j, \bs_j}$ 
with $j \in \{1,\ldots, L\} \setminus \{o, j_1, \ldots, j_4\}$ in the LHS of \eqref{mc1}
also contribute to the exponent $B$.
They correspond to the diagonal  local transitions 
$\bs'_j = \bs_j$, ${\bf a}_{j+1} = {\bf a}_j$, which yield
$\sum_{0 \le r < s \le n}\sigma_{j,r}a_{j,s}$ by \eqref{sdiag}.
In Figure~\ref{fig:1}, the diagonal contributions come from sites 
$j$ with $a_{j,h}=1$ for some $h \in \{h_0,\ldots,h_9\}$.
They correspond to the horizontal line segments $j \in (p(h), p'(h))$ 
excluding $o, j_1, \ldots, j_4$ already counted in \eqref{pow1}.
Taking into account overlaps such as $(p(h_8), p'(h_8)) \ni j_2$, 
the total contribution from the diagonal parts is given by
\begin{equation}\label{pow2}
\begin{split}
&\sum_{h \in \{h_0,\ldots, h_9\}}
\sum_{j \in (p(h), p'(h))} \sigma_{j, [0,h)} 
-\sigma_{o, [0,h_3)} -\sigma_{o, [0,h_5)} 
-\sigma_{j_1, [0,h_6)}  
\\
&\quad-\sigma_{j_2, [0,h_2)}  -\sigma_{j_2, [0,h_6)}  -\sigma_{j_2, [0,h_8)}
-\sigma_{j_3, [0,h_3)}-\sigma_{j_3, [0,h_6)} 
-\sigma_{j_4, [0,h_3)} -\sigma_{j_4, [0,h_7)}.  
\end{split}
\end{equation}
Summing \eqref{pow1} and \eqref{pow2}, we obtain 
\begin{align}
\sum_{h \in \{h_0,\ldots, h_9\}} 
\sum_{j \in [p(h), p'(h))} \sigma_{j, [0,h)} = \sum_{h \in \{h_0,\ldots, h_9\}} \ell_h,
\end{align}
where $\ell_h$ is defined in \eqref{lh}.
This result coincides with the power of $t$ in the RHS of \eqref{mc1}, completing the proof.
\end{proof}

\subsection{Contribution from non-minimum carriers}

Let us evaluate the  summand in  \eqref{HHj2} corresponding to the non-minimum carriers
${\bf a}_1,\ldots, {\bf a}_L \in \BB^k$ with $k>d$.
The following argument is a natural generalization of {\em Step 2} in \cite[Sec.5.3]{AK25}. 
In general, non-minimum carriers are obtained by supplementing  
common letters that are not contained in the minimum carries in tableau representation.
Recall that  in the plot as  Figure \ref{fig:1}, the minimum carries are represented as the 
horizontal line segments 
whose height signifies the letter from $[0,n]$ contained in their tableau representations.
(This is explained in Remark \ref{re:md}.)
The supplemented letters can be depicted 
as extra horizontal lines that do not overlap the existing horizontal line segments.
An example of them is shown in Figure \ref{fig:ff}.

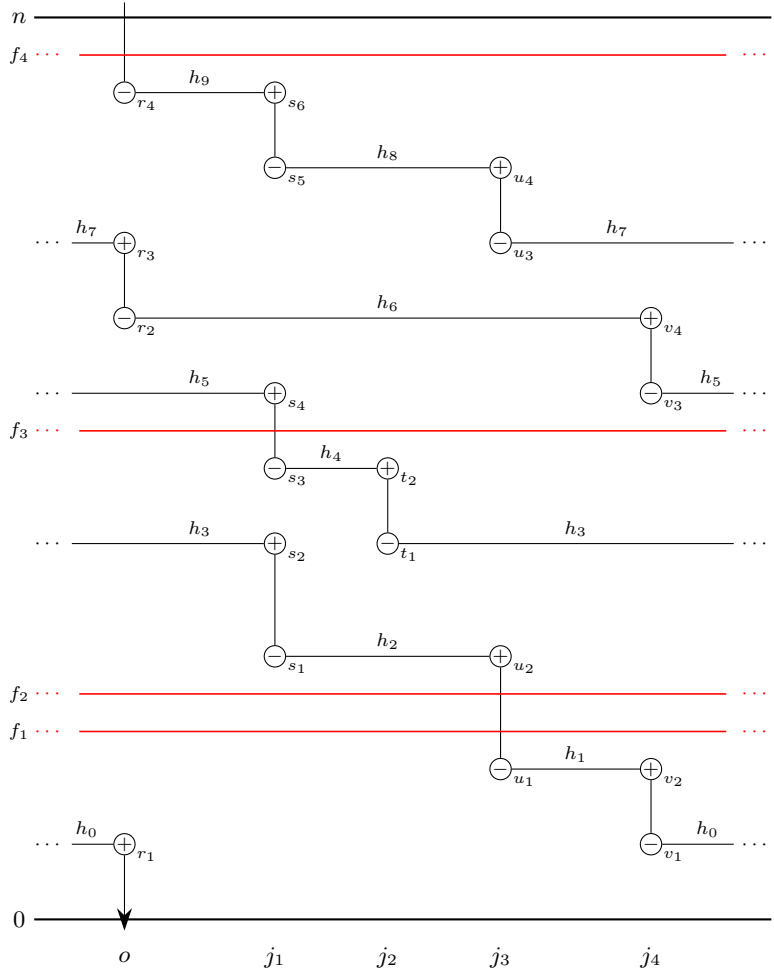
\begin{figure}[H]
  \centering
\begin{tikzpicture}[>=stealth, font=\small]
  \node (v1)  [draw, circle, inner sep=0pt] at (0,10) {\scriptsize$-$};
  \node [anchor=base west, xshift=-1.8pt, yshift=-2.3pt] at (v1.south east) {\tiny$r_4$};

  \node (v2)  [draw, circle, inner sep=0pt] at (2,10) {\scriptsize$+$};
  \node [anchor=base west, xshift=-1.8pt, yshift=-2.3pt] at (v2.south east) {\tiny$s_6$};

  \node (v3)  [draw, circle, inner sep=0pt] at (2,9) {\scriptsize$-$};
  \node [anchor=base west, xshift=-1.8pt, yshift=-2.3pt] at (v3.south east) {\tiny$s_5$};

  \node (v4)  [draw, circle, inner sep=0pt] at (5,9) {\scriptsize$+$};
  \node [anchor=base west, xshift=-1.8pt, yshift=-2.3pt] at (v4.south east) {\tiny$u_4$};

  \node (v5)  [draw, circle, inner sep=0pt] at (5,8) {\scriptsize$-$};
  \node [anchor=base west, xshift=-1.8pt, yshift=-2.3pt] at (v5.south east) {\tiny$u_3$};

  \node (v6)  [draw, circle, inner sep=0pt] at (0,8) {\scriptsize$+$};
  \node [anchor=base west, xshift=-1.8pt, yshift=-2.3pt] at (v6.south east) {\tiny$r_3$};

  \node (v7)  [draw, circle, inner sep=0pt] at (0,7) {\scriptsize$-$};
  \node [anchor=base west, xshift=-1.8pt, yshift=-2.3pt] at (v7.south east) {\tiny$r_2$};

  \node (v8)  [draw, circle, inner sep=0pt] at (7,7) {\scriptsize$+$};
  \node [anchor=base west, xshift=-1.8pt, yshift=-2.3pt] at (v8.south east) {\tiny$v_4$};

  \node (v9)  [draw, circle, inner sep=0pt] at (7,6) {\scriptsize$-$};
  \node [anchor=base west, xshift=-1.8pt, yshift=-2.3pt] at (v9.south east) {\tiny$v_3$};

  \node (v10) [draw, circle, inner sep=0pt] at (2,6) {\scriptsize$+$};
  \node [anchor=base west, xshift=-1.8pt, yshift=-2.3pt] at (v10.south east) {\tiny$s_4$};

  \node (v11) [draw, circle, inner sep=0pt] at (2,5) {\scriptsize$-$};
  \node [anchor=base west, xshift=-1.8pt, yshift=-2.3pt] at (v11.south east) {\tiny$s_3$};

  \node (v12) [draw, circle, inner sep=0pt] at (3.5,5) {\scriptsize$+$};
  \node [anchor=base west, xshift=-1.8pt, yshift=-2.3pt] at (v12.south east) {\tiny$t_2$};

  \node (v13) [draw, circle, inner sep=0pt] at (3.5,4) {\scriptsize$-$};
  \node [anchor=base west, xshift=-1.8pt, yshift=-2.3pt] at (v13.south east) {\tiny$t_1$};

  \node (v14) [draw, circle, inner sep=0pt] at (2,4) {\scriptsize$+$};
  \node [anchor=base west, xshift=-1.8pt, yshift=-2.3pt] at (v14.south east) {\tiny$s_2$};

  \node (v15) [draw, circle, inner sep=0pt] at (2,2.5) {\scriptsize$-$};
  \node [anchor=base west, xshift=-1.8pt, yshift=-2.3pt] at (v15.south east) {\tiny$s_1$};

  \node (v16) [draw, circle, inner sep=0pt] at (5,2.5) {\scriptsize$+$};
  \node [anchor=base west, xshift=-1.8pt, yshift=-2.3pt] at (v16.south east) {\tiny$u_2$};

  \node (v17) [draw, circle, inner sep=0pt] at (5,1) {\scriptsize$-$};
  \node [anchor=base west, xshift=-1.8pt, yshift=-2.3pt] at (v17.south east) {\tiny$u_1$};

  \node (v18) [draw, circle, inner sep=0pt] at (7,1) {\scriptsize$+$};
  \node [anchor=base west, xshift=-1.8pt, yshift=-2.3pt] at (v18.south east) {\tiny$v_2$};

  \node (v19) [draw, circle, inner sep=0pt] at (7,0) {\scriptsize$-$};
  \node [anchor=base west, xshift=-1.8pt, yshift=-2.3pt] at (v19.south east) {\tiny$v_1$};

  \node (v20) [draw, circle, inner sep=0pt] at (0,0) {\scriptsize$+$};
  \node [anchor=base west, xshift=-1.8pt, yshift=-2.3pt] at (v20.south east) {\tiny$r_1$};


    \draw (0,11.2)--(v1);
    \draw (v1) -- (v2) -- (v3) -- (v4) -- (v5);
    \draw (v6) -- (v7) -- (v8) -- (v9);
    \draw (v10) -- (v11) -- (v12) -- (v13);
    \draw (v14) -- (v15) -- (v16) -- (v17) -- (v18) -- (v19);
\draw[-{Stealth[length=3mm, width=2mm]}] (v20)--(0,-1.15);

  \node at ($(v1)!0.5!(v2)+(0,0.2)$) {\scriptsize$h_9$};
  \node at ($(v3)!0.5!(v4)+(0,0.2)$) {\scriptsize$h_8$};
  \node at (-0.5,8.2) {\scriptsize$h_7$};
  \node at ($(v5)!0.5!(8.1,8)+(0,0.2)$) {\scriptsize$h_7$};
  \node at ($(v7)!0.5!(v8)+(0,0.2)$) {\scriptsize$h_6$};
  \node at (1,6.2) {\scriptsize$h_5$};
  \node at ($(v9)!0.5!(8.6,6)+(0,0.2)$) {\scriptsize$h_5$};
  \node at ($(v11)!0.5!(v12)+(0,0.2)$) {\scriptsize$h_4$};
  \node at (1,4.2) {\scriptsize$h_3$};
  \node at ($(v13)!0.5!(8.5,4)+(0,0.2)$) {\scriptsize$h_3$};
  \node at ($(v15)!0.5!(v16)+(0,0.2)$) {\scriptsize$h_2$};
  \node at ($(v17)!0.5!(v18)+(0,0.2)$) {\scriptsize$h_1$};
  \node at (-0.5,0.2) {\scriptsize$h_0$};
  \node at ($(v19)!0.5!(8.5,0)+(0,0.2)$) {\scriptsize$h_0$};

  \draw (v5) -- (8.1,8);
  \draw (v9) -- (8.1,6);
  \draw (v13) -- (8.1,4);
  \draw (v19) -- (8.1,0);
  \draw (v6)  -- (-0.7,8);
  \draw (v10) -- (-0.7,6);
  \draw (v14) -- (-0.7,4);
  \draw (v20) -- (-0.7,0);

  \node at (8.4,8)  {\scriptsize$\cdots$};
  \node at (8.4,6)  {\scriptsize$\cdots$};
  \node at (8.4,4)  {\scriptsize$\cdots$};
  \node at (8.4,0)  {\scriptsize$\cdots$};

  \node at (-1.0,8) {\scriptsize$\cdots$};
  \node at (-1.0,6) {\scriptsize$\cdots$};
  \node at (-1.0,4) {\scriptsize$\cdots$};
  \node at (-1.0,0) {\scriptsize$\cdots$};

  \draw[red, semithick] (-0.6,1.5) -- (8.0,1.5);
  \node[red] at (-1.0,1.5) {\scriptsize$\cdots$};
  \node[red] at (8.4,1.5)  {\scriptsize$\cdots$};
  \node at (-1.4,1.5) {\scriptsize$f_1$};

  \draw[red, semithick] (-0.6,2.0) -- (8.0,2.0);
  \node[red] at (-1.0,2.0) {\scriptsize$\cdots$};
  \node[red] at (8.4,2.0)  {\scriptsize$\cdots$};
  \node at (-1.4,2.0) {\scriptsize$f_2$};

  \draw[red, semithick] (-0.6,5.5) -- (8.0,5.5);
  \node[red] at (-1.0,5.5) {\scriptsize$\cdots$};
  \node[red] at (8.4,5.5)  {\scriptsize$\cdots$};
  \node at (-1.4,5.5) {\scriptsize$f_3$};

  \draw[red, semithick] (-0.6,10.5) -- (8.0,10.5);
  \node[red] at (-1.0,10.5) {\scriptsize$\cdots$};
  \node[red] at (8.4,10.5)  {\scriptsize$\cdots$};
  \node at (-1.4,10.5) {\scriptsize$f_4$};

  \draw[thick] (-1.2,11) -- (8.6,11);
  \node at (-1.4,11) {$n$};

  \draw[thick] (-1.2,-1) -- (8.6,-1);
  \node at (-1.4,-1) {$0$};

  \node at (0, -1.5)  {$o$};
  \node at (2, -1.5)  {$j_1$};
  \node at (3.5, -1.5) {$j_2$};
  \node at (5, -1.5)   {$j_3$};
  \node at (7, -1.5)   {$j_4$};

\end{tikzpicture}
  \caption{ An example of non-minimum carriers inducing the same transition as Table~\ref{tab1}.
  Horizontal black line segments specify the minimum carriers as explained in Remark \ref{re:md}.
  The heights $f_1, f_2, f_3, f_4 \in [0,n]\setminus \{h_0,\ldots, h_9\}$
   of the extra red lines correspond to the supplemented letters to them.}
  \label{fig:ff}
\end{figure}

Let $\vec{\bs} \rightarrow \vec{\bs}'$ 
be the process described in Table \ref{tab1},
where particles of species $h_0, \ldots, h_g$ are moved.
Let ${\bf a}'_1,\ldots, {\bf a}'_L \in \BB^k$ with $k>d$
be the non-minimum carriers obtained by supplementing distinct numbers 
$f_1,\ldots, f_{k-d} \in \{\bar{h}_1,\ldots, \bar{h}_{n-g}\}= 
[0,n]\setminus \{h_0,\ldots, h_g\}$
to the minimum carriers ${\bf a}_1,\ldots, {\bf a}_L \in \BB^d$.
See \eqref{hhbar}.
Set ${\bf a}_j = (a_{j,0},\ldots, a_{j,n})$.
This construction is stated as 
\begin{align}\label{ap}
a_{j,f_1}=\cdots = a_{j,f_{k-d}}=0,
\qquad
{\bf a}'_j = {\bf a}_j + {\bf e}_{f_1}+ \cdots + {\bf e}_{f_{k-d}}
\end{align}
for all sites $j \in \{1,\ldots, L\}$, where ${\bf e}_j$ is defined in \eqref{ej}.

\begin{lemma}\label{le:nmm}
Under the definition \eqref{ap}, 
the product in the LHS of \eqref{mc1} for the 
non-minimum carries ${\bf a}'_1,\ldots, {\bf a}'_L \in \BB^k$  is given by
\begin{align}\label{nmc}
S(0)^{{\bf a}'_2, \bs'_1}_{{\bf a}'_1, \bs_1}
\cdots 
\dot{S}(0)^{{\bf a}'_{o+1}, \bs'_o}_{{\bf a}'_o, \bs_o} 
\cdots 
S(0)^{{\bf a}'_1, \bs'_L}_{{\bf a}'_L,  \bs_L}
= (-1)^{d+1}\prod_{h \in \{h_0, \ldots, h_g\}}\!\!(1 - t^{\sigma_{p(h),h}})\, t^{\ell_h}
\prod_{f\in \{f_1,\ldots, f_{k-d}\}}
t^{\KK_f + \Phi_f},
\end{align}
where $\Phi_f=\Phi_f(\vec{\bs}, \vec{\bs}')$ is defined in 
\eqref{wdef}--\eqref{rr2}.
\end{lemma}
\begin{proof}
We investigate the effect of introducing the red lines as in  Figure \ref{fig:ff} 
on the three types of factors considered individually in the proof of Lemma \ref{le:mc}.
The relevant formulas are \eqref{Sz0}--\eqref{bj1}.

$\bullet$ The factor of the form $\prod(1-t^\#)$.
There is no change in the difference $\vec{\bs}' -  \vec{\bs}$,
hence such factors in \eqref{Sz0} and \eqref{Sz1} neither change.

$\bullet$  The sign factor.
Let us illustrate the effect of red lines on \eqref{alp0} and \eqref{alp1}
along the example in Figure \ref{fig:ff}.
From \eqref{alp1}, we have an extra sign for each red line 
passing through site $o$ in the vertical segment of heights in $[r_1,r_2)$ and $[r_3, r_4)$.
From \eqref{alp0}, we also have an extra sign for each crossing of a red line and a vertical black line.
Obviously, their numbers are even and the  product is $+1$.
From \eqref{sdiag}, diagonal local transitions do not contribute a sign factor even in the presence of 
red lines.

$\bullet$ The power of $t$.
The relevant quantities are \eqref{bet0} and \eqref{bet1}.
In both of them, the first term $\sum_{0 < r < s \le n} j_r a_s$ gives rise to 
$\sum_{j=1}^L \sigma_{j,[0,f)}$ for the red line at height $f$.
This includes contributions from the sites undergoing  diagonal local transitions, and 
equals $\KK_f$ by definition \eqref{Ki}.
The remaining parts in \eqref{bet1} and \eqref{bet0} coincide with 
$\varphi'_f(\bs_o, \bs'_o)$ in \eqref{phip} for site $o$ and 
$\varphi_f(\bs_j, \bs'_j)$ in \eqref{phi} for the other site $j \neq o$, respectively.
Therefore their sum yield $\Phi_f(\vec{\bs}, \vec{\bs}')$ in \eqref{wdef}.
\end{proof}

Lemma \ref{le:nmm} includes Lemma \ref{le:mc} as the special case $k=d$.

\subsection{Completion of proof of Theorem \ref{th:main1}}\label{ss:proof1}
According to \eqref{HHj2}, we first take the sum of \eqref{nmc}
over all the not necessarily minimum carriers 
${\bf a}'_1,\ldots, {\bf a}'_{k-d} \in \BB^k$ for fixed $k \in [d,  n+1]$.
This sum is translated to the sum over all the subsets $\{f_1,\ldots, f_{k-d} \} \subseteq 
\{\bar{h}_1, \ldots, \bar{h}_{n-g}\}$ in \eqref{hhbar}.
By the definition \eqref{ekb},  the result is expressed in terms of the elementary symmetric function as
\begin{align}
(-1)^{d+1}\prod_{h \in \{h_0, \ldots, h_g\}}\!\!(1 - t^{\sigma_{p(h),h}})\, t^{\ell_h}\times 
e_{k-d}(t^{\KK_{\bar{h}_1}+\Phi_{\bar{h}_1}},
\ldots, 
t^{\KK_{\bar{h}_{n-g}}+\Phi_{\bar{h}_{n-g}}}).
\end{align}
Further multiplying $(-1)^{k-1}$ and taking the sum over $k \in [d,n+1]$, we obtain
\begin{equation}
\begin{split}
&\prod_{h \in \{h_0, \ldots, h_g\}}\!\!(1 - t^{\sigma_{p(h),h}})\, t^{\ell_h}
\sum_{k=d}^{n+1}(-1)^{k-d}
e_{k-d}(t^{\KK_{\bar{h}_1}+\Phi_{\bar{h}_1}},\ldots, 
t^{\KK_{\bar{h}_{n-g}}+\Phi_{\bar{h}_{n-g}}})
\\
&= \prod_{h \in \{h_0, \ldots, h_g\}}\!\!(1 - t^{\sigma_{p(h),h}})\, t^{\ell_h}
\prod_{h \in \{\bar{h}_1, \ldots, \bar{h}_{n-g}\}}
(1 - t^{\KK_h + \Phi_h}).
\end{split}
\end{equation}
In the last step, we have used \eqref{ekg} together with the condition
$d \le g$ mentioned at the end of Remark~\ref{re:md}.
This establishes \eqref{HHj2}.
In particular, it yields zero for unwanted~$\bs'$
(cf.~Example~\ref{ex:d}) as a consequence of
Proposition~\ref{pr:zr}.
This completes the proof of Theorem~\ref{th:main1}.

\section{Zamolodchikov-Faddeev algebra}\label{sec:zf}

\subsection{Strange five-vertex model}

Consider a two-dimensional vertex model whose local edge states take values in 
$\{0,1\}$.  
Among the sixteen possible configurations around a vertex, 
we assign $t$-oscillator valued  ``Boltzmann weights'' to the following five, 
and zero to the remaining eleven.\footnote{The configurations and associated weights in~\eqref{s5V} 
are used solely for defining the matrix product operators 
\(A_\alpha(z)\) in this section, and should not be confused 
with those appearing elsewhere in the context of 
\(R\)-matrices or transfer matrices.}
\begin{equation}
\newcommand{\fvm}[4]{
\begin{tikzpicture}[scale=.35,>=latex]
\draw[->] (-1,0) node[anchor=east]{$#1$} -- (1,0) node[anchor=west]{$#2$};
\draw[->] (0,-1) node[anchor=north]{$#3$} -- (0,1) node[anchor=south]{$#4$};
\end{tikzpicture}
}%
\begin{array}{*{5}{@{\hspace{20pt}}c}}
\fvm{1}{1}{1}{1} & \fvm{0}{0}{0}{1} & \fvm{1}{1}{0}{0} & \fvm{1}{0}{0}{1} & \fvm{0}{1}{1}{1} \\
1 & 1 & \kk & \aaa^- & \aaa^+
\end{array}
\label{s5V}
\end{equation}
Here the $t$-oscillator operators $\kk$, $\aaa^+$, and $\aaa^-$ act on the Fock space 
$\FF = \bigoplus_{m \in \Z_{\ge 0}}\C |m\rangle$ as
\begin{align}\label{tom}
\kk |m \rangle = t^m |m\rangle,
\quad
\aaa^+ |m \rangle = |m+1 \rangle,
\quad
\aaa^- |m \rangle = (1 - t^m)\,|m-1 \rangle.
\end{align}
It is a representation of the $t$-oscillator algebra 
\begin{align}\label{apc}
\ok\, \ap = t \ap \ok, \quad \ok\, \am = t^{-1} \am \ok,
\quad
\ap\am = 1-\ok,\quad \am \ap = 1-t\ok.
\end{align}
The model was introduced in~\cite[Eq.~(54)]{KOS24}, 
where the indices $0$ and $1$ are interchanged compared with the present convention.
Although the vertices are strange in that the usual arrow conservation 
does not hold, 
it nevertheless admits a natural interpretation in terms of the multiline-queue 
construction of stationary states~\cite{KOS24}. See also \cite{CMW22}.
One may also regard the diagrams in~\eqref{s5V} as being equipped with an 
additional arrow perpendicular to the vertex, representing the $t$-oscillator 
acting on the Fock space.

\subsection{Quantized CTM $A_0(z), \ldots,  A_n(z)$}

We introduce the operators $A_0(z), \ldots, A_n(z)$.
For $n=1$, we set 
\begin{align}\label{an1}
A_0(z)=1, \quad A_1(z) = z.
\end{align}
For $n \ge 2$, they are defined as the partition function of the strange five-vertex model 
in a triangular region as follows:
\begin{equation}\label{ctm1}
A_{i}(z) = 
\begin{tikzpicture}[scale=0.8,>=latex,baseline=-2.5cm,rotate=270]
\draw[->, rounded corners] (3.5,-1) -- (1,-1) -- (1,-0.2)
  node[anchor=west] {$\bar{\delta}_{i, n-1}$};
\draw (3.8,-1) node {$\vdots$};
\draw[-] (4.5,-1) -- (5.8,-1);
\fill[black] (5.45,-1) circle (.1) node[anchor=west, scale=.8] {$z$};

\draw[->, rounded corners] (3.5,-2) -- (2,-2) -- (2,-0.2)
  node[anchor=west] {$\bar{\delta}_{i, n-2}$};
\draw (3.8,-2) node {$\vdots$};
\draw[-] (4.5,-2) -- (5.8,-2);
\fill[black] (5.45,-2) circle (.1) node[anchor=west, scale=.8] {$z$};

\draw[->, rounded corners] (3.5,-3) -- (3,-3) -- (3,-0.2)
  node[anchor=west] {$\bar{\delta}_{i, n-3}$};
\draw (3.8,-3) node {$\vdots$};
\draw[-] (4.5,-3) -- (5.8,-3);
\fill[black] (5.45,-3) circle (.1) node[anchor=west, scale=.8] {$z$};

\draw[->, rounded corners] (5.8,-5) -- (5,-5) -- (5,-0.2)
  node[anchor=west] {$\bar{\delta}_{i,0}$};
\fill[black] (5.45,-5) circle (.1) node[anchor=west, scale=.8] {$z$};

\draw (4,-4) node {\rotatebox{90}{$\ddots$}};
\draw (5.5,-3.8) node {$\cdots$};
\draw (3.8,-0) node {$\vdots$};
\end{tikzpicture}
\end{equation}
where $\bar{\delta}_{i,j} = 1-\delta_{i,j}$ and 
the symbol $\bullet z$ means $1$ or $z$ according as the corresponding 
edge variable is $0$ or $1$, respectively.
The diagram indicates that all edge variables are summed over $\{0,1\}$.  
The variables on the right outgoing edges are fixed as shown, 
while no boundary condition is imposed along the bottom.  
Different vertices carry independent copies of the $t$-oscillators, 
so that $A_i(z)$ acts on $\FF^{\otimes n(n-1)/2}$.  

We number the vertices in~\eqref{ctm1} from top to bottom in the rightmost column, 
then proceed similarly in the next column to the left, and so on, 
assigning labels $1,\ldots,n(n-1)/2$ in this order.  
The $t$-oscillator generators acting on the $i$th copy of the Fock space~$\FF$ 
under this numbering will be distinguished by the subscript~$i$ when $n \ge 3$;  
generators with different subscripts commute.  

The diagram \eqref{ctm1} is formally similar to \cite[Fig.~13.1(b)]{Bax83};
however, it is an {\em operator} acting along the ``third'' direction
perpendicular to the sheet.
In this sense, we refer to \eqref{ctm1} as a 
{\em (quantized) corner transfer matrix} (CTM).
It is related to the operator $X_i(z)$ introduced in~\cite[Def.~15]{KOS24} by
\begin{align}\label{kosX}
A_i(z) = z^n X_{n-i}(z^{-1}), \qquad i = 0,1,\ldots,n.
\end{align}
In particular, the definition~\eqref{an1} originates from~\cite[Eq.~(84)]{KOS24}.
The CTM $A_i(z)$ is a basic constituent of the matrix product formulas for
stationary probabilities of multispecies ASEP \cite{KOS24} and $t$-PushTASEP \cite{AK25}
in the capacity-$1$ case.
Various versions of such operators have appeared in the literature; see, for example,
\cite{CDW15,PEM09}.
The formulation adopted here is the simplest one devised to date,
requiring only a {\em colorless} (i.e., two-state) model.

\begin{example}\label{ex:akos}
The following examples are deduced from \cite[Ex.18]{KOS24} and \eqref{kosX}.
For $n=2$, the operators $A_0(z), A_1(z), A_2(z)$ are given by
\begin{align*}
\raisebox{2pt}{$A_2(z)$} & \begin{array}{@{\;=}c@{\;+\;}c} \begin{tikzpicture}[scale=.65,>=latex,rounded corners,baseline=-1cm,rotate=-90]
\draw[->] (1.7,0) node[anchor=north] {$1$} -- (0,0) -- (0,0.7) node[anchor=west] {$1$};
\draw[->] (1.7,-1) node[anchor=north] {$1$} -- (1,-1) -- (1,0.7) node[anchor=west] {$1$};
\end{tikzpicture}z^2
& \begin{tikzpicture}[scale=.65,>=latex,rounded corners,baseline=-1cm,rotate=-90]
\draw[->] (1.7,0) node[anchor=north] {$1$} -- (0,0) -- (0,0.7) node[anchor=west] {$1$};
\draw[->] (1.7,-1) node[anchor=north] {$0$} -- (1,-1) -- (1,0.7) node[anchor=west] {$1$};
\end{tikzpicture}
z
\\
z^2 & z \aaa^+,
\end{array}
&
\raisebox{2pt}{$A_1(z)$} & \begin{array}{@{\;=}c} \begin{tikzpicture}[scale=.65,>=latex,rounded corners,baseline=-1cm,rotate=-90]
\draw[->] (1.7,0) node[anchor=north] {$0$} -- (0,0) -- (0,0.7) node[anchor=west] {$0$};
\draw[->] (1.7,-1) node[anchor=north] {$1$} -- (1,-1) -- (1,0.7) node[anchor=west] {$1$};
\end{tikzpicture}
z
\\
z \kk,
\end{array}
&
\raisebox{2pt}{$A_0(z)$} & \begin{array}{@{\;=}c@{\;+\;}c} \begin{tikzpicture}[scale=.65,>=latex,rounded corners,baseline=-1cm,rotate=-90]
\draw[->] (1.7,0) node[anchor=north] {$0$} -- (0,0) -- (0,0.7) node[anchor=west] {$1$};
\draw[->] (1.7,-1) node[anchor=north] {$1$} -- (1,-1) -- (1,0.7) node[anchor=west] {$0$};
\end{tikzpicture}
z
& \begin{tikzpicture}[scale=.65,>=latex,rounded corners,baseline=-1cm,rotate=-90]
\draw[->] (1.7,0) node[anchor=north] {$0$} -- (0,0) -- (0,0.7) node[anchor=west] {$1$};
\draw[->] (1.7,-1) node[anchor=north] {$0$} -- (1,-1) -- (1,0.7) node[anchor=west] {$0$};
\end{tikzpicture}
\\
z \aaa^- & 1.
\end{array}
\end{align*}

For $n=3$, the operators $A_0(z), A_1(z), A_2(z),A_3(z)$ are given by
\begin{align*}
\raisebox{9pt}{$A_3(z)$} & \begin{array}{@{\;=}c@{\;+\;}c@{\;+\;}c@{\;+\;}c@{\;+\;}c} \begin{tikzpicture}[scale=.6,>=latex,rounded corners,baseline=-1cm,rotate=-90]
\draw[->] (2.7,0) node[anchor=north] {$1$} -- (0,0) -- (0,0.7) node[anchor=west] {$1$};
\draw[->] (2.7,-1) node[anchor=north] {$1$} -- (1,-1) -- (1,0.7) node[anchor=west] {$1$};
\draw[->] (2.7,-2) node[anchor=north] {$1$} -- (2,-2) -- (2,0.7) node[anchor=west] {$1$};
\draw (1.5,0) node[anchor=west,scale=.7] {$1$};
\draw (1.5,-1) node[anchor=west,scale=.7] {$1$};
\draw (2,-.5) node[anchor=north,scale=.7] {$1$};
\end{tikzpicture}z^3
&
\begin{tikzpicture}[scale=.6,>=latex,rounded corners,baseline=-1cm,rotate=-90]
\draw[->] (2.7,0) node[anchor=north] {$1$} -- (0,0) -- (0,0.7) node[anchor=west] {$1$};
\draw[->] (2.7,-1) node[anchor=north] {$0$} -- (1,-1) -- (1,0.7) node[anchor=west] {$1$};
\draw[->] (2.7,-2) node[anchor=north] {$1$} -- (2,-2) -- (2,0.7) node[anchor=west] {$1$};
\draw (1.5,0) node[anchor=west,scale=.7] {$1$};
\draw (1.5,-1) node[anchor=west,scale=.7] {$0$};
\draw (2,-.5) node[anchor=north,scale=.7] {$1$};
\end{tikzpicture}
z^2
&
\begin{tikzpicture}[scale=.6,>=latex,rounded corners,baseline=-1cm,rotate=-90]
\draw[->] (2.7,0) node[anchor=north] {$1$} -- (0,0) -- (0,0.7) node[anchor=west] {$1$};
\draw[->] (2.7,-1) node[anchor=north] {$0$} -- (1,-1) -- (1,0.7) node[anchor=west] {$1$};
\draw[->] (2.7,-2) node[anchor=north] {$1$} -- (2,-2) -- (2,0.7) node[anchor=west] {$1$};
\draw (1.5,0) node[anchor=west,scale=.7] {$1$};
\draw (1.5,-1) node[anchor=west,scale=.7] {$1$};
\draw (2,-.5) node[anchor=north,scale=.7] {$0$};
\end{tikzpicture}
z^2
&
\begin{tikzpicture}[scale=.6,>=latex,rounded corners,baseline=-1cm,rotate=-90]
\draw[->] (2.7,0) node[anchor=north] {$1$} -- (0,0) -- (0,0.7) node[anchor=west] {$1$};
\draw[->] (2.7,-1) node[anchor=north] {$1$} -- (1,-1) -- (1,0.7) node[anchor=west] {$1$};
\draw[->] (2.7,-2) node[anchor=north] {$0$} -- (2,-2) -- (2,0.7) node[anchor=west] {$1$};
\draw (1.5,0) node[anchor=west,scale=.7] {$1$};
\draw (1.5,-1) node[anchor=west,scale=.7] {$1$};
\draw (2,-.5) node[anchor=north,scale=.7] {$1$};
\end{tikzpicture}
z^2
&
\begin{tikzpicture}[scale=.6,>=latex,rounded corners,baseline=-1cm,rotate=-90]
\draw[->] (2.7,0) node[anchor=north] {$1$} -- (0,0) -- (0,0.7) node[anchor=west] {$1$};
\draw[->] (2.7,-1) node[anchor=north] {$0$} -- (1,-1) -- (1,0.7) node[anchor=west] {$1$};
\draw[->] (2.7,-2) node[anchor=north] {$0$} -- (2,-2) -- (2,0.7) node[anchor=west] {$1$};
\draw (1.5,0) node[anchor=west,scale=.7] {$1$};
\draw (1.5,-1) node[anchor=west,scale=.7] {$1$};
\draw (2,-.5) node[anchor=north,scale=.7] {$0$};
\end{tikzpicture}
z
\\
z^3 & z^2 \aaa_1^+ \kk_3 & z^2 \aaa_2^+ \aaa_3^- & z^2 \aaa_3^+ & z \aaa_2^+,
\end{array}
\allowdisplaybreaks \\
\raisebox{9pt}{$A_2(z)$} & \begin{array}{@{\;=}c@{\;+\;}c}  \begin{tikzpicture}[scale=.6,>=latex,rounded corners,baseline=-1cm,rotate=-90]
\draw[->] (2.7,0) node[anchor=north] {$0$} -- (0,0) -- (0,0.7) node[anchor=west] {$0$};
\draw[->] (2.7,-1) node[anchor=north] {$1$} -- (1,-1) -- (1,0.7) node[anchor=west] {$1$};
\draw[->] (2.7,-2) node[anchor=north] {$1$} -- (2,-2) -- (2,0.7) node[anchor=west] {$1$};
\draw (1.5,0) node[anchor=west,scale=.7] {$0$};
\draw (1.5,-1) node[anchor=west,scale=.7] {$1$};
\draw (2,-.5) node[anchor=north,scale=.7] {$1$};
\end{tikzpicture}
z^2
&
\begin{tikzpicture}[scale=.6,>=latex,rounded corners,baseline=-1cm,rotate=-90]
\draw[->] (2.7,0) node[anchor=north] {$0$} -- (0,0) -- (0,0.7) node[anchor=west] {$0$};
\draw[->] (2.7,-1) node[anchor=north] {$1$} -- (1,-1) -- (1,0.7) node[anchor=west] {$1$};
\draw[->] (2.7,-2) node[anchor=north] {$0$} -- (2,-2) -- (2,0.7) node[anchor=west] {$1$};
\draw (1.5,0) node[anchor=west,scale=.7] {$0$};
\draw (1.5,-1) node[anchor=west,scale=.7] {$1$};
\draw (2,-.5) node[anchor=north,scale=.7] {$1$};
\end{tikzpicture}
z
\\
z^2 \kk_1 \kk_2 & z \kk_1 \kk_2 \aaa_3^+,
\end{array}
\allowdisplaybreaks \\
\raisebox{9pt}{$A_1(z)$} & \begin{array}{@{\;=}c@{\;+\;}c@{\;+\;}c} \begin{tikzpicture}[scale=.6,>=latex,rounded corners,baseline=-1cm,rotate=-90]
\draw[->] (2.7,0) node[anchor=north] {$0$} -- (0,0) -- (0,0.7) node[anchor=west] {$1$};
\draw[->] (2.7,-1) node[anchor=north] {$1$} -- (1,-1) -- (1,0.7) node[anchor=west] {$0$};
\draw[->] (2.7,-2) node[anchor=north] {$1$} -- (2,-2) -- (2,0.7) node[anchor=west] {$1$};
\draw (1.5,0) node[anchor=west,scale=.7] {$0$};
\draw (1.5,-1) node[anchor=west,scale=.7] {$1$};
\draw (2,-.5) node[anchor=north,scale=.7] {$1$};
\end{tikzpicture}
z^2
&
\begin{tikzpicture}[scale=.6,>=latex,rounded corners,baseline=-1cm,rotate=-90]
\draw[->] (2.7,0) node[anchor=north] {$0$} -- (0,0) -- (0,0.7) node[anchor=west] {$1$};
\draw[->] (2.7,-1) node[anchor=north] {$1$} -- (1,-1) -- (1,0.7) node[anchor=west] {$0$};
\draw[->] (2.7,-2) node[anchor=north] {$0$} -- (2,-2) -- (2,0.7) node[anchor=west] {$1$};
\draw (1.5,0) node[anchor=west,scale=.7] {$0$};
\draw (1.5,-1) node[anchor=west,scale=.7] {$1$};
\draw (2,-.5) node[anchor=north,scale=.7] {$1$};
\end{tikzpicture}
z
&
\begin{tikzpicture}[scale=.6,>=latex,rounded corners,baseline=-1cm,rotate=-90]
\draw[->] (2.7,0) node[anchor=north] {$0$} -- (0,0) -- (0,0.7) node[anchor=west] {$1$};
\draw[->] (2.7,-1) node[anchor=north] {$0$} -- (1,-1) -- (1,0.7) node[anchor=west] {$0$};
\draw[->] (2.7,-2) node[anchor=north] {$1$} -- (2,-2) -- (2,0.7) node[anchor=west] {$1$};
\draw (1.5,0) node[anchor=west,scale=.7] {$0$};
\draw (1.5,-1) node[anchor=west,scale=.7] {$0$};
\draw (2,-.5) node[anchor=north,scale=.7] {$1$};
\end{tikzpicture}
z
\\
z^2 \aaa_1^- \kk_2 & z \aaa_1^- \kk_2 \aaa_3^+ & z \kk_2 \kk_3,
\end{array}
\allowdisplaybreaks \\
\raisebox{9pt}{$A_0(z)$} & \begin{array}{@{\;=}c@{\;+\;}c@{\;+\;}c@{\;+\;}c@{\;+\;}c} \begin{tikzpicture}[scale=.6,>=latex,rounded corners,baseline=-1cm,rotate=-90]
\draw[->] (2.7,0) node[anchor=north] {$0$} -- (0,0) -- (0,0.7) node[anchor=west] {$1$};
\draw[->] (2.7,-1) node[anchor=north] {$1$} -- (1,-1) -- (1,0.7) node[anchor=west] {$1$};
\draw[->] (2.7,-2) node[anchor=north] {$1$} -- (2,-2) -- (2,0.7) node[anchor=west] {$0$};
\draw (1.5,0) node[anchor=west,scale=.7] {$1$};
\draw (1.5,-1) node[anchor=west,scale=.7] {$1$};
\draw (2,-.5) node[anchor=north,scale=.7] {$1$};
\end{tikzpicture}
z^2
&
\begin{tikzpicture}[scale=.6,>=latex,rounded corners,baseline=-1cm,rotate=-90]
\draw[->] (2.7,0) node[anchor=north] {$0$} -- (0,0) -- (0,0.7) node[anchor=west] {$1$};
\draw[->] (2.7,-1) node[anchor=north] {$1$} -- (1,-1) -- (1,0.7) node[anchor=west] {$1$};
\draw[->] (2.7,-2) node[anchor=north] {$0$} -- (2,-2) -- (2,0.7) node[anchor=west] {$0$};
\draw (1.5,0) node[anchor=west,scale=.7] {$1$};
\draw (1.5,-1) node[anchor=west,scale=.7] {$1$};
\draw (2,-.5) node[anchor=north,scale=.7] {$1$};
\end{tikzpicture}
z
&
\begin{tikzpicture}[scale=.6,>=latex,rounded corners,baseline=-1cm,rotate=-90]
\draw[->] (2.7,0) node[anchor=north] {$0$} -- (0,0) -- (0,0.7) node[anchor=west] {$1$};
\draw[->] (2.7,-1) node[anchor=north] {$0$} -- (1,-1) -- (1,0.7) node[anchor=west] {$1$};
\draw[->] (2.7,-2) node[anchor=north] {$1$} -- (2,-2) -- (2,0.7) node[anchor=west] {$0$};
\draw (1.5,0) node[anchor=west,scale=.7] {$1$};
\draw (1.5,-1) node[anchor=west,scale=.7] {$1$};
\draw (2,-.5) node[anchor=north,scale=.7] {$0$};
\end{tikzpicture}
z&
\begin{tikzpicture}[scale=.6,>=latex,rounded corners,baseline=-1cm,rotate=-90]
\draw[->] (2.7,0) node[anchor=north] {$0$} -- (0,0) -- (0,0.7) node[anchor=west] {$1$};
\draw[->] (2.7,-1) node[anchor=north] {$0$} -- (1,-1) -- (1,0.7) node[anchor=west] {$1$};
\draw[->] (2.7,-2) node[anchor=north] {$1$} -- (2,-2) -- (2,0.7) node[anchor=west] {$0$};
\draw (1.5,0) node[anchor=west,scale=.7] {$1$};
\draw (1.5,-1) node[anchor=west,scale=.7] {$0$};
\draw (2,-.5) node[anchor=north,scale=.7] {$1$};
\end{tikzpicture}
z
&
\begin{tikzpicture}[scale=.6,>=latex,rounded corners,baseline=-1cm,rotate=-90]
\draw[->] (2.7,0) node[anchor=north] {$0$} -- (0,0) -- (0,0.7) node[anchor=west] {$1$};
\draw[->] (2.7,-1) node[anchor=north] {$0$} -- (1,-1) -- (1,0.7) node[anchor=west] {$1$};
\draw[->] (2.7,-2) node[anchor=north] {$0$} -- (2,-2) -- (2,0.7) node[anchor=west] {$0$};
\draw (1.5,0) node[anchor=west,scale=.7] {$1$};
\draw (1.5,-1) node[anchor=west,scale=.7] {$1$};
\draw (2,-.5) node[anchor=north,scale=.7] {$0$};
\end{tikzpicture}
\\
z^2 \aaa_2^- & z \aaa_2^- \aaa_3^+ & z \aaa_3^- & z \aaa_1^+ \aaa_2^- \kk_3 & 1.
\end{array}
\end{align*}
\end{example}

\subsection{ZF algebra with structure function $\s_{1,1}(z)$}\label{ss:zfb}

Let us introduce the normalized version 
$\s(z) = \s_{1,1}(z)= S^{1}_{\;\,1}(z)/(1-tz)$ of Example \ref{ex:s11},
whose nonzero elements are given by 
\begin{align}\label{ns}
\s(z)^{{\bf e}_i, {\bf e}_i}_{{\bf e}_i, {\bf e}_i} = 1,
\quad
\s(z)^{{\bf e}_i, {\bf e}_j}_{{\bf e}_i, {\bf e}_j} = \frac{t^{\theta(i>j)}(1-z)}{1-tz},
\quad
\s(z)^{{\bf e}_i, {\bf e}_j}_{{\bf e}_j, {\bf e}_i} = \frac{z^{\theta(i>j)}(1-t)}{1-tz}
\quad (0 \le i \neq j \le n).
\end{align}
They satisfy 
\begin{align}
\sum_{0 \le a,b \le n}\s(z)^{{\bf e}_a, {\bf e}_b}_{{\bf e}_i, {\bf e}_j}&=1
\qquad (0 \le i,j \le n),
\label{stu}
\\
\s(1)^{{\bf e}_a, {\bf e}_b}_{{\bf e}_i, {\bf e}_j} &= \delta^a_j\delta^b_i,
\label{sep}
\end{align}
where the first one is called sum-to-unity condition.
 
The most fundamental property of the CTM operators 
$A_0(z),\ldots,A_n(z)$ is that they satisfy the 
Zamolod\-chikov--Faddeev (ZF) algebra, whose structure function is given by 
the basic stochastic $R$-matrix~\eqref{ns}:
\begin{proposition}\cite[Th.\,28]{KOS24}
\label{pr:zf1}
\begin{equation}\label{zf1}
A_b(x)\,A_a(y)
= \sum_{0 \le i,j \le n}
\s\!\left(\frac{y}{x}\right)^{{\bf e}_a,{\bf e}_b}_{{\bf e}_i,{\bf e}_j}\,
A_i(y)\,A_j(x)
\qquad (0 \le a,b \le n).
\end{equation}
\end{proposition}

\subsection{ZF algebra with structure function $\s_{k,l}(z)$}\label{ss:zfg}

Let us generalize the result of the previous subsection to the case
where the structure function is replaced by $\s_{k,l}(z)$, obtained
through the symmetric fusion of the fundamental one
$\s_{1,1}(z)$, as described in Appendix~\ref{app:rs}.
We begin by introducing the operator $A_{\mathbf i}(z)$ labeled by a 
semistandard tableau 
$
\mathbf{i} = (i_1,\ldots,i_l) \in \TT_l,
$
of row shape and length $l$, whose entries range in $\{0,\ldots,n\}$; 
see~\eqref{tl}.

Given a tableau $\mathbf i = (i_1,\ldots,i_l) \in \TT_l$, let 
$\mathcal C(\mathbf i)$ denote the set of its {\em distinct} permutations.
For example, when $l=3$ one has
\begin{equation}\label{Cex}
\begin{split}
\mathcal{C}\bigl((0,1,4)\bigr) &= 
\{(0,1,4),(0,4,1),(1,0,4),(1,4,0),(4,0,1),(4,1,0)\},\\
\mathcal{C}\bigl((2,3,3)\bigr) &= 
\{(2,3,3),(3,2,3),(3,3,2)\},\\
\mathcal{C}\bigl((5,5,5)\bigr) &= 
\{(5,5,5)\}.
\end{split}
\end{equation}

The quantized CTM $A_{\mathbf i}(z)$ associated with 
$\mathbf i \in \TT_l$ is defined by a ``fusion" of the basic CTMs as
\begin{align}\label{Abi}
A_{\mathbf i}(z)
= \sum_{{\mathbf i}' \in \mathcal{C}(\mathbf i)}
A_{i'_1}(t^{l-1}z)\,
A_{i'_2}(t^{l-2}z)\cdots
A_{i'_l}(z)
\qquad ({\mathbf i} \in \TT_l), 
\end{align}
where the sum is taken over 
${\mathbf i}'=(i'_1,\ldots, i'_l) \in \mathcal{C}({\mathbf i})$.
Formally, this is a $t$-deformed monomial symmetric polynomial in 
$A_0(z),\ldots,A_n(z)$.  It acts linearly on the same space
$\FF^{\otimes n(n-1)/2}$ as in the $l=1$ case.

\begin{example}\label{ex:As}
Consider first the case $n=1$.  
From~\eqref{an1} we obtain the scalar operators
\begin{align*}
&A_{00}(z)=A_0(tz)A_0(z)=1,\qquad
A_{01}(z)=A_0(tz)A_1(z)+A_1(tz)A_0(z)=(1+t)z,\\
&A_{11}(z)=A_1(tz)A_1(z)=tz^2,\\[2pt]
&A_{000}(z)=1,\quad
A_{001}(z)=(1+t+t^2)z,\quad
A_{011}(z)=t(1+t+t^2)z^2,\quad
A_{111}(z)=t^3z^3.
\end{align*}
In general one easily verifies that
\begin{align}\label{Asn1}
A_{\mathbf s}(z)
= z^r\, t^{r(r-1)/2}\,
\qbinom l{r}_t 
\qquad
\text{for}\quad
\mathbf s =
(\underbrace{0,\ldots,0}_{l-r},
\underbrace{1,\ldots,1}_{r})
\in \TT_l,
\end{align}
where
$\qbinom l{r}_t = \frac{(t)_l}{(t)_r (t)_{l-r}}$
with 
$(t)_r = \prod_{j=1}^r (1-t^j)$ is the $t$-binomial coefficient.
\end{example}

\begin{example}\label{ex:A22}
Next consider the case $n=2$ and $l=2$.
Applying~\eqref{Abi} to Example~\ref{ex:akos}, we obtain
\begin{align*}
A_{22}(z)
&= t^2 z^4 + t(1+t)z^3 \ap + t z^2 \ap{}^2,\\
A_{12}(z)
&= t(1+t)z^3 \ok + t(1+t)z^2 \ap \ok,\\
A_{11}(z)
&= t z^2 \ok{}^2,\\
A_{02}(z)
&= (1+t)^2 z^2 + t(1+t)z^3 \am + (1+t)z \ap - t(1+t)z^2 \ok,\\
A_{01}(z)
&= (1+t)z \ok + (1+t)z^2 \am \ok,\\
A_{00}(z)
&= 1 + (1+t)z \am + t z^2 \am{}^2.
\end{align*}
In relation to Remark \ref{re:pos} given later,
we note that $A_{02}(z)$ is also expressible without a minus sign as
$A_{02}(z) = (1+t)z^2 + t(1+t)z^2\ap \am + t(1+t)z^3 \am + (1+t)z \ap$
by using \eqref{apc}.
\end{example}

\begin{remark}\label{re:asym}
Adding relation~\eqref{zf1} to its counterpart with $a$ and $b$ interchanged
and using~\eqref{stu}, we obtain
\begin{align}
A_b(x)A_a(y)+A_a(x)A_b(y)
= A_b(y)A_a(x)+A_a(y)A_b(x)
\end{align}
for any $0\le a,b\le n$. This is regarded as invariance under the exchange of $x$ and $y$.
By repeatedly applying it to neighboring spectral parameters, we find that
for any $\mathbf i\in\TT_l$,
\begin{align}\label{Asym}
\sum_{{\mathbf i}' \in\mathcal{C}(\mathbf i)}
A_{i'_1}(z_1)A_{i'_2}(z_2)\cdots A_{i'_l}(z_l)
\quad\text{is invariant under any permutation of }z_1,\ldots,z_l.
\end{align}
Thus the spectral parameters in~\eqref{Abi} may actually be reordered arbitrarily.
The ordering chosen there is the one suitable for the proof of the following
Theorem~\ref{th:SAA}.
\end{remark}

\begin{theorem}\label{th:SAA}
The operators $A_{\mathbf i}(z)$ defined in \eqref{Abi}
satisfy the Zamolodchikov-Faddeev algebra
\begin{align}\label{saa}
A_{\mathbf b}(x) A_{\mathbf a}(y) = \sum_{{\mathbf i} \in \TT_k, {\mathbf j} \in \TT_l}
\s_{k,l}\left(\frac{y}{x}\right)^{\mathbf{a}, \mathbf{b}}_{\mathbf{i}, \mathbf{j}}
A_{\mathbf i}(y)A_{\mathbf j}(x)\qquad
(\mathbf{a} \in \TT_k, \, \mathbf{b} \in \TT_l)
\end{align}
for any $k, l \in \Z_{\ge 1}$, 
where the structure function $\s_{k,l}(z)^{\mathbf{a}, \mathbf{b}}_{\mathbf{i}, \mathbf{j}}$
is given in Appendix \ref{app:rs}.
\end{theorem}
\begin{proof}
We illustrate the proof for the case $(k,l) = (3,2)$, 
for which the structure function is depicted in \eqref{f32}.
The general case is completely analogous.
For brevity, we denote 
$\s(z)^{{\bf e}_a,{\bf e}_b}_{{\bf e}_i,{\bf e}_j} = 
\s_{1,1}(z)^{{\bf e}_a,{\bf e}_b}_{{\bf e}_i,{\bf e}_j}$, which appears in \eqref{ns}, 
\eqref{zf1} and \eqref{sdg}, 
simply by $\s(z)^{a,b}_{i,j}$.
Substituting 
\begin{align*} 
A_{\mathbf b}(x) = \sum_{{\mathbf b}' \in \mathcal{C}({\mathbf b})}A_{b'_1}(tx)A_{b'_2}(x),
\quad
A_{\mathbf a}(x) = \sum_{{\mathbf a}' \in \mathcal{C}({\mathbf a})}
A_{a'_1}(t^2y)A_{a'_2}(ty)A_{a'_3}(y)
\end{align*}
into the LHS of \eqref{saa}, we obtain
\begin{align}\label{a40}
\sum_{{\mathbf a}', {\mathbf b}'}
A_{b'_1}(tx)A_{b'_2}(x)A_{a'_1}(t^2y)A_{a'_2}(ty)A_{a'_3}(y).
\end{align}
We first move $A_{b'_2}(x)$ to the right through
$A_{a'_1}(t^2 y) A_{a'_2}(t y) A_{a'_3}(y)$ by successive use of \eqref{zf1}:
\begin{align}\label{a41}
A_{b'_2}(x)A_{a'_1}(t^2y)A_{a'_2}(ty)A_{a'_3}(y) = 
\sum_{i_1',i_2',i_3', j'_2, j''_2, j'''_2}
\s(zt^2)^{a'_1, b'_2}_{i'_1, j'_2}\s(zt)^{a'_2, j'_2}_{i'_2,j''_2}\s(z)^{a'_3, j''_2}_{i'_3, j'''_2}
A_{i'_1}(t^2y)A_{i'_2}(ty)A_{i'_3}(y)A_{j'''_2}(x),
\end{align}
where $z=y/x$, and all summation indices range over $\{0,\ldots,n\}$.
Next, we move $A_{b'_1}(t x)$ through the three operators appearing in \eqref{a41}:
\begin{align}\label{a42}
A_{b'_1}(tx)A_{i'_1}(t^2y)A_{i'_2}(ty)A_{i'_3}(y) = 
\sum_{i''_1, i''_2, i''_3, j'_1,j''_1, j'''_1}
\s(zt)^{i'_1, b'_1}_{i''_1, j'_1}\s(z)^{i'_2, j'_1}_{i''_2,j''_1}\s(z/t)^{i'_3, j''_1}_{i''_3, j'''_1}
A_{i''_1}(t^2y)A_{i''_2}(ty)A_{i''_3}(y)A_{j'''_1}(tx).
\end{align}
Combining \eqref{a41} and \eqref{a42}, the expression \eqref{a40} becomes
\begin{align}\label{a43}
\sum_{i''_1, i''_2, i''_3} A_{i''_1}(t^2y)A_{i''_2}(ty)A_{i''_3}(y)
\sum_{j'''_1, j'''_2} A_{j'''_1}(tx)A_{j'''_2}(x)
W(z)^{{\mathbf a}, {\mathbf b}}_{(i''_1, i''_2, i''_3), (j'''_1, j'''_2)},
\end{align}
where the coefficient
$W(z)^{{\mathbf a},{\mathbf b}}_{(i''_1,i''_2,i''_3),,(j'''_1,j'''_2)}$
admits the following diagrammatic representation (cf.\ \eqref{sdg}):
\begin{equation}\label{ff32}
W(z)^{{\mathbf a}, {\mathbf b}}_{(i''_1, i''_2, i''_3), (j'''_1, j'''_2)}=
\sum_{\substack{(a'_1,a'_2,a'_3) \,\in \,\mathcal{C}({\bf a}) \\ (b'_1,b'_2)\, \in \,\mathcal{C}({\bf b})}}
\;\sum_{\substack{ i'_1,i'_2,i'_3 \\ j'_1,j''_1, j'_2,j''_2}}
\begin{tikzpicture}[
  baseline={([yshift=-1.2ex]current bounding box.center)},
  x=0.46cm,y=0.46cm,>=Stealth,thick,
  line cap=round,line join=round,
  scale=1.2
]
  \draw[->] (-0.8,6) -- (1.7,6);
  \draw[->] ( 2.3,6) -- (4.8,6);
  \draw[->] (-0.8,3) -- (1.7,3);
  \draw[->] ( 2.3,3) -- (4.8,3);
  \draw[->] (-0.8,0) -- (1.7,0);
  \draw[->] ( 2.3,0) -- (4.8,0);

  \draw[->] (0.5,-1.3) -- (0.5,1.15);
  \draw[->] (0.5,1.85) -- (0.5,4.15);
  \draw[->] (0.5,4.95) -- (0.5,7.05);  %

  \draw[->] (3.5,-1.3) -- (3.5,1.15);
  \draw[->] (3.5,1.85) -- (3.5,4.15);
  \draw[->] (3.5,4.95) -- (3.5,7.05);

  \node at (2.0,6) {\small $i'_1$};
  \node at (2.0,3) {\small $i'_2$};
  \node at (2.0,0) {\small $i'_3$};

  \node at (0.5,4.50) {\small $j'_1$};
  \node at (0.5,1.50) {\small $j''_1$};
  \node at (3.5,4.50) {\small $j'_2$};
  \node at (3.5,1.50) {\small $j''_2$};

  \def\r{0.35}
  \foreach \X in {0.5,3.5}{
    \foreach \Y in {0,3,6}{
      \draw[thin] (\X,\Y) ++(180:\r) arc (180:270:\r);
    }
  }

  \node[anchor=south east] at (0.40,-1.10) {\scriptsize $z/t$};
  \node[anchor=south east] at (0.40, 2.00) {\scriptsize $z$};
  \node[anchor=south east] at (0.40, 5.00) {\scriptsize $zt$};
  \node[anchor=south east] at (3.40,-1.00) {\scriptsize $z$};
  \node[anchor=south east] at (3.40, 2.00) {\scriptsize $zt$};
  \node[anchor=south east] at (3.40, 5.00) {\scriptsize $zt^2$};

  \node[anchor=west]  at (4.8,6)    {\small $a'_1$};
  \node[anchor=west]  at (4.8,3)    {\small $a'_2$};
  \node[anchor=west]  at (4.8,0)    {\small $a'_3$};
  \node[anchor=south] at (0.5,7.05) {\small $b'_1$}; 
  \node[anchor=south] at (3.5,7.05) {\small $b'_2$};

  \node[anchor=east]  at (-0.8,6) {\small $i''_1$};
  \node[anchor=east]  at (-0.8,3) {\small $i''_2$};
  \node[anchor=east]  at (-0.8,0) {\small $i''_3$};
  \node[anchor=north] at (0.5,-1.3) {\small $j'''_1$};
  \node[anchor=north] at (3.5,-1.3) {\small $j'''_2$};
\end{tikzpicture}
\end{equation}
The right and left columns correspond to
\eqref{a41} and \eqref{a42}, respectively. From
Remark~\ref{re:indep}, the quantity
$W(z)^{{\mathbf a},{\mathbf b}}_{(i''_1,i''_2,i''_3),(j'''_1,j'''_2)}$
is invariant under permutations of $(i''_1,i''_2,i''_3)$ and likewise
under permutations of $(j'''_1,j'''_2)$. Accordingly, the sum over
$i''_1,i''_2,i''_3$ can be reorganized as
\begin{align*}
\sum_{i''_1,i''_2,i''_3 \in \{0,\ldots,n\}}
 = \sum_{0 \le i_1 \le i_2 \le i_3 \le n}
   \sum_{(i''_1,i''_2,i''_3) \in \mathcal{C}\bigl((i_1,i_2,i_3)\bigr)}
 = \sum_{{\mathbf i} \in \TT_3}
   \sum_{{\mathbf i}'' \in \mathcal{C}({\mathbf i})}.
\end{align*}
See \eqref{tl} and \eqref{Cex} for the definitions of $\TT_3$ and
$\mathcal{C}({\mathbf i})$. Treating the sum over $j'''_1,j'''_2$ similarly,
\eqref{a43} can be written as
\begin{align}\label{a44}
\sum_{{\mathbf i} \in \TT_3}\;
\sum_{{\mathbf i}'' \in \mathcal{C}({\mathbf i})}
 A_{i''_1}(t^2 y)\,A_{i''_2}(t y)\,A_{i''_3}(y)
\sum_{{\mathbf j} \in \TT_2}\;
\sum_{{\mathbf j}''' \in \mathcal{C}({\mathbf j})}
 A_{j'''_1}(t x)\,A_{j'''_2}(x)\,
 W(z)^{{\mathbf a},{\mathbf b}}_{{\mathbf i}'',{\mathbf j}'''}.
\end{align}
Comparing the diagrams in \eqref{ff32} and \eqref{f32}, and taking
Remark~\ref{re:indep} into account, we obtain
$W(z)^{{\mathbf a},{\mathbf b}}_{{\mathbf i}'',{\mathbf j}'''}
 = \s_{3,2}(z)^{{\mathbf a},{\mathbf b}}_{{\mathbf i},{\mathbf j}}$
for any ${\mathbf i}'' \in \mathcal{C}({\mathbf i})$ and
${\mathbf j}''' \in \mathcal{C}({\mathbf j})$.
Thus \eqref{a44} equals
\begin{equation}\label{a45}
\begin{split}
&\sum_{{\mathbf i} \in \TT_3,\; {\mathbf j} \in \TT_2}
 \s_{3,2}(z)^{{\mathbf a},{\mathbf b}}_{{\mathbf i},{\mathbf j}}\,
 \Bigl(\sum_{{\mathbf i}'' \in \mathcal{C}({\mathbf i})}
 A_{i''_1}(t^2 y)\,A_{i''_2}(t y)\,A_{i''_3}(y)\Bigr)
\Bigl( \sum_{{\mathbf j}''' \in \mathcal{C}({\mathbf j})}
 A_{j'''_1}(t x)\,A_{j'''_2}(x)\Bigr)
 \\
& =
\sum_{{\mathbf i} \in \TT_3,\; {\mathbf j} \in \TT_2}
 \s_{3,2}(z)^{{\mathbf a},{\mathbf b}}_{{\mathbf i},{\mathbf j}}\,
 A_{\mathbf i}(y)\,A_{\mathbf j}(x).
\end{split}
\end{equation}
\end{proof}

\begin{remark}\label{re:sym}
By applying the sum-to-unity property \eqref{stu2} of  $\s_{k,l}(z)$ to the ZF-algebra relation \eqref{saa},
Remark \ref{re:asym} is generalized to arbitrary $k$ and $l$ as follows:
\begin{align*}
\sum_{{\mathbf a} \in \TT_k, {\mathbf b} \in \TT_l}
A_{\mathbf a}(x) A_{\mathbf b}(y) = 
\sum_{{\mathbf a} \in \TT_k, {\mathbf b} \in \TT_l}
A_{\mathbf a}(y) A_{\mathbf b}(x). 
\end{align*}
\end{remark}

From \eqref{sss} and $\TT_1 = \TT^1$ (see \eqref{tk}, \eqref{tl}), 
the special case $k=1$ of Theorem \ref{th:SAA} is stated in terms of $S^1_{\;\; l}$ as 
\begin{align}\label{zf3}
\sum_{{\mathbf i} \in \TT^1, \, {\mathbf j} \in \TT_l}
S^1_{\;\; l}\left(\frac{x}{w}\right)^{{\mathbf a}, {\mathbf b}}_{{\mathbf i}, {\mathbf j}}
A_{\mathbf i}(t^{l-1}x)A_{\mathbf j}(w) = \Bigl(1-\frac{t^l x}{w}\Bigr)
A_{\mathbf b}(w)A_{\mathbf a}(t^{l-1}x)
\qquad ({\mathbf a} \in \TT^1, {\mathbf b} \in \TT_l).
\end{align}

The following lemma will be a key ingredient in the proof of
Proposition~\ref{pr:T1}.

\begin{lemma}\label{le:SA}
For ${\mathbf a}, {\mathbf i} \in \TT_1$ and ${\mathbf b}, {\mathbf j}  \in \TT_l$ satisfying the 
weight conservation 
(i.e., ${\mathbf a} + {\mathbf b} = {\mathbf i} + {\mathbf j}$ in the multiplicity representation), 
the following equality holds:
\begin{align}\label{SArel}
S^{1}_{\;\,l}(1)^{{\mathbf a},{\mathbf b}}_{{\mathbf i},{\mathbf j}}
A_{\mathbf j}(x)
=\begin{cases}
(1-t^l)
A_{{\mathbf b}\setminus {\mathbf i}}(x)A_{\mathbf a}(t^{l-1}x) 
& \text{if} \;\; {\mathbf i} \subseteq {\mathbf b},
\\
0 & \text{otherwise}.
\end{cases}
\end{align}
Here, $ {\mathbf i} \subseteq {\mathbf b}$ means the condition 
${\mathbf b}-{\mathbf i} \in \BB_{l-1}$ in the multiplicity representation, 
and ${\mathbf b}\setminus {\mathbf i} $ denotes the element 
in $\TT_{l-1}$ corresponding to 
${\mathbf b}-{\mathbf i} \in \BB_{l-1}$.
\end{lemma}

\begin{proof}
Let ${\mathbf b} = (b_1,\ldots, b_l) \in \TT_l$ and
${\mathbf j} = (j_1,\ldots, j_l) \in \TT_l$ be tableau representations.
Below we will freely identify the tableau and multiplicity representations as explained around
\eqref{iic}.
From \eqref{sss} and the expression for $\s_{1,l}(1)$
analogous to \eqref{f32}, we have
\begin{equation}\label{se}
\begin{split}
S^{1}_{\;\, l}(1)^{{\mathbf a},{\mathbf b}}_{{\mathbf i},{\mathbf j}}
&= (1-t^l)\s_{1,l}(t^{l-1})^{{\mathbf a},{\mathbf b}}_{{\mathbf i},{\mathbf j}}
\\
&=(1-t^l)\sum_{(b'_1,\ldots, b'_l) \in \mathcal{C}({\mathbf b})}
\sum_{\alpha_1,\ldots, \alpha_{l-1} \in \TT_1}
 \s(1)^{\alpha_1, b'_1}_{{\mathbf i}, j_1}
\s(t)^{\alpha_2, b'_2}_{\alpha_1, j_2} \cdots \s(t^{l-1})^{{\mathbf a}, b'_l}_{\alpha_{l-1}, j_l},
\end{split}
\end{equation}
where $\s(z)^{{\mathbf a},{\mathbf b}}_{{\mathbf i},{\mathbf j}}
=\s_{1,1}(z)^{{\mathbf a},{\mathbf b}}_{{\mathbf i},{\mathbf j}}$.
The set $\mathcal{C}({\mathbf b})$ has been defined in \eqref{Cex}.
The operator $A_{\mathbf j}(x)$ with ${\mathbf j} \in \TT_l$ 
is given, according to \eqref{Abi}, as
\begin{align}\label{aj1}
A_{\mathbf j}(x) = \sum_{(j'_1,\ldots, j'_l) \in \mathcal{C}({\mathbf j})}
A_{j'_1}(t^{l-1}x) \cdots A_{j'_l}(x).
\end{align}
From Remark \ref{re:indep}, the expression \eqref{se} is invariant under replacing $(j_1,\ldots, j_l) \in \TT_l$
by any array in $\mathcal{C}({\mathbf j})$.
Choosing it so as to coincide with $(j'_1,\ldots, j'_l) $ in \eqref{aj1},
the LHS of \eqref{SArel}, which is the product of \eqref{se} and \eqref{aj1},  is expressed as
\begin{align}\label{sa2}
(1-t^l)
\sum_{(b'_1,\ldots, b'_l) \in \mathcal{C}({\mathbf b})}
\sum_{\substack{\alpha_1,\ldots, \alpha_{l-1} \in \TT_1 \\ (j'_1, \ldots, j'_l) \in \mathcal{C}({\mathbf j})}}
 \s(1)^{\alpha_1, b'_1}_{{\mathbf i}, j'_1}
\s(t)^{\alpha_2, b'_2}_{\alpha_1, j'_2} \cdots \s(t^{l-1})^{{\mathbf a}, b'_l}_{\alpha_{l-1}, j'_l}
A_{j'_1}(t^{l-1}x) \cdots A_{j'_l}(x).
\end{align}
From \eqref{sep}, the sum is restricted to $b'_1={\mathbf i}$ and $\alpha_1 = j'_1$.
The former condition implies that \eqref{sa2} vanishes unless ${\mathbf i} \subseteq {\mathbf b}$ 
in agreement with \eqref{SArel}.
Assuming ${\mathbf i} \subseteq {\mathbf b}$ in what follows,  \eqref{sa2} is reduced to 
\begin{align*}
(1-t^l)
\sum_{({\mathbf i},b'_2,\ldots, b'_l) \in \mathcal{C}({\mathbf b})} 
\sum_{\substack{\alpha_2,\ldots, \alpha_{l-1} \in \TT_1 \\  (j'_1, \ldots, j'_l) \in \mathcal{C}({\mathbf j})}}
\s(t)^{\alpha_2, b'_2}_{j'_1, j'_2} 
\s(t^2)^{\alpha_3, b'_3}_{\alpha_2, j'_3}
\cdots \s(t^{l-1})^{{\mathbf a}, b'_l}_{\alpha_{l-1}, j'_l}
A_{j'_1}(t^{l-1}x) A_{j'_2}(t^{l-2}x) \cdots A_{j'_l}(x).
\end{align*}
The condition $({\mathbf i},b'_2,\ldots,b'_l)\in\mathcal{C}({\mathbf b})$ is
equivalent to $(b'_2,\ldots,b'_l)\in\mathcal{C}({\mathbf b}\setminus{\mathbf i})$.
In the product of $\s$ appearing here, the total incoming (subscript)
weight ${\mathbf j}+\alpha_2+\cdots + \alpha_{l-1}$ and the total outgoing (superscript) weight
${\mathbf b}-{\mathbf i}+{\mathbf a}+\alpha_2+\cdots + \alpha_{l-1}$ coincide by the assumption of the
lemma.  Under this circumstance, the summation
$(j'_1,\ldots,j'_l)\in\mathcal{C}({\mathbf j})$ may be replaced by the
independent sums $j'_1,\ldots,j'_l\in\TT_1$, since the total weight of
$(j'_1,\ldots,j'_l)$ is automatically constrained to be ${\mathbf j}$
by the weight conservation property of the individual $\s$'s.
Then, one can apply the ZF-algebra relation \eqref{zf1} successively to
take the sums over $(j'_1,j'_2)$, $(\alpha_2,j'_3)$, $(\alpha_3,j'_4)$,
\ldots, $(\alpha_{l-1},j'_l)$ in this order, thereby moving
$A_\bullet(t^{l-1}x)$ through to the right, where it eventually becomes
$A_{\mathbf a}(t^{l-1}x)$ as follows:
\begin{align*}
(1-t^l) \left( 
\sum_{(b'_2,\ldots, b'_l) \in \mathcal{C}({\mathbf b}\setminus {\mathbf i})} 
A_{b'_2}(t^{l-2}x) A_{b'_3}(t^{l-3}x)\cdots A_{b'_l}(x) 
\right)A_{\bf a}(t^{l-1}x).
\end{align*}
By the definition \eqref{Abi}, the sum in the parenthesis here is 
equal to $A_{{\mathbf b} \setminus {\mathbf i}}(x)$.
\end{proof}

\begin{example}
For $n=2$ and $l=3$, we have
\begin{equation*}
\begin{split}
&S^1_{\;\;3}(1)^{2,112}_{1,122}=t(1-t^2),\quad
A_{122}(x) = t^3(1+t+t^2)x^3(tx^2 \ok + x \ap \ok+ tx \ap \ok + \ap{}^2 \ok), 
\\
&A_{12}(x) = t(1+t)x^2(x \ok+\ap \ok),
\quad
A_2(x) =x^2+ x \ap.
\end{split}
\end{equation*}
Using \eqref{apc}, it is straightforward to verify that
$S^1_{\;\;3}(1)^{2,112}_{1,122}A_{122}(x)
=(1-t^3)A_{12}(x)A_2(t^2x)$.
\end{example}

\section{Stationary distribution}\label{sec:mp}

Here we derive our second and third main results, namely the matrix product formula
\eqref{pA} for the (unnormalized) stationary probability, and the partition function in 
Theorem \ref{th:z}.
We begin with the following elementary result.

\begin{proposition}
\label{pr:ir}
Suppose $t, x_1, \dots, x_L > 0$.
For any $n, l \in \Z_{\ge 1}$ and ${\bf m} = (m_0,\ldots, m_n)\in (\Z_{\ge 1})^{n+1}$ such that 
$m_0+\cdots + m_n = lL$,
the $n$-species capacity-$l$ inhomogeneous $t$-PushTASEP on $L$-site periodic chain 
in the sector $\mathbb{V}({\mathbf m})$ is irreducible.
\end{proposition}

\begin{proof}
From \eqref{Hdef} and \eqref{wm2}, the off-diagonal elements of
the Markov matrix $H_{n,l}=H_{n,l}(x_1,\ldots, x_L)$ are positive
under the assumption $t, x_1, \dots, x_L > 0$.
Thus, it suffices to show that for any pair of different configurations 
$\vec{\bs}, \vec{\bs}'  \in \mathcal{S}({\mathbf m})$
there is $m\in \Z_{\ge 1}$ such that 
$\langle \vec{\bs}'| (H_{n,l})^m |\vec{\bs}\rangle >0$.

Notice that any particle of species greater than $0$ can move from any site, bumping lower species particles.
To get $\vec{\bs}'$  from $\vec{\bs}$, we first move the particles of species $n$ 
to their desired locations by a sequence of moves. 
We then continue this way moving particles of species $n-1$, and so on.
One can also show the claim by induction on $L$.
\end{proof}

From Proposition~\ref{pr:ir}, 
it follows that the stationary distribution is unique. We now explain the matrix product formula.

\subsection{Matrix product formula}\label{ss:mpf}

Theorem~\ref{th:main1} with \eqref{Hc} transforms the problem of finding the
stationary state of the $n$-species, capacity-$l$
$t$-PushTASEP  on $L$-site periodic lattice with 
 inhomogeneity $x_1, \ldots, x_L$ 
 to the study of eigenvectors of the transfer matrices
$T^0(z|x_1,\ldots,x_L), \ldots, T^{n+1}(z|x_1,\ldots,x_L)$.
Thanks to the commutativity~\eqref{tcom},  it is further
reduced to finding the stationary eigenvector of the transfer matrix
$T^1(z|x_1,\ldots,x_L)$.
This task was done in our previous work~\cite{AK25} for the
case $l=1$ by a matrix product construction based on the basic CTMs
$A_0(z),\ldots,A_n(z)$.
In this section, we generalize it to arbitrary capacity
$l\ge1$.
We will also provide a derivation of the probability conservation property
of the Markov matrix $\mathcal{H}_{n,l}(x_1,\ldots,x_L)$ based on the transfer matrices,
which is independent of the argument given in Section~\ref{ss:pc}.

Let us introduce a state vector whose coefficients are given in the matrix product form:
\begin{subequations}
\begin{align}
&|\mathbb{P}({\mathbf m})\rangle = \sum_{(\bs_1,\ldots, \bs_L) \in \mathcal{S}({\mathbf m})}
{\mathbb{P}}(\bs_1,\ldots, \bs_L)
|\bs_1,\ldots, \bs_L\rangle \in \mathbb{V}({\mathbf m}),
\label{pvec}
\\
&{\mathbb{P}}(\bs_1,\ldots, \bs_L)
=\mathrm{Tr}\left(A_{\bs_1}(x_1)\cdots A_{\bs_L}(x_L)\right),
\label{pA}
\end{align}
\end{subequations}
where $\mathbb{V}({\bf m})$ and $\mathcal{S}({\bf m})$ are defined in 
\eqref{Vm} and \eqref{Sm}, respectively.
The  trace is taken over the Fock space $\FF^{\otimes n(n-1)/2}$,
and is nonzero and convergent under the assumption $m_0,\ldots, m_n \ge 1$
as a formal power series in $t$.\footnote{
See Remark~\ref{re:pos} below for the nonvanishing property, and
\eqref{pA3} together with the discussion following \cite[eq.~(59)]{KOS24}
for convergence in the case $l=1$.}
By the construction, 
${\mathbb{P}}(\bs_1,\ldots, \bs_L)$ is a polynomial in $x_1,\ldots, x_L$ and a rational function of $t$.

Set 
\begin{align}\label{Lam}
\Lambda^k(z|x_1,\ldots,x_L)
 &= e_{k-1}(t^{\KK_1},\ldots,t^{\KK_n})
    \prod_{j=1}^L\!\left(1-\frac{t^l z}{x_j}\right)
   + e_{k}(t^{\KK_1},\ldots,t^{\KK_n})
    \prod_{j=1}^L\!\left(1-\frac{z}{x_j}\right) \quad (0 \le k \le n+1),
\end{align}
where $e_k$ denotes the elementary symmetric polynomial~\eqref{ekb},
and $\KK_i$, defined in~\eqref{Ki}, depends on~${\mathbf m}$.
This is a Yang-Baxterization (spectral parameter dependent version) 
of the $k$'th elementary symmetric polynomial in the following sense:
\begin{subequations}
\begin{align}
\Lambda^k(z|x_1,\ldots,x_L) &= \prod_{j=1}^L \prod_{s=1}^{k-1}\left(1-\frac{t^{-s}z}{x_j}\right)^{-1}
\sum_{0 \le i_1 < \dots < i_k \le n}
\fbox{$i_1$}_{\,z} \,\fbox{$i_2$}_{\,t^{-1}z} \cdots \fbox{$i_k$}_{\,t^{-k+1}z},
\label{bx1}
\\
\fbox{$i$}_{\, z} &= t^{\KK_i}\prod_{j=1}^L\left(1-t^{l\delta_{i,0}}\frac{z}{x_j}\right)
\qquad (0 \le i \le n).
\label{bx2}
\end{align}
\end{subequations}
The formulas \eqref{Lam}-\eqref{bx2} reduce to \cite[Eq.~(6.1)--(6.2b)]{AK25} when $l=1$.

Let $\langle \Omega({\mathbf m}) |$ be the vector in the dual space of $\mathbb{V}({\mathbf m})$
defined as the sum of all basis vectors with unit coefficients:
\begin{align}\label{omg}
\langle \Omega({\mathbf m}) |
= \sum_{ (\bs_1,\ldots, \bs_L) \in \mathcal{S}({\mathbf m})}
\langle \bs_1,\ldots, \bs_L|.
\end{align}

\begin{proposition}\label{pr:T1}
The vectors $|\mathbb{P}({\mathbf m})\rangle$ and $\langle \Omega({\mathbf m}) |$
are the right and left eigenvectors of $T^1(z|x_1,\ldots,x_L)$, respectively,
with the common eigenvalue $\Lambda^1(z|x_1,\ldots,x_L)$ given by \eqref{Lam}.
Namely, they satisfy the following relations:
\begin{align}
T^1(z|x_1,\ldots, x_L)|\mathbb{P}({\mathbf m})\rangle
&= \Lambda^1(z|x_1,\ldots, x_L) |\mathbb{P}({\mathbf m})\rangle,
\label{tpL}\\
\langle \Omega({\mathbf m})|T^1(z|x_1,\ldots, x_L)| &=
\langle \Omega({\mathbf m})| \Lambda^1(z|x_1,\ldots, x_L).
\label{omt}
\end{align}
\end{proposition}

\begin{proof}
From \eqref{Sd}, \eqref{tke} and \eqref{Lam}, 
both sides of \eqref{tpL} and \eqref{omt} are  polynomials in $z$ of degree at most $L$. 
Therefore it suffices to verify the equalities at the $L+1$ points
$z=0, x_1,\ldots, x_L$.
At $z=0$,  both \eqref{tpL} and \eqref{omt} follow readily from \eqref{tk0}, \eqref{Lam}  
and $e_{k-1}(t^{\KK_1},\ldots, t^{\KK_n}) + e_k(t^{\KK_1},\ldots, t^{\KK_n})
= e_k(t^{\KK_0},t^{\KK_1},\ldots, t^{\KK_n})$.

To prove \eqref{tpL} at the remaining points,  we compute 
the action of $T^1(z|x_1,\ldots, x_L)$ using \eqref{tke}:
\begin{align}
&T^1(z|x_1,\ldots, x_L) |\mathbb{P}({\mathbf m})\rangle
=\sum_{(\bs_1,\ldots, \bs_L) \in \mathcal{S}({\bf m})}
{\mathbb{P}}(z|\bs_1,\ldots, \bs_L)|\bs_1,\ldots, \bs_L\rangle,
\\
&{\mathbb{P}}(z|\bs_1,\ldots, \bs_L) 
= \sum_{\substack{{\mathbf a}_1,\ldots, {\mathbf a}_L \in \TT^1 \\ 
(\bs'_1,\ldots, \bs'_L) \in \mathcal{S}({\bf m})} }
\!\!\!S \Bigl( \frac{z}{x_1} \Bigr)^{{\mathbf a}_2,\bs_1}_{{\mathbf a}_1, \bs'_1}
S \Bigl( \frac{z}{x_2} \Bigr)^{{\mathbf a}_3,\bs_2}_{{\mathbf a}_2, \bs'_2}
\cdots 
S \Bigl( \frac{z}{x_L} \Bigr)^{{\mathbf a}_1,\bs_L}_{{\mathbf a}_L, \bs'_L}
\mathrm{Tr}(A_{\bs'_1}(x_1) \cdots A_{\bs'_L}(x_L)),
\label{ppmp}
\end{align}
where $S(z)^{{\mathbf a}, {\mathbf b}}_{{\mathbf i}, {\mathbf j}} 
= S^1_{\;\,l}(z)^{{\mathbf a}, {\mathbf b}}_{{\mathbf i}, {\mathbf j}}$. 
The summation over $(\bs'_1,\ldots, \bs'_L)$ is automatically 
restricted to $\mathcal{S}({\bf m})$ \eqref{Sm} 
by the weight conservation property of $T^1(z|x_1,\ldots, x_L)$. 
Hence, we simply write  
$\bs'_1,\ldots, \bs'_L \in \TT_l$.

We now consider the point $z=x_1$. 
In view of \eqref{Lam}, we are to show
\begin{align}\label{paim}
{\mathbb{P}}(z=x_1|\bs_1,\ldots, \bs_L) = {\mathbb{P}}(\bs_1,\ldots, \bs_L)
(1-t^l)\prod_{j=2}^L\bigl(1-\frac{t^lx_1}{x_j}\bigr).
\end{align} 
By Lemma \ref{le:SA}, the expression \eqref{ppmp} reduces to
\begin{equation}
\begin{split}
{\mathbb{P}}(x_1|\bs_1,\ldots, \bs_L) 
= (1-t^l) \!\!\! \!\!\! 
\sum_{\substack{{\mathbf a}_1,\ldots, {\mathbf a}_L \in \TT^1, \, ({\mathbf a}_1 \subseteq \bs_1)
\\ \bs'_2,\ldots, \bs'_L \in \TT_l}}
\!\!\! & S\Bigl( \frac{x_1}{x_2} \Bigr)^{{\mathbf a}_3,\bs_2}_{{\mathbf a}_2, \bs'_2}
S\Bigl( \frac{x_1}{x_3} \Bigr)^{{\mathbf a}_4,\bs_3}_{{\mathbf a}_3, \bs'_3}
\cdots 
S\Bigl( \frac{x_1}{x_L} \Bigr)^{{\mathbf a}_1,\bs_L}_{{\mathbf a}_L, \bs'_L}
\\
& \!\!\! 
\times \mathrm{Tr}(A_{\bs_1\setminus {\mathbf a}_1}(x_1) A_{{\mathbf a}_2}(t^{l-1}x_1) 
A_{\bs'_2}(x_2) \cdots A_{\bs'_L}(x_L)),
\label{ppmp2}
\end{split}
\end{equation}
where the condition ${\mathbf a}_1 \subseteq \bs_1$
is as in Lemma \ref{le:SA}.
Applying the ZF algebra relation \eqref{zf3} successively, we perform the
summations over
$({\mathbf a}_2, \bs'_2), ({\mathbf a}_3, \bs'_3), \ldots, ({\mathbf a}_l, \bs'_l)$ 
in this order, thereby moving $A_\bullet(t^{l-1}x_1)$ to the right:
\begin{equation*}
\begin{split}
{\mathbb{P}}(x_1|\bs_1,\ldots, \bs_L) 
=(1-t^l)\prod_{j=2}^L\bigl(1-\frac{t^lx_1}{x_j}\bigr)
\sum_{{\mathbf a}_1 \in \TT^1,  ({\mathbf a}_1 \subseteq \bs_1)}
\mathrm{Tr}(
A_{\bs_1\setminus {\mathbf a}_1}(x_1)A_{\bs_2}(x_2) \cdots A_{\bs_L}(x_L)A_{{\mathbf a}_1}(t^{l-1}x_1)).
\end{split}
\end{equation*}
By the cyclicity of the trace and the definition \eqref{Abi}, the sum above is rewritten as
\begin{equation*}
\begin{split}
&\sum_{{\mathbf a}_1 \in \TT^1,  ({\mathbf a}_1 \subseteq \bs_1)}\mathrm{Tr}(
A_{{\mathbf a}_1}(t^{l-1}x_1)A_{\bs_1\setminus {\mathbf a}_1}(x_1)A_{\bs_2}(x_2) \cdots A_{\bs_L}(x_L))
= \mathrm{Tr}(A_{\bs_1}(x_1)A_{\bs_2}(x_2) \cdots A_{\bs_L}(x_L)),
\end{split}
\end{equation*}
which coincides with ${\mathbb{P}}(\bs_1,\ldots, \bs_L)$ in \eqref{pA}.
This establishes \eqref{paim}.  

Next we show \eqref{omt} at $z=x_1$.
From \eqref{omg},  \eqref{tke} and \eqref{tkr},  we have
\begin{equation*}
\langle \Omega({\mathbf m})|T^1(z|x_1,\ldots, x_L)| \bs_1,\ldots, \bs_L\rangle
= 
\sum_{\substack{{\mathbf a}_1,\ldots, {\mathbf a}_L \in \TT^1 \\ 
(\bs'_1,\ldots, \bs'_L) \in \mathcal{S}({\bf m})} }
\!\!\!S \Bigl( \frac{z}{x_1} \Bigr)^{{\mathbf a}_2,\bs'_1}_{{\mathbf a}_1, \bs_1}
S \Bigl( \frac{z}{x_2} \Bigr)^{{\mathbf a}_3,\bs'_2}_{{\mathbf a}_2, \bs_2}
\cdots 
S \Bigl( \frac{z}{x_L} \Bigr)^{{\mathbf a}_1,\bs'_L}_{{\mathbf a}_L, \bs_L},
\end{equation*}
where $S(z)^{{\mathbf a}, {\mathbf b}}_{{\mathbf i}, {\mathbf j}} 
= S^1_{\;\,l}(z)^{{\mathbf a}, {\mathbf b}}_{{\mathbf i}, {\mathbf j}}$. 
Setting $z=x_1$ and applying \eqref{se} with ${\mathbf j}=(j_1,\ldots,j_l)$ replaced by 
$\bs_1=(\sigma_{1,1},\ldots,\sigma_{1,l})\in\TT_l$ for the leftmost $S$,
together with \eqref{sss} for the remaining ones,
the RHS becomes
\begin{equation}\label{rrk}
\begin{split}
(1-t^l)\prod_{j=2}^L\bigl(1-\frac{t^lx_1}{x_j}\bigr)
\sum_{\substack{{\mathbf a}_1,\ldots, {\mathbf a}_L \in \TT^1 \\ 
(\bs'_1,\ldots, \bs'_L) \in \mathcal{S}({\bf m})} }
\sum_{\substack{\alpha_1,\ldots, \alpha_{l-1} \in \TT_1 \\
(b_1,\ldots, b_l) \in \mathcal{C}(\bs'_1)}}
& \s_{1,1}(1)^{\alpha_1, b_1}_{{\mathbf a}_1,  \sigma_{1,1}}
\s_{1,1}(t)^{\alpha_2, b_2}_{\alpha_1, \sigma_{1,2}} \cdots 
\s_{1,1}(t^{l-1})^{{\mathbf a}_2, b_l}_{\alpha_{l-1}, \sigma_{1,l}}
\\
\times &\; \s_{1,l}\Bigl( \frac{t^{l-1}x_1}{x_2} \Bigr)^{{\mathbf a}_3,\bs'_2}_{{\mathbf a}_2, \bs_2}
\cdots 
\s_{1,l}\Bigl( \frac{t^{l-1}x_1}{x_L} \Bigr)^{{\mathbf a}_1,\bs'_L}_{{\mathbf a}_L, \bs_L}.
\end{split}
\end{equation}
The prefactor in the first line coincides with
$\Lambda^1(x_1\mid x_1,\ldots,x_L)$ in \eqref{Lam}.
Thus it remains to show that the summation part in \eqref{rrk} is equal to~$1$.
Note that the sums over $\bs'_1\in\TT_l$ and
$(b_1,\ldots,b_l)\in\mathcal{C}(\bs'_1)$ are equivalent to the sums over
$b_1,\ldots,b_l\in\TT_1$.
Then \eqref{sep} allows us to eliminate $\alpha_1$ and $b_1$,
thereby reducing the summations in \eqref{rrk} to
\begin{equation}\label{hrm}
\begin{split}
\sum_{\substack{{\mathbf a}_1,\ldots, {\mathbf a}_L \in \TT^1 \\ 
\bs'_2,\ldots, \bs'_L \in \mathcal{S}({\bf m})} }
\sum_{\substack{\alpha_2,\ldots, \alpha_{l-1} \in \TT_1 \\
b_2,\ldots, b_l \in \TT_1}}
\s_{1,1}(t)^{\alpha_2, b_2}_{\sigma_{1,1}, \sigma_{1,2}} \cdots 
\s_{1,1}(t^{l-1})^{{\mathbf a}_2, b_l}_{\alpha_{l-1}, \sigma_{1,l}}
\s_{1,l}\Bigl( \frac{t^{l-1}x_1}{x_2} \Bigr)^{{\mathbf a}_3,\bs'_2}_{{\mathbf a}_2, \bs_2}
\cdots 
\s_{1,l}\Bigl( \frac{t^{l-1}x_1}{x_L} \Bigr)^{{\mathbf a}_1,\bs'_L}_{{\mathbf a}_L, \bs_L}.
\end{split}
\end{equation}
Now we can apply the sum-to-unity property \eqref{stu2}, either with $(k,l)=(1,1)$ or with
$(k,l)=(1,l)$, to perform the summations over
\begin{equation*}
({\mathbf a}_1,\bs'_L),\; ({\mathbf a}_L,\bs'_{L-1}),\;
\ldots,\; ({\mathbf a}_3,\bs'_2),\; ({\mathbf a}_2,b_l),\;
(\alpha_{l-1},b_{l-1}),\; \ldots,\; (\alpha_2,b_2)
\end{equation*}
in this order from the left, thereby confirming that \eqref{hrm} is indeed equal to~$1$.

Thus far, we have proved \eqref{tpL} and \eqref{omt} at $z=x_1$.
The equalities at $z=x_2,\ldots,x_L$ follow from the case $z=x_1$
by the $\Z_L$-cyclic translational symmetry of the system.
\end{proof}

Functional relations among the commuting family of transfer matrices  
(cf.~\cite[Sec.~3,~8]{KNS11}) and the analytic Bethe ansatz (cf.~\cite{R87, KS95, FH15}),  
together with the Yang-Baxterized structure \eqref{bx1}--\eqref{bx2} 
indicate that \eqref{tpL} automatically leads to \eqref{tpk} below:
\begin{align}
T^k(z|x_1,\ldots, x_L)|\mathbb{P}({\mathbf m})\rangle
&= \Lambda^k(z|x_1,\ldots, x_L)|\mathbb{P}({\mathbf m})\rangle,
\label{tpk}
\\
\langle \Omega({\mathbf m})|T^k(z|x_1,\ldots, x_L)| &=
\langle \Omega({\mathbf m})| \Lambda^k(z|x_1,\ldots, x_L),
\label{omk}
\end{align}
where $0 \le k \le n+1$.
In fact,  the cases $k=0$ and $k=n+1$ are explicitly verified by \eqref{t0f} and \eqref{tnf}, respectively.
In general, $\Lambda^k(z|x_1,\ldots, x_L)$ is the so-called 
{\em dressed vacuum form} \cite{R87, KS95},  in which all Baxter
$Q$-functions are constant, as discussed in \cite[Sec.~4.5]{KMMO16}.
In what follows, we assume \eqref{tpk},
leaving a self-contained derivation analogous to Proposition \ref{pr:T1}
for the case $k=1$ as a future problem.

In Section \ref{ss:pf}, we will show that
$\langle \Omega({\mathbf m})|\mathbb{P}({\mathbf m})\rangle \neq 0$ for generic 
$x_1,\ldots, x_L$ and $t$.
Therefore, \eqref{omk} follows from \eqref{tpk} together with the fact that
$\langle \Omega({\mathbf m})|$ is a left eigenvector of
$T^k(z|x_1,\ldots, x_L)$.
The latter is guaranteed by \eqref{omt} and the commutativity 
\eqref{tcom}.

Let $\dot{\Lambda}^k(z)=\frac{d\Lambda^k(z)}{dz}$. 
Differentiating
\eqref{Lam} at $z=0$, we obtain
\begin{align*}
\dot{\Lambda}^k(0)
= -\!\left(
    t^l e_{k-1}(t^{\KK_1},\ldots,t^{\KK_n})
   + e_k(t^{\KK_1},\ldots,t^{\KK_n})\right)
    \sum_{j=1}^L \frac{1}{x_j} = -\, e_k(t^l,t^{\KK_1},\ldots,t^{\KK_n})
    \sum_{j=1}^L \frac{1}{x_j}.
\end{align*}
Thus, the alternating sum \eqref{Hc} with $\dot{T}^k(0)$ replaced by the
eigenvalues $\dot{\Lambda}^k(0)$ becomes
\begin{align}\label{Dm2}
D_{\mathbf m}^{-1}
 \sum_{k=0}^{n+1}(-1)^{k-1}\dot{\Lambda}^k(0)
 &= D_{\mathbf m}^{-1}
    \!\left(\sum_{j=1}^L \frac{1}{x_j}\right)
    \sum_{k=0}^{n+1}(-1)^k
     e_k(t^l,t^{\KK_1},\ldots,t^{\KK_n})
   = \sum_{j=1}^L \frac{1}{x_j},
\end{align}
where the last equality follows from \eqref{ekg}, \eqref{Ki} and \eqref{Dm}.
Consequently, the expression \eqref{Hc} for the Markov matrix 
in terms of transfer matrices can be rewritten as
\begin{align}\label{ccH}
\mathcal{H}_{n,l}(x_1,\ldots,x_L)
 = D_{\mathbf m}^{-1}
   \frac{d}{dz}
   \sum_{k=0}^{n+1}(-1)^{k-1}
   \!\left.
   \bigl(T^k(z|x_1,\ldots,x_L)
        -\Lambda^k(z|x_1,\ldots,x_L)\bigr)
   \right|_{z=0}.
\end{align}
From this and \eqref{tpk} and \eqref{omk}, 
the joint right eigenvector 
$|\mathbb{P}({\mathbf m})\rangle$ 
and the joint left eigenvector 
$\langle \Omega({\mathbf m})|$ of the transfer matrices satisfy
the stationarity condition and the probability conservation condition, respectively:
\begin{align}
\mathcal{H}_{n,l}(x_1,\ldots,x_L)
\,|\mathbb{P}({\mathbf m})\rangle &= 0,
\label{HP0}
\\
\langle \Omega({\mathbf m})| \mathcal{H}_{n,l}(x_1,\ldots,x_L) &=0.
\label{pc}
\end{align}
By the uniqueness of the stationary state in each $\mathbb{V}({\mathbf m})$,
\eqref{HP0} implies that, up to normalization, the vector
$|\mathbb{P}({\mathbf m})\rangle$ in \eqref{pvec} is the stationary state,
and that ${\mathbb{P}}(\bs_1,\ldots,\bs_L)$ in \eqref{pA} provides a
matrix product formula for the (unnormalized) stationary probability.
In view of \eqref{omg}, the property \eqref{pc} implies that
$\mathcal{H}_{n,l}(x_1,\ldots,x_n)$ is probability conserving.
This fact was proved in Section~\ref{ss:pc} for $H_{n,l}(x_1,\ldots,x_n)$ by a purely combinatorial method.
The present approach yields an alternative proof based on the integrability
of the underlying vertex model.

Let us present a few examples illustrating applications of the matrix
product formula \eqref{pA} for the stationary probability.

\begin{example}\label{ex:sl2}
For the single species case $n=1$, the operator $A_{\bs_j}(x)$ reduces to
the scalar \eqref{Asn1} for
$\bs_j=(\sigma_{j,0},\sigma_{j,1})\in\BB_l$
in the multiplicity representation.
As a consequence, the unnormalized stationary probability
\eqref{pA} factorizes into a product of single-site contributions:
\begin{align}
\mathbb{P}(\bs_1,\ldots,\bs_L)
=
\prod_{j=1}^L
x_j^{\sigma_{j,1}}\,
t^{\sigma_{j,1}(\sigma_{j,1}-1)/2}
\qbinom{l}{\sigma_{j,1}}_{t}.
\end{align}

Consider the $L=l=4$ case and the states 
$\vec{\bs} = ((0,4), (1,3), (2,2), (3,1))$ and 
$\vec{\bs}' = ((1,3), (1,3), (1,3), (3,1))$ in $\mathcal{S}({\bf m})$ 
with the particle multiplicity ${\bf m}=(6,10)$.
From
\begin{alignat*}{2}
\langle \vec{\bs}'| H_{1,4}(x_1,x_2,x_3,x_4) | \vec{\bs}\rangle &= \frac{t(1-t^2)}{(1-t^6)x_1},
&\quad
\mathbb{P}(\vec{\bs}) &= \frac{t^{10}(1-t^3)(1-t^4)^3 x_1^4x_2^3x_3^2x_4}{(1-t)^3(1-t^2)},
\\
\langle \vec{\bs}| H_{1,4}(x_1,x_2,x_3,x_4) | \vec{\bs}'\rangle &= \frac{t^4(1-t)(1-t^3)}{(1-t^4)(1-t^6)x_3},
&\quad
\mathbb{P}(\vec{\bs}') &= \frac{t^{9}(1-t^4)^4 x_1^3x_2^3x_3^3x_4}{(1-t)^4}, 
\end{alignat*}
we obtain
$\langle \vec{\bs}'| H_{1,4}(x_1,x_2,x_3,x_4) | \vec{\bs}\rangle \mathbb{P}(\vec{\bs}) 
/\bigl(\langle \vec{\bs}| H_{1,4}(x_1,x_2,x_3,x_4) | \vec{\bs}'\rangle \mathbb{P}(\vec{\bs}')\bigr)
= t^{-2} \neq 1$,
which shows that the detailed balance is not satisfied in general.
\end{example}

\begin{example}\label{ex:sst}
For $n=l=2$, $L=3$ and the particle multiplicity ${\bf m} = (2,3,1)$, the vector 
\eqref{pvec} is given by
\begin{equation*}
\begin{split}
\frac{(1-t^3)(1+t^2)}{t(1+t)}|\mathbb{P}({\mathbf m})\rangle 
&=t x_1^2 x_2^2 \left(t^2 x_1+x_1+x_3\right)|12,11,00\rangle 
+t x_1^2 x_3^2 \left(t^2 x_1+t^2 x_2+x_1\right)|12,00,11\rangle\\
&+
(t+1) x_1^2 x_2 x_3 \left(t^3 x_1+t^2 x_1+t^2 x_2+t x_1+t x_3+x_1\right)|12,01,01\rangle \\
&+ x_1^2 x_2 x_3 \left(t^4 x_2+t^3 x_2+t^3 x_3+t^2 x_2+t x_2+x_2\right)|11,02,01\rangle
\\
&+ x_1^2x_2x_3\left(t^4 x_3+t^3 x_3+t^2 x_3+t x_2+t x_3+x_3\right)|11,01,02\rangle + \mathrm{cyclic},
\end{split}
\end{equation*}
where $\,\mathrm{cyclic}\,$ means the ten additional terms obtained by the
$\Z_3$-cyclic shifts
$|\bs_1, \bs_2, \bs_3\rangle \mapsto |\bs_{i+1}, \bs_{i+2}, \bs_{i+3}\rangle$
and $x_i \mapsto x_{i+1}$ for $i=1,2$,
with all indices understood modulo $3$.
Using Example \ref{ex:A22} and \eqref{apc}, the second line, for example, 
 is derived from the following calculation:
\begin{equation}\label{p23}
\begin{split}
\mathbb{P}(12,01,01)& = \mathrm{Tr}(A_{12}(x_1) A_{01}(x_2)A_{01}(x_3))
\\
&=\mathrm{Tr}\left( t(1+t)x_1^2(x_1+\ap)\ok\, \cdot 
(1+t)x_2(1+x_2 \am) \ok \, \cdot 
(1+t)x_3(1+x_3 \am) \ok \right)
\\
&=t(1+t)^3 x_1^2x_2x_3\mathrm{Tr}\left(\ok{}^3 (x_1+\ap)(1+t^{-1}x_2 \am)(1+t^{-2}x_3 \am)\right)
\\
&= t(1+t)^3 x_1^2x_2x_3 \mathrm{Tr}\left(\ok{}^3(x_1+t^{-1}x_2 \ap \am + t^{-2}x_3 \ap \am)\right)
\\
&= t(1+t)^3 x_1^2x_2x_3  \mathrm{Tr}\left(\ok{}^3(x_1+t^{-1}x_2 + t^{-2}x_3-(t^{-1}x_2 + t^{-2}x_3)\ok)\right)
\\
&= t(1+t)^3 x_1^2x_2x_3 \left(\frac{x_1+t^{-1}x_2 + t^{-2}x_3}{1-t^3}
-\frac{t^{-1}x_2 + t^{-2}x_3}{1-t^4}\right)
\\
&= \frac{t(1+t)^2 x_1^2x_2x_3 (t^3 x_1+t^2 x_1+t^2 x_2+t x_1+t x_3+x_1)}{(1-t^3)(1+t^2)}.
\end{split}
\end{equation}
\end{example}

\begin{example}\label{ex:L6}
For $n=2$, $l=1$ and $L=6$ with the same particle multiplicity
${\mathbf m}=(2,3,1)$ as in Example~\ref{ex:sst},
the unnormalized stationary probability is computed using
the operators $A_0(x)$, $A_1(x)$, and $A_2(x)$ given in
Example~\ref{ex:akos}.
The result is as follows:
\begin{equation}\label{p1}
\begin{split}
\mathbb{P}(1,2,0,1,0,1)
&=x_1 x_2 x_4 x_6 \left(t^3 x_2+t^3 x_3+t^2 x_2+t^2 x_5+t x_2+x_2\right)/D,
\\
\mathbb{P}(1,2,0,1,1,0)
&=x_1 x_2 x_4 x_5\left(t^3 x_2+t^3 x_3+t^2 x_2+t x_2+t x_6+x_2\right)/D,
 \\
\mathbb{P}(1,2,1,0,0,1)
&=   x_1 x_2 x_3 x_6 \left(t^3 x_2+t^2 x_2+t^2x_4+t^2 x_5+t x_2+x_2\right)/D,
\\
\mathbb{P}(1,2,1,0,1,0)
&=   x_1 x_2 x_3 x_5 \left(t^3 x_2+t^2 x_2+t^2 x_4+t x_2+tx_6+x_2\right)/D,
\\
\mathbb{P}(2,1,0,1,0,1)
&=   x_1 x_2 x_4 x_6 \left(t^3 x_1+t^2 x_1+t^2 x_3+t x_1+t x_5+x_1\right)/D,
\\
\mathbb{P}(2,1,0,1,1,0)
&=   x_1 x_2 x_4 x_5\left(t^3 x_1+t^2 x_1+t^2 x_3+t x_1+x_1+x_6\right)/D,
\\
\mathbb{P}(2,1,1,0,0,1)
&=   x_1 x_2 x_3 x_6 \left(t^3 x_1+t^2 x_1+t x_1+tx_4+t x_5+x_1\right)/D,
\\
\mathbb{P}(2,1,1,0,1,0)
&=   x_1 x_2 x_3 x_5 \left(t^3 x_1+t^2 x_1+t x_1+t x_4+x_1+x_6\right)/D,
\end{split}
\end{equation}
where $D=(1+t^2)(1-t^3)$.
\end{example}
   
Let us write the matrix product formula \eqref{pA} as
\begin{align}\label{pA2}
\mathbb{P}_l
\!\left(\genfrac{}{}{0pt}{}{x_1,\ldots,x_L}{\bs_1,\ldots,\bs_L}\right)
= \mathrm{Tr}(A_{\bs_1}(x)\cdots A_{\bs_L}(x_L))\qquad (\bs_1,\ldots, \bs_L \in \TT_l),
\end{align}
which makes explicit the dependence on the capacity $l$ and the
site-wise inhomogeneities $x_1,\ldots,x_L$.
It possesses cyclic translational symmetry:
\begin{align}\label{cts}
\mathbb{P}_l
\!\left(\genfrac{}{}{0pt}{}{x_1,\ldots,x_L}{\bs_1,\ldots,\bs_L}\right)
=
\mathbb{P}_l
\!\left(\genfrac{}{}{0pt}{}{x_L, \ldots,x_{L-1}}{\bs_L,\ldots,\bs_{L-1}}\right).
\end{align}
From \eqref{Abi} and \eqref{Asym}, this expression \eqref{pA2} is reduced to
the $l=1$ case as
\begin{align}\label{pA3}
\mathbb{P}_l
\!\left(\genfrac{}{}{0pt}{}{x_1,\ldots,x_L}{\bs_1,\ldots,\bs_L}\right)
= \sum_{(\sigma_{1,1},\ldots, \sigma_{1,l}) \in \mathcal{C}(\bs_1)}
\cdots
\sum_{(\!\sigma_{L,1},\ldots, \sigma_{L,l}) \in \mathcal{C}(\bs_L)}
\mathbb{P}_1
\!\left(\genfrac{}{}{0pt}{}{\;  x_1,\;  tx_1, \ldots, t^{l-1}x_1 , 
\ldots, \; \,  x_L, \;\, tx_L, \ldots, t^{l-1}x_L}{\!\!\sigma_{1,1}, \sigma_{1,2},  \ldots, \sigma_{1,l}, \;\; \ldots, \; 
\sigma_{L,1}, \sigma_{L,2},  \ldots, \sigma_{L,l} }\right),
\end{align}
where $\mathcal{C}(\bs)$ is defined in \eqref{Cex}.
The modification of the inhomogeneity parameters in this formula
can be viewed as a plethystic substitution,
corresponding to a formal duplication of each variable $x_j$
into $l$ components with geometric weights,
summarized as $x_j \longmapsto  x_j(1-t^l)/(1-t):= (x_j,tx_j,\ldots, t^{l-1}x_j)$.

\begin{example}\label{ex:3}
An example of \eqref{pA3} with $n=2$, $l=2$ and $L=3$ is 
\begin{equation*}
\begin{split}
\mathbb{P}_2
\!\left(\genfrac{}{}{0pt}{}{x_1,x_2,x_3}{12, 01, 01}\right)
&=
\mathbb{P}_1
\!\left(\genfrac{}{}{0pt}{}{x_1, tx_1, x_2, tx_2, x_3, tx_3}{ \!1 \; , \; 2\; ,\; 0\;, \;1\;,\;0\;,\;1}\right) +
\mathbb{P}_1
\!\left(\genfrac{}{}{0pt}{}{x_1, tx_1, x_2, tx_2, x_3, tx_3 }{ \!1 \; , \; 2\; ,\; 0\;, \;1\;,\;1\;,\;0}\right) \\
&+
\mathbb{P}_1
\!\left(\genfrac{}{}{0pt}{}{x_1, tx_1, x_2, tx_2, x_3, tx_3 }{ \!1 \; , \; 2\; ,\; 1\;, \;0\;,\;0\;,\;1}\right) +
\mathbb{P}_1
\!\left(\genfrac{}{}{0pt}{}{x_1, tx_1, x_2, tx_2, x_3, tx_3 }{ \!1 \; , \; 2\; ,\; 1\;, \;0\;,\;1\;,\;0}\right) \\
&
+\mathbb{P}_1
\!\left(\genfrac{}{}{0pt}{}{x_1, tx_1, x_2, tx_2, x_3, tx_3 }{ \!2 \; , \; 1\; ,\; 0\;, \;1\;,\;0\;,\;1}\right) +
\mathbb{P}_1
\!\left(\genfrac{}{}{0pt}{}{x_1, tx_1, x_2, tx_2, x_3, tx_3 }{ \!2 \; , \; 1\; ,\; 0\;, \;1\;,\;1\;,\;0}\right) \\
&+
\mathbb{P}_1
\!\left(\genfrac{}{}{0pt}{}{x_1, tx_1, x_2, tx_2, x_3, tx_3 }{ \!2 \; , \; 1\; ,\; 1\;, \;0\;,\;0\;,\;1}\right) +
\mathbb{P}_1
\!\left(\genfrac{}{}{0pt}{}{x_1, tx_1, x_2, tx_2, x_3, tx_3 }{ \!2 \; , \; 1\; ,\; 1\;, \;0\;,\;1\;,\;0}\right).
\end{split}
\end{equation*}
The LHS is given by \eqref{p23}.
The RHS is the sum of the eight quantities
$\mathbb{P}$ in \eqref{p1}, with the parameters
$(x_1,x_2,x_3,x_4,x_5,x_6)$ set as
$(x_1,t x_1,x_2,t x_2,x_3,t x_3)$.
The equality of the two sides can be verified by direct computation.
\end{example}

\begin{remark}\label{re:pos}
The trace \eqref{pA2} 
is positive for $x_1,\ldots, x_L>0$ and $0<t<1$.
In fact, thanks to \eqref{pA3}, this claim reduces to the case $l=1$.
It then follows readily from the fact that
the CTMs $A_0(x), \ldots, A_n(x)$ in \eqref{ctm1} are linear combinations of
products of $t$-oscillator generators with positive coefficients,
and that the matrix elements of the $t$-oscillators in \eqref{tom} are positive.
\end{remark}

\subsection{Partition function}\label{ss:pf}
We now focus on 
\begin{equation}\label{zdef}
Z_{l,\mathbf m}(x_1,\ldots, x_L;t)
= \langle \Omega({\mathbf m})|\mathbb{P}({\mathbf m})\rangle
=\sum_{(\bs_1,\ldots, \bs_L) \in \mathcal{S}({\mathbf m})}
\mathbb{P}_l\!\left(\genfrac{}{}{0pt}{}{x_1,\ldots,x_L}{\bs_1,\ldots,\bs_L};t\right),
\end{equation}
where $Ll = | {\mathbf m}|=m_0+\cdots + m_n$.
This quantity is the normalization factor of the stationary probability,
often referred to as the \emph{partition function}.
Here we have made the dependence on $t$ and the capacity $l$ explicit.
Due to Remark \ref{re:pos}, the partition function $Z_{l,\mathbf m}(x_1,\ldots, x_L;t)$
is also positive for $x_1,\ldots, x_L>0$ and $0<t<1$.
From \eqref{pA2} and the exchange symmetry in Remark \ref{re:sym}, 
it is a symmetric polynomial in $x_1,\ldots, x_L$ with rational coefficients.

Our third main result in this paper is the following corollary of \eqref{pA3}:
\begin{theorem}\label{th:z}
The partition function with general capacity $l$ is reduced to the $l=1$ case as follows:
\begin{align}\label{zpr}
Z_{l,\mathbf m}(x_1,x_2, \ldots,  x_L;t)
= Z_{1,\mathbf m}\left(\frac{1-t^l}{1-t}x_1, \frac{1-t^l}{1-t}x_2, \ldots,
\frac{1-t^l}{1-t}x_L;t\right),
\end{align}
where $(1-t^l)x_j/(1-t)$ conventionally denotes the $l$-tuple of geometrically weighted 
inhomogeneity parameters $x_j, tx_j,\ldots, \ldots, t^{l-1}x_j$,
and  the RHS corresponds to a system on a lattice of length $lL$
with capacity $1$. (See also the explanation after \eqref{zpr0}.)
\end{theorem}

\begin{example}\label{ex:zz}
For $n=2$, we have
\begin{align*} 
Z_{2,(1,2,1)}(x_1,x_2;t) &= \frac{t(1+t)x_1x_2(x_1+x_2)^2}{1-t},
\\
Z_{1,(1,2,1)}(x_1,x_2,x_3,x_4;t)&= \frac{e_1(x_1,x_2,x_3,x_4)e_3(x_1,x_2,x_3,x_4)}{1-t^2},
\end{align*}
where $e_k$ denotes the elementary symmetric polynomial defined in \eqref{ekb}.
They satisfy the relation \eqref{zpr}, namely,
\begin{align*}
Z_{2,(1,2,1)}(x_1,x_2;t) = Z_{1,(1,2,1)}(x_1,tx_1, x_2, tx_2).
\end{align*}
\end{example}

In \cite{CMW22, AMW24}, the partition function $Z_{l,\mathbf m}(x_1,\ldots,x_L;t)$
of the $t$-PushTASEP with capacity $l=1$ and ${\mathbf m} = (m_0,\ldots, m_n)$ 
has been linked to Macdonald polynomials $P_\lambda(x_1,\ldots, x_L; q,t)$
at $q=1$, where $\lambda = \langle n^{m_n}, \dots, 1^{m_1}, 0^{m_0} \rangle$ in frequency notation.
Combined with the standard factorized formula for the Macdonald polynomials at $q=1$ \cite[p.\,324, (iv)]{Mac95},
this yields
\begin{subequations}
\begin{align}
Z_{1,\mathbf m}(x_1,\ldots,x_L;t) &= \Delta_{\mathbf m}(t)^{-1}\prod_{i = 1}^n 
e_{m_i + \cdots + m_n}(x_1,\ldots,x_L),
\label{ze}\\
\Delta_{\mathbf m}(t) &= \prod_{1 \le i \le j \le n-1}(1-t^{m_i+m_{i+1}+\cdots + m_j}).
\label{amt}
\end{align}
\end{subequations}
The function $\Delta_{\mathbf m}(t)$, 
for example $1, 1-t^{m_1}, (1-t^{m_1})(1-t^{m_1+m_2})(1-t^{m_2})$ 
for $n=1,2,3$, respectively, happens to coincide with the denominator of the Weyl character formula for $sl_{n}$ 
under the formal identification of $m_i$ with the $i$th simple root $(1 \le i \le n-1)$.
It follows from \eqref{zpr} and \eqref{ze} that $Z_{l,\mathbf m}(x_1,x_2, \ldots,  x_L;t)$ is a
{\em homogeneous} symmetric polynomial in $x_1, \dots, x_L$ of degree $\sum_{i=1}^n im_i$.
Theorem \ref{th:z} together with \eqref{ze},  implies the factorization
\begin{align}\label{zfac}
Z_{1, \bf m}(x_1, \dots, x_L;t) = \Delta_{\mathbf m}(t)^{-1}\prod_{i=1}^{n} 
Z_{l, (m_0 + \cdots + m_{i-1}, m_{i} + \cdots + m_n)}(x_1, \dots, x_L;t).
\end{align}
The RHS consists of the partition functions for the single-species case $n=1$,
which are readily obtained from Example \ref{ex:sl2} and the definition \eqref{zdef}.

\begin{example}\label{ex:fac}
\begin{align*}
Z_{2,(1,1,2)}(x_1,x_2;t) &= \frac{1}{1-t}Z_{2,(1,3)}(x_1,x_2;t) Z_{2,(2,2)}(x_1,x_2;t),
\\
Z_{2,(1,3)}(x_1,x_2;t) &=t(1+t)x_1x_2(x_1+x_2),
\\
Z_{2,(2,2)}(x_1,x_2;t)&=t(x_1^2+x_2^2) + (1+t)^2x_1x_2,
\\
Z_{2,(1,2,3)}(x_1,x_2,x_3;t) &= 
\frac{1}{1-t^2}Z_{2,(3,3)}(x_1,x_2,x_3;t) Z_{2,(1,5)}(x_1,x_2,x_3;t),
\\
Z_{2,(3,3)}(x_1,x_2,x_3;t) &= (1+t)\bigl((1+t^2)x_1x_2x_3+t(x_1+x_2)(x_1+x_3)(x_2+x_3)\bigr),
\\
Z_{2,(1,5)}(x_1,x_2,x_3;t) &= t^2(1+t)x_1x_2x_3(x_1x_2+x_1x_3+x_2x_3).
\end{align*}
\end{example}
 
Given a local state $\bs=(\sigma_0,\ldots, \sigma_n) \in \BB_l$ in the multiplicity representation \eqref{vl},
we write $\widehat{\bs} = (\sigma_n,\ldots, \sigma_0)$, which corresponds to the conjugation of particle species
$a \leftrightarrow n-a$.
Theorem~\ref{th:z} together with \eqref{ze}--\eqref{amt} implies the following symmetry of the partition function
under particle conjugation.
\begin{corollary}\label{cor:ph}
\begin{align}\label{ph}
Z_{l,{\mathbf m}}(x_1,\ldots, x_L;t) = 
(-1)^{n(n-1)/2}t^{\gamma}(x_1\cdots x_L)^{nl}
Z_{l, \widehat{\mathbf m}}(x_1^{-1},\ldots, x^{-1}_L;t^{-1}),
\end{align}
where $\widehat{\mathbf m}= (m_n,\ldots, m_0)$ for 
${\mathbf m}=(m_0,\ldots, m_n)$, and 
$\gamma = Lnl(l-1)/2-\sum_{0<j<n}j(n-j)m_j$.
\end{corollary}

This identity is readily derived from the relation
$z_1\cdots z_{N}e_k(z^{-1}_1,\ldots, z^{-1}_{N}) = e_{N-k}(z_1,\ldots, z_N)$
with $N= lL = m_0+\cdots + m_n$ and 
$(z_1,\ldots, z_N)= (x_1,\ldots, t^{l-1}x_1, \ldots, x_L, \ldots, t^{l-1}x_L)$.

Conjecturally, a stronger symmetry holds that would imply
Corollary~\ref{cor:ph}.

\begin{conjecture}\label{conj:pp}
For $(\bs_1,\ldots, \bs_L) \in \mathcal{S}({\mathbf m})$ in \eqref{Sm},
the following equality holds:
\begin{equation}\label{pmp}
\mathbb{P}_l\!\left(\genfrac{}{}{0pt}{}{x_1,\ldots,x_L}{\bs_1,\ldots,\bs_L};t\right)
=(-1)^{n(n-1)/2}t^{\gamma}(x_1\cdots x_L)^{nl}\,
\mathbb{P}_l\!\left(\genfrac{}{}{0pt}{}{x_1^{-1},\ldots,x_L^{-1}}
{\widehat{\bs}_1,\ldots,\widehat{\bs}_L}; t^{-1}\right),
\end{equation}
where the exponent $\gamma$ 
is the same as the one in Corollary \ref{cor:ph}.
\end{conjecture}

This conjecture is curious, since, for instance in Example~\ref{ex:akos},
the CTM operators $A_2(z)$ and $A_1(z)$, which correspond to the conjugate pair particles
for $(n,l)=(3,1)$, contain\emph{different} numbers of monomials in the $t$-oscillator generators.

Let 
\begin{equation}\label{pid}
\pi_l\!\left(\genfrac{}{}{0pt}{}{x_1,\ldots,x_L}{\bs_1,\ldots,\bs_L};t\right)
=Z_{l,{\mathbf m}}(x_1,\ldots, x_L;t)^{-1}\,
\mathbb{P}_l\!\left(\genfrac{}{}{0pt}{}{x_1,\ldots,x_L}{\bs_1,\ldots,\bs_L};t\right)
\end{equation}
be the {\em normalized} stationary probability,
Corollary \ref{cor:ph} and Conjecture \ref{conj:pp} lead to 
the invariance under the particle conjugation:
\begin{equation}
\pi_l\!\left(\genfrac{}{}{0pt}{}{x_1,\ldots,x_L}{\bs_1,\ldots,\bs_L};t\right)
=
\pi_l\!\left(\genfrac{}{}{0pt}{}{x_1^{-1},\ldots,x_L^{-1}}
{\widehat{\bs}_1,\ldots,\widehat{\bs}_L}; t^{-1}\right).
\end{equation}

\begin{example}\label{ex:pp0}
Consider the case $n=L=l=2$.
Local states are expressed in the tableau representation $\TT_2$ defined in \eqref{tl},
and are written without parentheses.
For ${\mathbf m}=(2,1,1)$, we have 
\begin{align*}
\mathbb{P}_2\!\left(\genfrac{}{}{0pt}{}{x_1,x_2}{00,12};t\right)
&= \frac{t(1+t)}{1-t}x_2^2(x_1+x_2),
\quad 
\mathbb{P}_2\!\left(\genfrac{}{}{0pt}{}{x_1,x_2}{01, 02};t\right)
=\frac{1+t}{1-t}x_1x_2\bigl(tx_1+(1+t+t^2)x_2\bigr),
\\
Z_{2,(2,1,1)}(x_1,x_2;t)&= \frac{1+t}{1-t}(x_1+x_2)
\bigl((1+t^2)x_1x_2+t(x_1+x_2)^2\bigr).
\end{align*}
For $\widehat{\mathbf m} = (1,1,2)$, we have
\begin{align*}
\mathbb{P}_2\!\left(\genfrac{}{}{0pt}{}{x_1,x_2}{22,01};t\right)
&= \frac{t^2(1+t)}{1-t}x_1^3x_2(x_1+x_2),
\quad 
\mathbb{P}_2\!\left(\genfrac{}{}{0pt}{}{x_1,x_2}{12, 02};t\right)
=\frac{t(1+t)}{1-t}x_1^2x_2^2\bigl(tx_2+(1+t+t^2)x_1\bigr),
\\
Z_{2,(1,1,2)}(x_1,x_2;t)&= \frac{t(1+t)}{1-t}x_1x_2(x_1+x_2)
\bigl((1+t^2)x_1x_2+t(x_1+x_2)^2\bigr).
\end{align*}
The unnormalized probabilities $\mathbb{P}_2$ for the other states in 
$\mathcal{S}({\mathbf m})$ and $\mathcal{S}(\widehat{\mathbf m})$
are obtained by \eqref{cts}.
They satisfy the following relations in agreement with 
Corollary \ref{cor:ph} and Conjecture \ref{conj:pp}:
\begin{align*}
&\mathbb{P}_2\!\left(\genfrac{}{}{0pt}{}{x_1,x_2}{00,12};t\right) = -t^3(x_1x_2)^4\,
\mathbb{P}_2\!\left(\genfrac{}{}{0pt}{}{x_1^{-1},x_2^{-1}}{22,\;\;01\;};t^{-1}\right),
\quad
\mathbb{P}_2\!\left(\genfrac{}{}{0pt}{}{x_1,x_2}{01,02};t\right) = -t^3(x_1x_2)^4\,
\mathbb{P}_2\!\left(\genfrac{}{}{0pt}{}{x_1^{-1},x_2^{-1}}{12,\;\; 02\;};t^{-1}\right),
\\
&Z_{2,(2,1,1)}(x_1,x_2;t) = -t^3(x_1x_2)^4
Z_{2,(1,1,2)}(x_1^{-1},x_2^{-1};t^{-1}).
\end{align*}
\end{example}

\begin{example}\label{ex:pp}
Consider the case $n=L=3$ and $l=2$.
Local states are expressed in the tableau representation $\TT_2$ defined in \eqref{tl},
and are written without parentheses.
\begin{align*}
\mathbb{P}_2\!\left(\genfrac{}{}{0pt}{}{x_1,x_2,x_3}{01,22,33};t\right)
&= \frac{t^5(1-t^2)x_1(x_2x_3)^3e_2}{(1-t)(1-t^3)^2(1-t^4)}
\left(
t^4x_1^2+\frac{t(1-t^2)(1-t^4)}{(1-t)^2} x_1x_3 
+\frac{(1-t^3)(1-t^4)}{(1-t)(1-t^2)}x_3^2\right),
\\
\mathbb{P}_2\!\left(\genfrac{}{}{0pt}{}{x_1,x_2,x_3}{23,11,00};t\right)
 &= \frac{t^2(1-t^2)(x_1x_2)^2e_1}{(1-t)(1-t^3)^2(1-t^4)}
\left(
\frac{t(1-t^3)(1-t^4)}{(1-t)(1-t^2)}x_1^2+
\frac{(1-t^2)(1-t^4)}{(1-t)^2}x_1x_3+ tx_3^2\right),
\\
Z_{2,(1,1,2,2)}(x_1,x_2,x_3;t) &=
\frac{t^3e_2e_3}{(1-t)^2(1-t^3)}
\bigl(te_1^2+(1+t^2)e_2\bigr)\bigl((1+t^2)e_1e_3+t e_2^2\bigr),
\\
Z_{2,(2,2,1,1)}(x_1,x_2,x_3;t) &=
\frac{te_1}{(1-t)^2(1-t^3)}
\bigl(te_1^2+(1+t^2)e_2\bigr)\bigl((1+t^2)e_1e_3+t e_2^2\bigr),
\end{align*}
where $e_i = e_i(x_1,x_2,x_3)$.
They satisfy the following relations in agreement with 
Corollary \ref{cor:ph} and Conjecture \ref{conj:pp} with 
${\mathbf m}=(1,1,2,2)$:
\begin{align*}
\mathbb{P}_2\!\left(\genfrac{}{}{0pt}{}{x_1,x_2,x_3}{01,22,33};t\right)
&= -t^3(x_1x_2x_3)^6 \,
\mathbb{P}_2\!\left(\genfrac{}{}{0pt}{}{\, x_1^{-1},\, x_2^{-1},\, x_3^{-1}}{23,\;\,\;11,\;\;00};t^{-1}\right),
\\
Z_{2,(1,1,2,2)}(x_1,x_2,x_3;t) 
&= -t^3(x_1x_2x_3)^6 \, Z_{2,(2,2,1,1)}(x^{-1}_1,x^{-1}_2,x^{-1}_3;t^{-1}).
\end{align*}
\end{example}

Mathematically, it is natural to formally generalize the matrix product formula
$\mathrm{Tr}\left(A_{\bs_1}(x_1)\cdots A_{\bs_L}(x_L)\right)$
and the resulting partition function 
by introducing ``fugacity'' parameters $q_1,\ldots,q_N$ as
\begin{equation}
\mathrm{Tr}\left(q_1^{{\mathbf h}_1}\cdots q_N^{{\mathbf h}_N}
A_{\bs_1}(x_1)\cdots A_{\bs_L}(x_L)\right),
\end{equation}
where $N=n(n-1)/2$ is the number of independent $t$-oscillators appearing in the
CTM operators \eqref{ctm1} and \eqref{Abi}, and ${\mathbf h}_i$ denotes the number operator
acting on the $i$th copy of the Fock space \eqref{tom} by ${\mathbf h}_i|m\rangle = m|m\rangle$.

Studies related to the general capacity~$l$ case from the viewpoint of algebraic
combinatorics have been carried out in~\cite{GW20,M25}.
In those works it is observed that
symmetric fusion of $R$-matrices, or of the corresponding vertex models, manifests
itself as a plethystic substitution at the level of partition functions.
Our result \eqref{zpr} provides a concrete realization of this phenomenon
based on CTM type matrix product operators.

\appendix

\section{Stochastic $R$-matrix $\s_{k,l}(z)$ for symmetric tensor representation}\label{app:rs}

We depict the elements $\s(z)^{{\bf e}_a,{\bf e}_b}_{{\bf e}_i,{\bf e}_j}$ of the 
stochastic $R$-matrix $\s(z)=\s_{1,1}(z)$ in \eqref{ns} as follows:

\begin{equation}
\centering
\begin{tikzpicture}[baseline=(current bounding box.center),
                    scale=.42,>=Stealth,thick,
                    line cap=round,line join=round]
  \draw[->] (-1.2,0) node[anchor=east]{$i$} -- ( 1.2,0)
            node[anchor=west]{$a$};
  \draw[->] ( 0,-1.2) node[anchor=north]{$j$} -- (0, 1.2)
            node[anchor=south]{$b$};

  \draw (0,0) ++(180:0.45) arc (180:260:0.45);
  \node at ($(0,0)+(235:0.95)$) {$z$};

  \node[overlay,anchor=east] at (-2.5,0)
    {$\s(z)^{\mathbf e_a,\mathbf e_b}_{\mathbf e_i,\mathbf e_j}\, =$};
\end{tikzpicture}
\label{sdg}
\end{equation}

Applying the fusion procedure \cite{KRS81} along the symmetric tensor 
representations, one can construct a family of stochastic 
$R$-matrices $\s_{k,l}(z)$ acting on $V_k \otimes V_l$
for arbitrary $k,l \in \Z_{\ge 1}$ ($V_l$ is defined in \eqref{vl}):
\begin{subequations}
\begin{align}
\s_{k,l}(z): &\; V_k \otimes V_l \longrightarrow  V_k \otimes V_l,
\\
\s_{k,l}(z)(v_{\bf i} \otimes v_{\bf j}) 
&= \sum_{{\bf a} \in \BB_k,\; {\bf b} \in \BB_l}
\s_{k,l}(z)^{{\bf a}, {\bf b}}_{{\bf i}, {\bf j}}\,
v_{\bf a} \otimes v_{\bf b}
\qquad ({\bf i} \in \BB_k,\; {\bf j} \in \BB_l).
\end{align}
\end{subequations}
This construction is well known; hence we only present the result in a 
form adapted to our discussion in Section~\ref{sec:zf}.
We illustrate the construction of the element 
$\s_{k,l}(z)^{{\bf a}, {\bf b}}_{{\bf i}, {\bf j}}$
along the example of $(k,l)=(3,2)$ with 
\begin{align*}
{\bf a} &= (a_1,a_2,a_3) \in \TT_3, 
\qquad \;{\bf i} = (i_1,i_2,i_3)  \in \TT_3,\\
{\bf b} &= (b_1,b_2) \in \TT_2, 
\qquad \qquad {\bf j} = (j_1,j_2) \in \TT_2,
\end{align*}
where we employ the tableau representations in \eqref{tl}.

The $R$-matrix for the  symmetric fusion  $\s_{k,l}(z)^{{\bf a}, {\bf b}}_{{\bf i}, {\bf j}}$ 
is given by the sum of products of the basic $(k,l)=(1,1)$ weights 
\eqref{ns}, \eqref{sdg} as follows:
\begin{equation}\label{f32}
\s_{3,2}(z)^{{\bf a}, {\bf b}}_{{\bf i}, {\bf j}} =
\sum_{\substack{(a'_1,a'_2,a'_3) \,\in \,\mathcal{C}({\bf a}) \\ (b'_1,b'_2)\, \in \,\mathcal{C}({\bf b})}}
\;\sum_{\bullet \in \{0,\ldots, n\}}
\begin{tikzpicture}[baseline={([yshift=-1.2ex]current bounding box.center)},%
  x=0.46cm,y=0.46cm,>=Stealth,thick,line cap=round,line join=round]

  \draw[->] (-0.8,0) -- (4.8,0);
  \draw[->] (-0.8,3) -- (4.8,3);
  \draw[->] (-0.8,6) -- (4.8,6);

  \draw[->] (0.5,-1.3) -- (0.5,7.3);
  \draw[->] (3.5,-1.3) -- (3.5,7.3);

  \def\r{0.35}
  \foreach \X in {0.5,3.5}{
    \foreach \Y in {0,3,6}{
      \draw[thin] (\X,\Y) ++(180:\r) arc (180:270:\r);
    }
  }

  \node[anchor=south east] at (0.40,-1.10) {\scriptsize $z/t$};
  \node[anchor=south east] at (0.40, 2.00) {\scriptsize $z$};
  \node[anchor=south east] at (0.40, 5.00) {\scriptsize $zt$};
  \node[anchor=south east] at (3.40,-1.00) {\scriptsize $z$};
  \node[anchor=south east] at (3.40, 2.00) {\scriptsize $zt$};
  \node[anchor=south east] at (3.40, 5.00) {\scriptsize $zt^2$};

  \node[scale=1.1] at (2.0,0)   {$\bullet$};
  \node[scale=1.1] at (2.0,3)   {$\bullet$};
  \node[scale=1.1] at (2.0,6)   {$\bullet$};
  \node[scale=1.1] at (0.5,1.5) {$\bullet$};
  \node[scale=1.1] at (3.5,1.5) {$\bullet$};
  \node[scale=1.1] at (0.5,4.5) {$\bullet$};
  \node[scale=1.1] at (3.5,4.5) {$\bullet$};

  \node[anchor=west]  at (4.8,6) {\small $a'_1$};
  \node[anchor=west]  at (4.8,3) {\small $a'_2$};
  \node[anchor=west]  at (4.8,0) {\small $a'_3$};
  \node[anchor=south] at (0.5,7.3) {\small $b'_1$};
  \node[anchor=south] at (3.5,7.3) {\small $b'_2$};

  \node[anchor=east] at (-0.8,6) {\small $i_1$};
  \node[anchor=east] at (-0.8,3) {\small $i_2$};
  \node[anchor=east] at (-0.8,0) {\small $i_3$};

  \node[anchor=north] at (0.5,-1.3) {\small $j_1$};
  \node[anchor=north] at (3.5,-1.3) {\small $j_2$};
\end{tikzpicture}
\end{equation}
Here  $\mathcal{C}({\bf a})$, for example, 
is the set of distinct permutations of the tableau ${\bf a}$ as defined in \eqref{Cex}.
The sums with respect to $\bullet$ are taken independently 
over $\TT_1 \simeq \{0,\ldots,n\}$ in~\eqref{tl}. 
In general, the spectral parameter takes the value $z$ 
at the bottom right vertex and is multiplied by $t$ when going upward, 
and by $t^{-1}$ when going to the left.

In addition to the Yang-Baxter relation, the stochastic $R$-matrix satisfies the sum-to-unity property:
\begin{align}\label{stu2}
\sum_{{\bf a} \in \TT_k, {\bf b} \in \TT_l} \s_{k,l}(z)^{{\bf a}, {\bf b}}_{{\bf i}, {\bf j}} =1
\qquad ({\bf i} \in \TT_k, {\bf j} \in \TT_l).
\end{align}

Factorized formulas are known for $k=1$ and $l=1$. 
Explicitly, one has
\begin{align}
\s_{1,l}(t^{l-1}z)^{{\bf e}_a,  {\bf b}}_{{\bf e}_i, {\bf j}}
&=\frac{(1-t^{j_a}z^{\delta_{a,i}})z^{\theta(i<a)}t^{j_0+\cdots+j_{a-1}}}{1-t^lz},
\label{ss1l}
\\
\s_{k,1}(z)^{{\bf a}, {\bf e}_b}_{{\bf i}, {\bf e}_j}
&=\frac{(1-t^{i_b}z^{\delta_{b,j}})z^{\theta(j>b)}t^{i_{b+1}+\cdots + i_n}}{1-t^k z}.
\label{ssk1}
\end{align}
In the general case, an approach based on the three-dimensional integrability and the tetrahedron equation 
leads to the formula 
$
\s_{k,l}(z)^{{\bf a}, {\bf b}}_{{\bf i}, {\bf j}} 
= \mathcal{A}(q^{\,l-k}/z)^{{\bf a}, {\bf b}}_{{\bf i}, {\bf j}}
\big|_{q=t^{1/2}},
$
where the function $\mathcal{A}(z)^{{\bf a}, {\bf b}}_{{\bf i}, {\bf j}}$ for ${\bf a}, {\bf i} \in \BB_k$ and 
${\bf b}, {\bf j} \in \BB_l$ 
is defined by~\cite[eqs.~(13.47)--(13.50), (13.127)]{K22} with $(l,m)$ there replaced with $(k,l)$.

From the natural identification $V^1 \simeq V_1$ and $\BB^1 = \BB_1$ in
\eqref{vk} and \eqref{vl}, the $R$-matrices
$S^{k}_{\;\,l}(z) \in \mathrm{End}(V^k \otimes V_l)$ introduced in
Section~\ref{sec:s} and
$\s_{k,l}(z) \in \mathrm{End}(V_k \otimes V_l)$ are essentially identical
in the special case $k = 1$.  
Indeed, they are related by
\begin{align}\label{sss}
S^{1}_{\;\,l}\bigl(z\bigr)
 = \s_{1,l}\!\left(t^{\,l-1} z\right)\,(1 - t^{\,l} z),
\end{align}
which coincides with the numerator of \eqref{ss1l}.
For $l = 1$, this follows by comparing Example~\ref{ex:s11} with \eqref{ns}.
 
\begin{remark}\label{re:indep}
As a consequence of the symmetric fusion, 
the sum~\eqref{f32} is actually invariant under the replacement of 
$(i_1,i_2,i_3)$ with any $(i'_1,i'_2,i'_3) \in \mathcal{C}(\mathbf{i})$, 
and likewise of $(j_1,j_2)$ with any $(j'_1,j'_2) \in \mathcal{C}(\mathbf{j})$.
A similar invariance is valid for general $k$ and $l$.
\end{remark}

\section*{Acknowledgments}
A.~A. was partially supported by SERB Core grant CRG/2021/001592 
as well as the DST FIST program~-~2021 [TPN - 700661]. 
A.~K. thanks the organizers of the {\em Program on Classical, Quantum and
Probabilistic Integrable Systems} (March~24--May~24,~2025) at the Center of
Mathematical Sciences and Applications, Harvard University, and the 
{\em Symposium on Solvable Lattice Models \& Interacting Particle Systems}
(August~24--30,~2025) at Schloss Elmau, Germany, for their kind invitation
and warm hospitality, where part of this work was carried out.
A.~K. also thanks Amol Aggarwal, Sylvie Corteel, Leonid Petrov, Michael Wheeler,
and Lauren Williams for their interest and useful communications.
A.~K. is supported by JSPS KAKENHI No.~24K06882.


\begin{thebibliography}{99}

\bibitem[ANP23]{ANP23}
A.~Aggarwal, M.~Nicoletti, L.~Petrov,
Colored interacting particle systems on the ring: Stationary measures from Yang-Baxter equation.
Compos. Math. {\bf 161} no. 8 1855--1922 (2025).

\bibitem[AK25]{AK25}
A.~Ayyer, A.~Kuniba,
Multispecies inhomogeneous $t$-pushTASEP from antisymmetric fusion,
Electron. J. Probab. {\bf 30} article no. 190, 1--28 (2025).

\bibitem[AM23]{AM23}
A.~Ayyer, J.~B. Martin, 
The inhomogeneous multispecies PushTASEP: Dynamics and symmetry.
arXiv:2310.09740.

\bibitem[AMW24]{AMW24}
A.~Ayyer, J.~B. Martin, L.~Williams,
The inhomogeneous $t$-PushTASEP and Macdonald polynomials.
arXiv:2403.10485.

\bibitem[Bax83]{Bax83}
R.~J.~Baxter,
{\em Exactly solved models in statistical mechanics},
Academic Press, London (1983).

\bibitem[BW22]{BW22}
A.~Borodin, M.~Wheeler,
{\em Colored stochastic vertex models and their spectral theory}.
 Ast\'erisque, (437): ix+225  (2022).

\bibitem[CDW15]{CDW15}
L.~Cantini, J. ~de~Gier, M.~Wheeler,
Matrix product formula for Macdonald polynomials.
J. Phys. A: Math. Theor. {\bf 48}: 384001, 25 (2015).

\bibitem[CMW22]{CMW22}
S.~Corteel, O.~Mandelshtam, L.~Williams,
From multiline queues to Macdonald polynomials via the exclusion process.
Amer. J. Math., 144(2):395--436  (2022).

\bibitem[CP13]{CP13}
I.~Corwin, L.~Petrov,
The $q$-PushASEP: A new integrable model for traffic in 1+1 dimension.
arXiv:1308.3124.

\bibitem[FH15]{FH15}
E.~Frenkel, D.~Hernandez,
Baxter's relations and spectra of quantum integrable models,
Duke. Math. J. {\bf 164} 2407--2460 (2015).

\bibitem[GW20]{GW20}
A.~Garbali, M.~Wheeler,
Modified Macdonald polynomials and integrability,
Commun. Math. Phys. {\bf 374} 1809--1876 (2020).

\bibitem[IKT12]{IKT12} 
R.~Inoue, A.~Kuniba, T.~Takagi, 
Integrable structure of box-ball system:
crystal, Bethe ansatz, ultradiscretization and tropical geometry.
J. Phys. A: Math. Theor. {\bf 45}  073001 (2012).

\bibitem[KRS81]{KRS81}
P.~P.~Kulish, N.~Yu.~Reshetikhin, E.~K.~Sklyanin,
Yang-Baxter equation and representation theory.
Lett.  Math. Phys. {\bf 5} 393--403 (1981). 

\bibitem[K22]{K22}
A.~Kuniba,
{\em Quantum groups in three-dimensional integrability}, 
Springer, Singapore  (2022).
  
\bibitem[KMMO16]{KMMO16}
A.~Kuniba, V.~V. Mangazeev, S.~Maruyama, M.~Okado,
Stochastic $R$ matrix for $U_q(A_n^{(1)})$.
Nucl. Phys. B{\bf  913}:248--277 (2016).

\bibitem[KNS11]{KNS11}
A.~Kuniba, T.~Nakanishi, J.~Suzuki,
T-systems and Y-systems in integrable systems,
J. Phys. A: Math. Theor. {\bf 44} 103001  (2011).

\bibitem[KOS24]{KOS24}
A.~Kuniba,  M.~Okado, T.~Scrimshaw,
Strange five vertex model and multispecies ASEP in a ring,
arXiv:2408.12092.

\bibitem[KS95]{KS95}
A.~Kuniba, J.~Suzuki,
Analytic Bethe ansatz for fundamental representations of Yangians,
Commun. Math. Phys. {\bf 173} 225--264 (1995).

\bibitem[Mac95]{Mac95}
I.~G.~ Macdonald, 
{\em Symmetric functions and Hall polynomials}, second ed.
Oxford Mathematical Monographs. The Clarendon Press, Oxford University Press, New York
(1995).

\bibitem[M25]{M25}
O.~Mandelshtam,
New formulas for Macdonald polynomials via the multispecies exclusion and zero range processes,
Contemp. Math. 
{\em Macdonald Theory and Beyond: Combinatorics, Geometry, and Integrable Systems}
{\bf  815}, AMS, Providence, RI  (2025).

\bibitem[P19]{P19} L.~Petrov,
PushTASEP in inhomogeneous space.
{\em Electron. J. Probab.}, {\bf 25}: Paper No. 114, 25 (2020).

\bibitem[PEM09]{PEM09}
S.~Prolhac, M.~R. Evans, K.~Mallick,
The matrix product solution of the multispecies partially asymmetric  exclusion process.
J. Phys. A: Math. Theor. {\bf 42}: 165004, 25 (2009).

\bibitem[R87]{R87}
N.~Yu.~Reshetikhin, 
The spectrum of the transfer matrices connected with Kac-Moody algebras, 
Lett. Math. Phys. {\bf 14} 235--246 (1987).

\bibitem[STF80]{STF80}
E.~K.~Sklyanin, L.~A.~Takhatajan, L.~D.~Faddeev,
 Quantum inverse problem method~I,
 Theor.~Math.~Phys. {\bf 40}  688--706 (1980).
 
 \end{thebibliography}
\end{document}